\documentclass{amsproc}

 \usepackage{hyperref}

 \newtheorem{theorem}{Theorem}[section]
 \newtheorem{corollary}[theorem]{Corollary}
 \newtheorem{lemma}[theorem]{Lemma}
 
 \newtheorem{proposition}[theorem]{Proposition}
 \theoremstyle{definition}
 \newtheorem{definition}[theorem]{Definition}
 \theoremstyle{remark}
 \newtheorem{remark}[theorem]{Remark}
 \newtheorem*{example}{Example}
 \newtheorem*{examples}{Examples}

\numberwithin{equation}{section}

\usepackage{color}
\usepackage{graphics}
\usepackage{graphicx}
\usepackage{epsfig}
\usepackage{amssymb}
\usepackage{verbatim}
\usepackage{graphics,amsmath,amsfonts,amssymb} %% add this and next lines if pictures should be in esp format
\usepackage{epsfig,verbatim}
\input xy
\xyoption{all}
\def\sympow{{\setbox0\hbox{$\bigcirc$}\setbox1\hbox to\wd0{\hss$s$\hss}%
\wd1 0pt\box1\box0}}%symmetric power
\begin{document}

\title[Galoisian Approach to integrability of
Schr\"odinger Equation]{Galoisian Approach to integrability of
Schr\"odinger Equation}

\author[P. Acosta-Hum\'anez]{Primitivo B. Acosta-Hum\'anez}
\address[P. Acosta-Hum\'anez]{Instituto de Matem\'aticas y sus Aplicaciones (IMA), Universidad Sergio Arboleda, Bogot\'a \& Santa Marta, Colombia}
%\curraddr{Department of Mathematics and Statistics,
%Case Western Reserve University, Cleveland, Ohio 43403}
\email{primi@ima.usergioarboleda.edu.co}

\author[J. Morales-Ruiz]{Juan J. Morales-Ruiz}
\address[J. Morales-Ruiz]{Escuela de Caminos, Canales y Puertos, Universidad Polit\'ecnica de Madrid, Madrid, Spain}
\email{juan.morales-ruiz@upm.es}

%    Information for third author
\author[J.-A. Weil]{Jacques-Arthur Weil}
\address[J.-A. Weil]{Institute XLIM (UMR CNRS n¡6172), DMI, Universit\'e de Limoges, Limoges, France}
\email{jacques-arthur.weil@unilim.fr}

%    General info
%\date{January 1, 1994 and, in revised form, June 22, 1994.}

\dedicatory{Dedicated to Jerry Kovacic, in memoriam.}


\begin{abstract}
In this paper, we examine the non-relativistic stationary Schr\"odinger equation from a differential Galois-theoretic perspective. The main algorithmic tools are pullbacks of second order ordinary linear differential operators, so as to achieve rational function coefficients (``algebrization''), and Kovacic's algorithm for solving the resulting equations.
In particular, we use
this Galoisian approach to analyze Darboux transformations, Crum iterations
and supersymmetric quantum mechanics. We obtain the
ground states, eigenvalues, eigenfunctions, eigenstates and differential
Galois groups of a large class of Schr\"odinger equations, e.g.
those with exactly solvable and shape invariant potentials
(the terms are defined within). Finally, we introduce a method
for determining when exact solvability is possible.
\end{abstract}

\maketitle \tableofcontents
\section{Introduction} 
In classical mechanics, the
Hamiltonian corresponding to the energy (kinetic plus potential)
is given by
$$H={\|\overrightarrow{p}\|^2\over 2m}+U(\overrightarrow{x}),\quad \overrightarrow{p}=(p_1,\ldots,p_n),\quad \overrightarrow{x}=(x_1,\ldots,x_n),$$
and the corresponding (Hamiltonian) equations are 
	$\dot{\overrightarrow{x}}=\frac{\partial H}{\partial \overrightarrow{p}}$ and  
	$\dot{\overrightarrow{p}}=-\frac{\partial H}{\partial \overrightarrow{x}}$.
In quantum mechanics, the momentum $\overrightarrow{p}$ is
replaced by $\overrightarrow{p}=-\imath\hbar\nabla$, the Hamiltonian
operator by the Schr\"odinger (non-relativistic, stationary)
operator 
$$H=-{\hbar^2\over 2m}\nabla^2+U(\overrightarrow{x})$$
and Hamilton's equations by the Schr\"odinger equation $H\Psi=E\Psi$.
\\
The eigenfunction $\Psi$
is known as the {\it{wave function}}, the eigenvalue $E$ is the {\em{energy
level}}, $U(\overrightarrow{x})$ is the {\em{potential or
potential energy}}. Solutions $(\Psi,E)$ of the Schr\"odinger equation
are the {\em{states}} of the particle. 
In specific applications, the potentials will satisfy some conditions depending of
 physical restrictions  such as barrier, scattering, etc., see
 \cite{cokasu,gapa,lali,masa,sc}.
\\
The Schr\"odinger operator is known to be a \textit{self-adjoint operator}:
$H^\dagger=H$ when viewed within the context of a  suitable complex and
separable Hilbert space. As a consequence, the spectrum $\mathrm{Spec}(H)$ of $H$
must be real, bounded below, and be the disjoint union of the \textit{point
spectrum} $\mathrm{Spec}_p(H)$ and the continuous spectrum
$\mathrm{Spec}_c(H)$ (see for example \cite{bana,pru,te}).
Wave functions belonging to the continuous spectrum may arise from barrier potentials
and from periodic boundary conditions. The
\textit{transmission and reflection coefficients} are related with
the barrier potentials, \cite{gapa}.

In this paper, we
consider the one-dimensional Schr\"odinger equation with the normalization $\hbar=2m=1$, i.e.
\begin{displaymath}\label{hams}
H\Psi=E\Psi, \quad H=-\partial_z^2+V(z),
\end{displaymath}
 where we may have $z=x$ (cartesian coordinate) or $z=r$ (radial coordinate). 
 We let $\Psi_n$ denote the wave function corresponding to an energy level $E=E_n$.
 \\
 
 Our purpose is to investigate solutions and their classification from the viewpoint of differential Galois theory.
 We should point out that this viewpoint has already been explored in nearby theories, for example by
 Spiridonov in \cite{sp} 
%\\  \JA{explain this + add refs Couch and Holder from Ref, Etingof, etc.}
%Spiridonov also pointed out   the relevance  of the differential Galois theory in quantum mechanics \cite{sp}.  \\
For sake of completeness, we now recall a few definitions related to quantum mechanics, and several useful
results on the one-dimensional Schr\"odinger equation.
%\\ \JA{should make the structure clearer here ! }

\textbf{Bound States}\label{defbs} A solution $\Psi_n$ is called
a \textit{bound state}
  when its norm is finite and the corresponding eigenvalue $E_{n}$ belongs to the point spectrum of $H$, i.e.
\begin{displaymath}\label{condqm}
E_n\in \mathrm{Spec}_p(H),\quad \int|\Psi_n(x)|^2dx<\infty, \quad
n\in\mathbb{Z}_+.
\end{displaymath}

\textbf{Ground State and Excited States.}  Wave functions $\Psi_0,\Psi_1,\ldots,\Psi_n,\ldots$
of the bound states are numbered so that the corresponding energies are ordered $E_0<E_1<\cdots<E_n<\ldots$.
The state $(\Psi_0,E_{0})$ with minimal energy
 is called the \textit{ground state} and the remaining states $\Psi_1,\ldots,\Psi_n,\ldots$ are
 called the \textit{excited states}.\\

\textbf{Scattering States.}  A solution $\Psi$ corresponding to
the level energy $E$ is called a \textit{scattering state}
  when $E$ belongs to the continuous spectrum of $H$ and its norm is  infinite.
\medskip

To model a particle moving in one dimension, we use the classical one dimensional Schr\"odinger equation with
 Cartesian coordinate $x$. 
 We could also consider motion in three dimension by reducing the associated Schr\"odinger equation to
 the one-dimensional 
\textit{radial equation} (see \cite{gapa,lali, sc}).\\

 \label{introd}
\textbf{Darboux Transformation}. The Darboux transformation is a process which, starting from a second order linear differential equation such as the Schr\"odinger equation, produces a family of second order differential equations with a similar shape. This vague formulation is made clear in the following theorem, taken from from \cite{da1},
which is slightly more general than what is known today as the Darboux transformation.

\emph{For any value of a constant $m$, let $y=y_{m}$ denote a solution of the equation
\begin{displaymath}\label{orda7}\partial_x^2y+ P\partial_xy+(Q -mR)y = 0.\end{displaymath} 
Suppose, in addition, that we are given a solution $\theta$ of the equation
\begin{displaymath}\partial_x^2\theta+P\partial_x\theta+Q\theta=0.\end{displaymath} 
Then the function
\begin{displaymath}\label{orda9} u={\partial_xy-{\partial_x\theta\over \theta}y\over
\sqrt{R}},\end{displaymath} is a solution of the equation
\begin{displaymath}\label{orda10}
\partial_x^2u+P\partial_xu+
\left(\theta\sqrt{R}\partial_x\left({P\over
\theta\sqrt{R}}\right)-\theta\sqrt{R}\partial_x^2\left({1\over
\theta\sqrt{R}}\right)
-mR\right)u=0,\end{displaymath} when $m\neq
0$.}\\

This \textit{general Darboux transformation} $(y,\theta)\mapsto u$ preserves the shape of the differential equation and maps $Q$ to $%\left(
\theta\sqrt{R}\partial_x\left({P\over
\theta\sqrt{R}}\right)-\theta\sqrt{R}\partial_x^2\left({1\over\theta\sqrt{R}}\right)$.
In  \cite{da1,da2},  Darboux considered the
particular case $R=1$ and $P=0$, which is known today as the
\textit{Darboux transformation}. It is the following corollary of the above:
\\

\emph{For any value of a constant $m$, let $y=y_{m}$ denote a solution of the equation
\begin{displaymath}\label{orda11}
\partial_x^2y= (f(x) + m)y. \end{displaymath}
Suppose, in addition, that we are given a solution $\theta$ of the equation
 $\partial_x^2\theta= (f(x)+m_1)\theta$ for some fixed value $m_{1}$ (e.g $m_{1}=0$).
Then the function
$$u = \partial_x y-{\partial_x\theta\over\theta}y$$ is a solution of the equation
\begin{displaymath}\label{orda12} \partial_x^2u= \left(\theta\partial_x^2\left(1\over \theta\right)-m_1+m \right)u, \end{displaymath}
for $m\neq m_1$. Furthermore, $$\theta\partial_x^2\left(1\over
\theta\right)-m_1=f(x)-2\partial_x\left({\partial_x\theta\over
\theta}\right)=2\left({\partial_x\theta\over
\theta}\right)^2-f(x)-2m_1.$$ }
Some examples can be found as exercises in the book of  Ince 
\cite[p. 132]{in}.
Darboux's transformation is related with with Delsarte's
transformation operators (named transmutations in  \cite{de}), which
today are called \textit{intertwiners} or \textit{intertwining operators}.
A good short survey about Darboux transformations can be found in \cite{ro}.
\medskip

Crum, inspired by the works of Liouville \cite{lio,lio2} obtained
an iterative generalization of Darboux's result appropriate for Sturm-Liouville systems:
he proved that the Sturm-Liouville conditions are preserved under Darboux
transformations, see \cite{cr}. Crum's result is presented in
the following theorem. Recall that the Wronskian determinant $Wr$ of
$k$ functions $f_1, f_2, \ldots, f_k$ is 
$$Wr(f_1,\ldots,f_k)=\det A,\quad A_{ij}=\partial_x^{i-1}f_j,\quad i,j=1,2,\ldots, k.$$

\textbf{Crum's Theorem.} Let  $\Psi_1,  \Psi_2, \ldots \Psi_n$ be
solutions of the Schr\"odinger equation $H\Psi = E \Psi$ for energy levels $E = E_1, E_2, \ldots,E_n$ respectively.
Letting $$\Psi[n] = {Wr(\Psi_1,\ldots, \Psi_n,\Psi)\over
Wr(\Psi_1, \ldots, \Psi_n)},\quad\hbox{\rm and }\; V[n] = V - 2\partial_x^2 \ln
Wr(\Psi_1, \ldots, \Psi_n),$$
we obtain the Schr\"odinger equation
$$H^{[n]}\Psi[n]=E\Psi[n],\quad\hbox{\rm where }\; H^{[n]}=-\partial_x^2 +
V[n] 
\quad\hbox{\rm for }\; E\neq E_i, 1\leq i\leq n$$ 

Crum's transformation coincides with Darboux's transformation when
$n=1$; as shown in  \cite{ornoro},  iterations of Darboux transformations coincide with the
Crum iteration. Both formalisms allow us to
obtain new families of Schr\"odinger equations while preserving the
spectrum and the Sturm-Liouville conditions, see
\cite{masa,ornoro}. We note that there are extensions of Crum's
iteration connecting the
Sturm-Liouville theory with orthogonal polynomial theory \cite{kr}.\\

\textbf{Supersymmetric quantum mechanical system.} 
Witten introduced,  in \cite[\S 6]{wi},  some models for which
\textit{dynamical breaking of supersymmetry} is possible.\\

A \textit{supersymmetric quantum mechanical system} is described by a Hamiltonian $\mathcal H$
together with several (self-adjoint) operators $Q_i$ which commute with the hamiltonian
($[Q_i,\mathcal H]=0$)  and satisfy the anticommutation relations 
\begin{displaymath}\label{wit2} \{Q_i,Q_j\}=\delta_{ij}\mathcal H,\quad\hbox{\rm where }\; \{Q_i,Q_j\}=Q_iQ_j+Q_jQ_i. \end{displaymath}
%\\\JA{Clarify this definition - }

The simplest example of a supersymmetric quantum mechanical system
arises in the case $n=2$, which is considered in this paper.
The wave function of $\mathcal H\Phi=E\Phi$ is then a
two-component \textit{Pauli spinor},
$$\Phi(x)=\begin{pmatrix}\Psi_+(x)\\\Psi_-(x)\end{pmatrix}.$$

The \textit{supercharges} $Q_i$ are defined as
\begin{displaymath}\label{wit3}
Q_\pm=\frac{\sigma_1p\pm\sigma_2W(x)}2,\quad Q_+=Q_1,\, Q_-=
Q_2,\quad p = -i\partial_x,
\end{displaymath} where the \textit{superpotential} $W$ is an arbitrary
function of $x$ and $\sigma_i$ are the usual \textit{Pauli spin
matrices}. Using the expressions above we obtain $\mathcal H$:
\begin{displaymath}\label{wit4} \mathcal H=2Q_-^2=2Q_+^2=\frac{I_2p^2+I_2W^2(x)+\sigma_3\partial_xW(x)}2,\quad\hbox{\rm with}\; I_2=\begin{pmatrix} 1&0\\ 0&1
    \end{pmatrix}.
\end{displaymath}
The \textit{supersymmetric partner Hamiltonians} $H_\pm$ are given
by
$$H_\pm=-{1\over 2}\partial_x^2+V_\pm,\quad V_\pm=\left({W\over \sqrt{2}}\right)^2\pm{1\over \sqrt{2}}\partial_x\left({W\over \sqrt{2}}\right).$$

The potentials $V_\pm$ are called \textit{supersymmetric partner
potentials} and are linked with the superpotential $W$ through a
Riccati equation. Then $\mathcal H$ can be written as
$$\mathcal H=\begin{pmatrix} H_+&0\\ 0&H_- \end{pmatrix},$$ which
leads us to the Schr\"odinger equations $H_+\Psi_+=E\Psi_+$ and
$H_-\Psi_-=E\Psi_-$. To solve $\mathcal H\Phi=E\Phi$ is equivalent to solve simultaneously
$H_+\Psi_+=E\Psi_+$ and $H_-\Psi_-=E\Psi_-$.\\

In \cite[\S 6]{wi}, V.B. Matveev and M. Salle recast Darboux transformations in this context; 
they interpret the Darboux theorem as a Darboux covariance of a Sturm-Liouville problem and
prove the following result (see also \cite[\S 5-6]{ro}). \\
\emph{ The case $n=2$ in Supersymetric Quantum Mechanics is
equivalent to a single Darboux transformation. }\\

As a generalization of the method to solve the harmonic oscillator
\cite{ge,duka}, the ladder (raising and lowering) operators are
defined as
$$A^+=-\partial_x-{\partial_x\Psi_0\over \Psi_0},\quad
A=\partial_x-{\partial_x\Psi_0\over \Psi_0},$$ which are strongly related with the supercharges $Q_\pm$ in Witten's formalism.
We have,
$$A\Psi_0=0,\quad A^+A=H_-,\quad AA^+=H_+=-\partial_x^2+V_+(x),\quad \textrm{where}$$
$$V_+(x)=V_-(x)-2\partial_x\left({\partial_x\Psi_0\over\Psi_0}\right)=-V_-(x)+2\left({\partial_x\Psi_0\over\Psi_0}\right)^2.$$
The supersymmetric partner potentials $V_+$ and $V_-$ have the
same energy levels, except for $E_0^{(-)}=0$. In terms of the
superpotential $W(x)$, the operators $A$ and $A^+$ are given by
$$A^+=-\partial_x+W(x),\quad A=\partial_x+W(x).$$

Similarly, the supersymmetric partner potentials
$V_{\pm}(x)$ and the superpotential $W(x)$ satisfy:
$${V_+(x) + V_-(x)\over 2}=W^2(x),\quad [A,A^+]=2\partial_x W(x).$$

\textbf{Shape invariance.} 
Gendenshte\"{\i}n, in his remarkable paper \cite{ge}, introduced
the concept of \textit{shape invariance}, which is a property of families of potentials with respect to their
parameters. 
Given a family $a=(a_{n})_{n}$ of parameters,  the shape
invariance condition can be expressed as :
$$V_{n+1}(x; a_n) = V_n(x; a_{n+1}) + R(a_n),\quad V_-=V_0,\quad V_+=V_1,$$ where $R$ is a
remainder, which does not depends on $x$. When this relation holds, we say that
$V=V_-=V_0$ is a \textit{shape invariant potential} (with respect to the given family or transformation), see also \cite{ro,cokasu,duka}.\\

\textbf{Gendenshte\"{\i}n Theorem.}\emph{ Consider the
Schr\"odinger equation $H\Psi=E_n\Psi$, where $V=V_-=V_0$ is a
shape invariant potential. If we fix the first level of energy
$E_0 = 0$, then the excited spectrum and the wave functions are
given respectively by
\begin{displaymath}\label{dukage} E_n = \sum_{k=2}^{n+1}R(a_k),\quad \Psi_n^{(-)}(x;
a_1) = \prod_{k=1}^nA^+(x; a_k)\Psi^{(-)}_0(x;
a_{n+1}).\end{displaymath}}
\medskip

%In \cite{schr}, Schr\"odinger gave factorizations of the Hypergeometric
%equation which were used by Natanzon in \cite{nat} to obtain
%the well known \textit{Natanzon's potentials}, i.e. potentials
%which can be obtained by transformations of the Hypergeometric
%equation and its confluences, see \cite{cokasu,cokasu2}. 
%\\\JA{I don't understand :  In particular, the \textit{Ginocchio potentials} are obtained through
%the Gegenbauer polynomials.}%\\
%
%In \cite{nat}, Natanzon defines a \emph{solvable potential} (also known as \emph{exactly solvable potential})
%as a potential for which the Schr\"odinger equation can be reduced to a
%hypergeometric or confluent hypergeometric form (see the end of section 2 for a recalls on hypergeometric functions).
% \JA{really ? this is not what we say in section 3}

The structure of the rest of the paper is as follows. In section 2 we introduce the necessary concepts and results of the Differential Galois Theory of linear ordinary differential equations. In section 3 we introduce some  terminology about solvability  and we obtain some results about the invariance of the Galois group with respect to Darboux transformations. Section 4 is devoted to the study of the rational potentials and  section 5 to transcendental potentials by means of  the algebrization method.

%%%%%%%%%%%%% SECTION 2
%%%%%%%%%%%%%%%%%%%%%%%%%%%%%%%%%%%%%%%%%%%%%%%%%%
\section{Differential Galois Theory}
In this section, we recall some differential Galois material to study the one-dimensional  Schr\"odinger equation,
$$\partial_x^2\Psi=(V-\lambda)\Psi $$
where $V$ is some function and $\lambda$ is a constant (the spectral parameter).
There are many references on differential Galois theory (or ``Picard-Vessiot'' theory), for example
	\cite{be,crhamo,mo,vasi}.

%%%%%%%%%%%%%%%%%%%%%%%%%
\subsection{Picard-Vessiot Galois theory}
We establish an algebraic model for functions and the corresponding Galois theory.
\\

\textbf{Differential Fields.} 
Let $K$ be a commutative field of characteristic zero. A {\em derivation} of $K$ is a map
$\partial_x : K\rightarrow K$
satisfying $\partial_x (a + b) =
\partial_x a + \partial_x b$ and $\partial_x(ab) = \partial_x a \cdot b + a
\cdot\partial_x b$ for all $a,b\in K$.  We then say that $(K,\partial_{x})$ (or just $K$, when there is no ambiguity) is
a {\it{differential field}} with the derivation $\partial_x$.\\
We assume that $K$ contains an element $x$ such that $\partial_{x}(x)=1$.
Let $\mathcal C$ denote
the field of constants of $K$: $$\mathcal C = \{c\in K |
\partial_x c = 0\}.$$ It is also of characteristic zero and will be assumed to be algebraically closed. \\
Throughout this paper, the \emph{coefficient field} for a differential equation 
(e.g the Schr\"odinger equation $\partial_x^2\Psi=(V-\lambda)\Psi$) is defined as the smallest differential field containing all the coefficients
of the equation. \\

We will mostly analyze  second order linear homogeneous differential equations,
i.e equations of the form
\begin{displaymath}\label{soldeq}
\mathcal L:= \partial^2_xy+a\partial_x y+by=0,\quad a,b\in K.
\end{displaymath}
so the rest of the theory will be explained in this context.
\\

\textbf{Picard-Vessiot Extension.} 
Let $L$ be a differential field containing $K$ (a differential extension of $K$).
We say that $L$ is a \emph{Picard-Vessiot} extension of $K$ for $\mathcal L$
if there exist two linearly independent $y_1, y_2\in L$ solutions of $\mathcal L$ such 
that $L= K\langle y_1, y_2 \rangle$ (i.e $L=K(y_1, y_2,\partial_xy_1, \partial_xy_2)$)
and $L$ and $K$ has the same field of constants $\mathcal{C}$.\\

In what follows, we choose a Picard-Vessiot extension and the term ``solution of $\mathcal L$'' will mean
``solution of $\mathcal L$ in $L$''. So any solution of $\mathcal L$ is a linear combination
(over $\mathcal{C}$) of $y_{1}$ and $y_{2}$.\\

\textbf{Differential Galois Groups} 
A $K$-automorphism $\sigma$ of the Picard-Vessiot extension $L$ is called a differential automorphism 
if it leaves $K$ fixed and commutes with the derivation.  This means that $\sigma(\partial_xa)=\partial_x(\sigma(a))$ for all $a\in L$ and
$\forall a\in K,$ $\sigma(a)=a$.
\\
The group of all differential automorphisms of $L$
over $K$ is called the {\it{differential Galois group}} of $L$
over $K$ and is denoted by ${\rm DGal}(L/K)$. \\

Given $\sigma \in \mathrm{DGal}(L/K)$, we see that $\{\sigma y_1, \sigma y_2\}$ are also solutions of $\mathcal L$. Hence there exists a matrix

$$A_\sigma=
\begin{pmatrix}
a & b\\
c & d
\end{pmatrix}
\in \mathrm{GL}(2,\mathbb{C}),$$ such that
$$\sigma(
\begin{pmatrix}
y_{1} &
y_{2}
\end{pmatrix})
=
\begin{pmatrix}
\sigma (y_{1}) 	&
\sigma (y_{2})
\end{pmatrix}
=\begin{pmatrix} y_{1}& y_{2}
\end{pmatrix}A_\sigma.$$ 

As $\sigma$ commutes with the derivation, this extends naturally to an action on a fundamental solution matrix of the companion first order system associated with 
$\mathcal{L}$.
$$\sigma(
\begin{pmatrix}
y_{1}&y_2\\
\partial_xy_1&\partial_xy_{2}
\end{pmatrix})
=
\begin{pmatrix}
\sigma (y_{1})&\sigma (y_2)\\
\sigma (\partial_xy_1)&\sigma (\partial_xy_{2})
\end{pmatrix}
=\begin{pmatrix} y_{1}& y_{2}\\\partial_xy_1&\partial_xy_2
\end{pmatrix}A_\sigma.$$

This defines a faithful representation $\mathrm{DGal}(L/K)\to
\mathrm{GL}(2,\mathbb{C})$ and it is possible to consider
$\mathrm{DGal}(L/K)$ as a subgroup of $\mathrm{GL}(2,\mathbb{C})$.
It depends on the choice of the fundamental system $\{y_1,y_2\}$,
but only up to conjugacy.\\

Recall that an algebraic group $G$ is an algebraic manifold
endowed with a group structure. 
Let $\mathrm{GL}(n,\mathbb{C})$ denote, as usual, the set of invertible $n\times n$  matrices 
with entries in $\mathbb{C}$ (and $\mathrm{SL}(n,\mathbb{C})$ be the set of matrices with determinant
equal to $1$).
A linear algebraic group will be a subgroup of  $\mathrm{GL}(n,\mathbb{C})$ equipped with a 
structure of algebraic group.
One of the fundamental results of the Picard-Vessiot theory is the
following theorem (see \cite{ka,kol}).\\

\emph{The differential Galois group $\mathrm{DGal}(L/K)$ is an
algebraic subgroup of $\mathrm{GL}(2,\mathbb{C})$.}\\
In fact, the differential Galois group measures the algebraic relations between the solutions (and their derivatives).
It is sometimes viewed as the object which should tell ``what algebra sees of the dynamics of the solutions''.\\
%%%%%%%%%%%%%%%%%%%%%%%%%%%

In an algebraic group $G$, the largest connected algebraic
subgroup of $G$ containing the identity, noted $G^{\circ}$, is a normal subgroup of finite index.
It is often called the connected component of the identity.
If $G=G^0$ then $G$ is a \textit{connected group}. \\
When $G^0$ satisfies some property, 	we say that $G$ virtually
satisfies this property. For example, virtually solvability of $G$
means solvability of $G^0$ and virtual
abelianity of $G$ means abelianity of $G^0$ (see \cite{we2}).\\

\textbf{Lie-Kolchin Theorem.} \emph{Let $G\subseteq
\mathrm{GL}(2,\mathbb{C})$ be a virtually solvable group. Then
$G^0$ is triangularizable, i.e it is conjugate to a subgroup of
upper triangular matrices.}\\

%%%%%%%%%%%%%%%%%%%%%%
\subsection{Algebraic Subgroups of $SL(2,\mathbb{C})$}

We recall below some examples of subgroups of $\mathrm{SL}(2,\mathbb{C})$ for later use.\\

\textbf{Reducible subgroups} These are the groups which leave a non-trivial subspace of $V$ invariant.
They are classified in two categories.\\
\emph{Diagonal groups: }
the identity group: $\{e\}=\left\{\begin{pmatrix}1&0\\0&1\end{pmatrix}\right\}$, 
the $n-$roots: $\mathbb{G}^{[n]}=\left\{\begin{pmatrix}c&0\\0&c^{-1}\end{pmatrix},\quad c^n=1\right\}$,
the multiplicative group: $\mathbb{G}_m=\left\{\begin{pmatrix}c&0\\0&c^{-1}\end{pmatrix},\quad c\in\mathbb{C}^*\right\}$
\\
\emph{Triangular groups: }
the additive group: $\mathbb{G}_a=\left\{\begin{pmatrix}1&d\\0&1\end{pmatrix},\quad d\in\mathbb{C}\right\}$,
the $n-$quasi-roots: $\mathbb{G}^{\{n\}}=\left\{\begin{pmatrix}c&d\\0&c^{-1}\end{pmatrix},\quad c^n=1, \quad d\in\mathbb{C}\right\}$,
the Borel group: $\mathbb{B}=\mathbb{C}^*\ltimes\mathbb{C}=\left\{\begin{pmatrix}c&d\\0&c^{-1}\end{pmatrix},\quad c\in\mathbb{C}^*, \quad d\in\mathbb{C}\right\}$\\

%%%%
\textbf{Irreducible subgroups}
The infinite dihedral group (also called meta-abelian group):\\
 $\mathbb{D}_{\infty}=\left\{\begin{pmatrix}c&0\\0&c^{-1}\end{pmatrix},\quad c\in\mathbb{C}^*\right\}\cup\left\{\begin{pmatrix}0&d\\-d^{-1}&0\end{pmatrix},\quad d\in\mathbb{C}^*\right\}$ and its finite subgroups $\mathbb{D}_{2n}$ (where $c$  and $d$ spans the $n$-th roots of unity).
 \\
There are also three other finite irreducible (primitive) groups: the tetrahedral group $A_4^{\mathrm{SL}_2}$ of order $24$, the octahedral group $S_4^{\mathrm{SL}_2}$ of order $48$, and the icosahedral group $A_5^{\mathrm{SL}_2}$ of order $120$.

\textbf{Integrability.} 
We say that the linear differential equation $\mathcal L$ is (Liouville) \textit{integrable} if
the Picard-Vessiot extension $L\supset K$ is obtained as a tower
of differential fields $K=L_0\subset L_1\subset\cdots\subset
L_m=L$ such that $L_i=L_{i-1}(\eta)$ for $i=1,\ldots,m$, where
either
\begin{enumerate}
\item $\eta$ is {\emph{algebraic}} over $L_{i-1}$, that is $\eta$ satisfies a
polynomial equation with coefficients in $L_{i-1}$.
\item $\eta$ is {\emph{primitive}} over $L_{i-1}$, that is $\partial_x\eta \in L_{i-1}$.
\item $\eta$ is {\emph{exponential}} over $L_{i-1}$, that is $\partial_x\eta /\eta \in L_{i-1}$.
\end{enumerate}
\medskip

We remark that the usual terminology in differential algebra for
integrable equations is that the corresponding Picard-Vessiot
extensions are called \textit{Liouvillian}. The following theorem is due to Kolchin.\\

\emph{The equation $\mathcal L$ is integrable if and only if
$\mathrm{DGal}(L/K)$ is virtually solvable.}\\

%%%%%%%%%%%%%%%%%%%%%%%%%%%

\subsection{Eigenrings}
Eigenrings are a tool, introduced by M.F. Singer, to study the decomposability of systems or operators. A general reference is 
\cite{ba} (and references therein) .
\\

\textbf{Eigenrings of Systems.}  Consider the system
$\partial_xX=-AX$, denoted by $[A]$, where $A\in\mathrm{GL}(2,K)$.
The eigenring  of the system $[A]$, denoted by $\mathcal E(A)$, is
the set of $2\times 2$ matrices $P$ in $K$ satisfying
\begin{displaymath} \partial_xP = PA - AP.
\end{displaymath} 

\textbf{Eigenrings of Operators.} Let $\mathfrak{L}$ be a second
order differential operator, i.e $\mathcal{L}:=\mathfrak{L}(y)=0$.
Denote $V(\mathfrak{L})$ as the solution space of $\mathcal{L}$.
The \textit{eigenring } of $\mathfrak{L}$, denoted by
$\mathcal{E}(\mathfrak{L})$, is the set of all operators
$\mathfrak R$ of order less than that of $\mathfrak{L}$ for which $\mathfrak R(V(\mathfrak L))$ is a subset
of $V(\mathfrak L)$, that is $\mathfrak{LR}=\mathfrak{SL}$, where
$\mathfrak S$ is also an operator. For more details see
\cite{si,ba,ho,ho1,ho2}.\\

We study an equation $\mathfrak{L}(y)=0$ with $\mathfrak
L=\partial_x^2+p\partial_x+q$ ; this is equivalent to the system
$[A]$, where $A$ is given by
$$A=\begin{pmatrix}0&-1\\q&p\end{pmatrix}\quad  p,q\in K.$$ 

Relations between the system approach and the operator approach are given in the following simple lemmas .\\

\textbf{Lemma.} Let ${\mathfrak L}=\partial_x^2+p\partial_x+q$,  
$ A=\begin{pmatrix}0&-1\\q&p\end{pmatrix}$ and 
$P=\begin{pmatrix}a&b\\c&d\end{pmatrix}$. %%,\quad a,b,c,d,p,q\in K.$$
Then:
\begin{enumerate}
\item If $P\in\mathcal E(A)$, then $\mathfrak
R=a+b\partial_x\in\mathcal{E}(\mathfrak L)$.
\item If $\mathfrak R=a+b\partial_x\in\mathcal E (\mathfrak
L)$, then $P\in\mathcal E(A)$, where $P$ is given by
$$P=\begin{pmatrix}a&b\\\partial_xa-bq&a+\partial_xb-bp\end{pmatrix}.$$
\item $1\leq \dim_{\mathcal C }\mathcal{E}(\mathfrak L)\leq 4$.
\item $P\in
\mathrm{GL}(2,K)\Leftrightarrow
\frac{\partial_xa}{a}-\frac{a}{b}+p\neq\frac{\partial_xb}{b}-\frac{b}{a}q$.
\end{enumerate}
\medskip

\textbf{Lemma.}
 Assume that $\mathfrak L=\partial_x^2+b$, where $b\in
K$. The following statements hold:

\begin{enumerate}
\item If $\dim_{\mathcal C}\mathcal{E}(\mathfrak L)=1$, then either
the differential Galois group is irreducible ($\mathbb{D}_{\infty}$ or
primitive), or indecomposable ($G\subseteq \mathbb{B}$, $G$ non-diagonal).
%\notin\{e,\mathbb{G}_m,\mathbb{G}_a, \mathbb{G}^{\{n\}},\mathbb{G}^{[n]}\}$).

\item If $\dim_{\mathcal C}\mathcal{E}(\mathfrak L)=2$, then either the differential Galois group is the additive group or it is contained in the multiplicative group (but is not the identity group). 
%Thus, we can have two solutions but not over the differential field $K$.

\item If $\dim_{\mathcal C}\mathcal{E}(\mathfrak L)=4$, then the differential Galois group is the
identity group. In this case we have 2 independent solutions
$\zeta_1$ and $\zeta_2$ in which $\zeta_1^2$, $\zeta_2^2$ and
$\zeta_1\zeta_2$ are elements of the differential field $K$, i.e.
the solutions of $\mathcal L ^{\sympow 2}$ belong to to the
differential field $K$.
\end{enumerate}

%\JA{Need a proof here}

%\JA{Here, add notion of homomorphism of differential operators - useful for DT}

%%%%%%%%%%%%%%%%%%%%%%%%%%%

\subsection{Some useful special functions}
Some standard special functions will be used below.

\textbf{The Gauss hypergeometric equation.} It is a particular case of \emph{Riemann's equation} (see
\cite{iksy,po,ki,marram,dulo}) and is given by
\begin{displaymath}\label{hypergeometric}
\partial_x^2y+{\left(\gamma-(\kappa+\beta+1)x\right)\over
x(1-x)}\partial_xy-{\kappa\beta\over x(1-x)} y=0.
\end{displaymath}
The cases when it is solvable can be found in the Kimura table \cite{ki} (or via Kovacic's algorithm, see appendix).
\medskip

\textbf{The confluent hypergeometric equations (Kummer, Whittaker).} It is a degenerate
form of the hypergeometric differential equation where two of the
three regular singularities merge into an irregular singularity.
For example, making ``$1$ tend to $\infty$" in a suitable way, the
hypergeometric equation may transform to two classical forms:
\begin{itemize}
\item \textit{Kummer's} form \begin{displaymath}\label{kummer} \partial_x^2y+{c-x\over x}\partial_xy-{a\over x}y=0\end{displaymath}

\item \textit{Whittaker's} form \begin{displaymath}\label{whittaker} \partial_x^2y=\left(\frac14-{\kappa\over x}+{4\mu^2-1\over 4x^2}\right)y\end{displaymath}
\end{itemize}
where the parameters of the two equations are linked by
$\kappa=\frac{c}2-a$ and $\mu=\frac{c}2-\frac12$. The Galoisian
structure of these equations has been deeply studied in
\cite{marram,dulo}. The following theorem is due to Martinet \&
Ramis, \cite{marram}.\\

\emph{The Whittaker's differential equation is integrable if and
only if either, $\kappa+\mu\in\frac12+\mathbb{N}$, or
$\kappa-\mu\in\frac12+\mathbb{N}$, or
$-\kappa+\mu\in\frac12+\mathbb{N}$, or
$-\kappa-\mu\in\frac12+\mathbb{N}$. }\\

\textbf{The Bessel equation.} It is a particular case of the
confluent hypergeometric equation and is given by
\begin{displaymath}\label{bessel} \partial_x^2y+{1\over x}\partial_xy+{x^2-n^2\over
x^2}y=0.\end{displaymath} Under a suitable transformation, the
reduced form of the Bessel's equation is a particular case of the
Whittaker's equation. Thus, we can obtain the following well known
result, see \cite[p. 417]{kol} and see also \cite{ko,mo}.\\

\emph{The Bessel's differential equation is
integrable if and only if $n\in \frac12+\mathbb{Z}$. }\\

By double confluence of the hypergeometric equation, that is
making ``$0$ and $1$ tend to $\infty$" in a suitable way, one gets
the \textit{parabolic cylinder equation} (also known as
\textit{Weber's} equation):
\begin{displaymath}\label{weber} \partial_x^2y=\left(\frac14x^2-{1\over
2}-n\right)y,\end{displaymath} which is integrable if and only if
$n\in\mathbb{Z}$, see \cite{ko, dulo}. Under suitable
transformations one can gets the Rehm's form of the Weber's
equation:
\begin{displaymath}\label{rehm}
\partial_x^2y=\left(ax^2+2bx+c\right)y,\quad a\neq
0,\end{displaymath} so that ${b^2-c\over a}$ is an odd integer.\\

The hypergeometric equation, including confluences, is a
particular case of the differential equation
\begin{displaymath}\label{eqorpol}
\partial_x^2y+\frac{L}{Q}\partial_xy+\frac{\lambda}{Q} y,\quad
\lambda\in\mathbb{C},\quad L=a_0+a_1x,\quad Q=b_0+b_1x+b_2x^2.
\end{displaymath}
We mention this because the \textit{classical orthogonal polynomials} and
\textit{Bessel polynomials} are solutions of such an equation (see
\cite{ch,is,niuv}). We remark that integrability conditions and
solutions of differential equations with solutions orthogonal
polynomials, including Bessel polynomials, can be obtained
applying Kovacic's algorithm. Similarly, one can apply
Kovacic's algorithm to obtain the same results given by Kimura
\cite{ki} and Martinet \& Ramis \cite{marram}. Also we recall that
Duval \& Loday-Richaud applied Kovacic's algorithm in \cite{dulo} to some
families of special functions, including several of the above.

%%%%%%%%%%%%%

%\section{Some Preliminaries}
\section{Notions of Algebraically Solvable and Quasi-Solvable Potentials}\label{algebrizationsec}

The main object of our Galoisian analysis is the one-dimensional Schr\"odinger
equation, written as
\begin{equation}\label{equ1}
\mathcal L_\lambda:=H\Psi=\lambda\Psi, \quad
H=-\partial_x^2+V(x),\quad V\in K,
\end{equation}
where $K$ is a differential field (with $\mathbb{C}$ as field of
constants). We are interested in the integrability (in Galoisian sense) of equation \eqref{equ1}.
Denote by $L_\lambda$ the Picard-Vessiot extension of $K$ for $\mathcal L_\lambda$. 
The differential Galois group of $\mathcal L_\lambda$ is $\mathrm{DGal}(L_\lambda/K)$.

\subsection{Algebraic Spectrum and Algebraic Solvability}
Let $\Lambda\subseteq{\mathbb{C}}$ denote the set of 
eigenvalues $\lambda$ such that equation (\ref{equ1}) is
integrable in Galoisian sense. The set $\Lambda$ will be called \textit{the algebraic spectrum} (or alternatively \textit{the Liouvillian spectral set}) of
$H$. We remark that $\Lambda$ can be $\emptyset$, i.e.,
$\mathrm{DGal}(L_\lambda/K)={\rm SL}(2,\mathbb{C})$ $\forall
\lambda \in \mathbb{C}.$ On the other hand, the classification of subgroups of $SL_{2}(\mathbb{C})$ shows that,
 if $\lambda_0\in \Lambda$, then
$(\mathrm{DGal}(L_{\lambda_0}/K))^0\subseteq\mathbb{B}$.\\

We let $\Lambda_+$ be the set $\{\lambda\in\Lambda\cap \mathbb{R}: \lambda\geq 0\}$ and
$\Lambda_-$ be the set $\{\lambda\in\Lambda\cap \mathbb{R}:
\lambda\leq 0\}$.

\begin{definition}[Algebraically Solvable and Quasi-Solvable Potentials] We say that the
potential $V(x)\in K$ is:
\begin{itemize}
\item an \textit{algebraically solvable potential} when $\Lambda$ is an infinite set, or
\item an \textit{algebraically quasi-solvable potential} when $\Lambda$ is a non-empty finite
set, or
\item an \textit{algebraically non-solvable potential} when $\Lambda=\emptyset$.
\end{itemize}
When $\textrm{Card}(\Lambda)=1$, we say that $V(x)\in K$ is a
\textit{trivial} algebraically quasi-solvable potential.
\end{definition}
\medskip

\begin{examples} Let $K=\mathbb{C}(x)$.
\begin{enumerate}
\item If $V(x)=x$ (Airy equation), then
$\Lambda=\emptyset$, $V(x)$ is algebraically non-solvable, see
\cite{ka,ko}.
\item If $V(x)=0$, then $\Lambda=\mathbb{C}$, i.e., $V(x)$ is algebraically solvable. Furthermore,
$${\mathrm{DGal}}(L_0/K)=\{e\},\quad\hbox{\rm and }\; {\mathrm{DGal}}(L_{\lambda}/K)=\mathbb{G}_m \;\hbox{\rm for }\; \lambda\neq 0.$$
\item If $V(x)={x^2\over 4}+\frac12$, then $\Lambda=\{n:
n\in\mathbb{Z}\}$, $V(x)$ is algebraically solvable. This example
corresponds to the Weber's equation given previously.
\item If $V(x)=x^4-2x$, then $\Lambda=\{0\}$ and $V(x)$ is a trivial
algebraically quasi-solvable potential (see next section for tools to prove this).
\end{enumerate}
\end{examples}
\medskip

\noindent We are interested in the spectrum (analytic spectrum) of
the algebraically solvable and quasi-solvable potentials, that is,
$\mathrm{Spec}(H)\cap \Lambda\neq\emptyset$. 
%For example, the
%potential $V(x)=|x|$ has point spectrum (see \cite{bana}) although
%$V(x)$ is algebraically non-solvable. 
%%\JA - this does not make sense : how do we define |x| algebraically?
%

%\JA{I have not read Natanzon but elsewhere in the paper we said that Natanzon called solvable a potential which stems from the hypergeometric equation : what is correct ? I'll assume that it is this the definition below}
Potentials for which $\mathrm{Spec}(H)\cap \Lambda$ is an infinite set  are called \textit{solvable
(or exactly solvable) potentials} (see Natanzon \cite{nat}) in the 
physics litterature. 
Similarly, potentials for which $\mathrm{Spec}(H)\cap \Lambda$ is a finite
set are usually called \textit{quasi-exactly solvable (or quasi-solvable) potentials}
(Turbiner \cite{tu}, Bender \& Dunne \cite{bedu}, Bender \&
Boettcher \cite{bebo}, Saad et al. \cite{sahaci}, Gibbons \&
Vesselov \cite{give}).\\

\subsection{Changing the Base Field}
In what follows, we will make transformations which may change the base field; typically, when $K$ is not
$\mathbb{C}(x)$, we will look for transformations to another differential with rational coefficients.
We now collect facts on how this affects the differential Galois group and its solvability.

\begin{definition}\label{deftriso}
Let $\mathcal{L}$ and $\widetilde{\mathcal{L}}$ be a pair of linear
differential equations defined over differential fields $K$ and
$\widetilde{K}$ respectively, with Picard-Vessiot extensions $L$
and $\widetilde{L}$. Let $\varphi$ be a transformation mapping $K$ to $\widetilde{K}$
and $\mathcal{L}$ to $\widetilde{\mathcal{L}}$ (so we may also choose $\varphi$
to map $L$ to  $\widetilde{L}$).
We say that:
\begin{enumerate}
\item $\varphi$ is an {\it isogaloisian transformation} if
$$\mathrm{DGal}(L/K)=\mathrm{DGal}(\widetilde{L}/\widetilde{K}).$$
When $\widetilde{L}=L$ and $\widetilde{K}=K$, we'll say that $\varphi$
is a {\it strong isogaloisian transformation}.
\item $\varphi$ is a {\it virtually isogaloisian transformation} if $$(\mathrm{DGal}(L/K))^0=(\mathrm{DGal}(\widetilde{L}/\widetilde{K}))^0.$$
\end{enumerate}
\end{definition}

\begin{remark}\label{prop1} The eigenring s of two operators $\mathfrak
L$ and $\widetilde{\mathfrak L}$ are preserved under isogaloisian
transformations.
\end{remark}

%\JA{restate this cleanly, case when f not in K also -- }
\begin{proposition}\label{rlig} Consider the
differential equations
$$\mathcal{L}:=\partial_x^2y+a\partial_xy+by=0,\quad \widetilde{\mathcal{L}}:=\partial_x^2\zeta=r\zeta,\quad a,\,b,\,r\,\in K.$$
Let $\kappa\in\mathbb{Q}$, $f\in K$, $a=2\kappa\partial_x(\ln f)$
and $\varphi$ be the transformation such that
$\mathcal{L}\mapsto\widetilde{\mathcal{L}}$. The following
statements holds:

\begin{enumerate}
\item $\varphi$ is a strong isogaloisian transformation for
$\kappa\in\mathbb{Z}$.
\item $\varphi$ is a virtually strong isogaloisian transformation for
$\kappa\in\mathbb{Q}\setminus\mathbb{Z}$.
\end{enumerate}
\end{proposition}

\begin{proof} Let $\mathcal{B}=\{y_1,y_2\}$ be a basis
of solutions of $\mathcal{L}$ and $L$ be the corresponding Picard-Vessiot extension of $K$;
similarly, let $\mathcal{B}'=\{\zeta_1,\zeta_2\}$ be a basis of
solutions of $\widetilde{\mathcal{L}}$ if a Picard-Vessiot extension $\widetilde L$.
 With the change of dependent variable
$y=\zeta e^{-\frac12\int a}$ we obtain $r=a^2/4+\partial_x a/2-b$
and $\widetilde K=K$. Thus, the relationship between
$L$ and $\widetilde L$ depends on $a$:
\begin{enumerate}
\item If $\kappa=n\in\mathbb{Z}$, then $\mathcal{B}'=\{f^ny_1,f^ny_2\}$
which means that $L=\widetilde L$ and $\varphi$ is strong
isogaloisian.
\item If $\kappa=\frac{n}m$, with $\gcd(n,m)=1$, $\frac{n}m\notin\mathbb{Z}$, then $\mathcal{B}'=\{f^{n\over m} y_1,f^{n\over m}y_2\}$
which means that $\widetilde L$ is either an algebraic extension
of degree at most $m$ of $L$, and $\varphi$ is virtually strong
isogaloisian, or $L=\widetilde L$ when $f^{\frac{n}m}\in K$ which
means that $\varphi$ is strong isogaloisian.
\end{enumerate}
\end{proof}
\begin{remark} The transformation $\varphi$ in proposition \ref{rlig} is
not injective, there are a lot of differential equations
$\mathcal{L}$ that are transformed in the same differential
equation $\widetilde{\mathcal{L}}$. They are called projectively equivalent.
\end{remark}
\medskip

%\JA{do we use this ? where do we ? should we keep this ?}
The following is an easy consequence of the proposition: \\

\begin{corollary} Let $\mathcal{L}$ be the differential equation $$\partial_x\left(a \partial_xy\right)=(\lambda b-\mu)y,\quad
a,b\in K,\quad \lambda,\mu\in\mathbb{C}$$ 
where $L$,
$\widetilde L$, $\widetilde{\mathcal{L}}$ and $\varphi$ are given
as in proposition \ref{rlig}. Then either $\widetilde L$ is a
quadratic extension of $L$ (which means that $\varphi$ is virtually
strong isogaloisian) or $\widetilde L=L$ (this is when $a^{\frac12}\in K$)
which means that $\varphi$ is strong isogaloisian.
\end{corollary}

%%%%%%%%%% NEW SECTION PUT HERE

\section{Galoisian Aspects of Darboux Transformations and Shape Invariance}

Here we present a Galoisian approach to the Darboux transformation, the
Crum iteration and shape invariant potentials.

Recall that 
$Wr(y_1,\ldots,y_n)$ denotes the Wronskian
$$Wr(y_1,\ldots,y_n)=\left|\begin{matrix}y_1&\cdots& y_n\\
 \vdots& & \vdots \\ \partial^{n-1}_xy_1&\cdots&
 \partial^{n-1}_xy_n\end{matrix}\right|.$$ 
 The Darboux transformation will be abreviated by $\mathrm{DT}$
 and its $n$-th iteration by $\mathrm{DT}_n$; similarly, 
 the abreviation $\mathrm{CI}_n$ will refer to the Crum iteration.  
 
As above, $K$ denotes
 the coefficient field of the studied linear differential equation, i.e. the smallest differential containing the coefficients of the equation.  We mostly consider equations over the base field $K=\mathbb{C}(x)$.
  
%% new by primi
%%%%%%%%%%%%%%%%%%%%%%%%%%%%%%%%%%%%%

  \subsection{Darboux Transformations}
The following result comes from the original Darboux transformation (see \cite{da1,da2}), but rewritten in a way which is adapted to our Galoisian analysis.
%\subsection{Classical Approach}
%\begin{theorem}[Galoisian version of DT]\label{darth} 
\begin{theorem}[Darboux Transformation]\label{darth} 
Consider the Schr\"odinger operators $H_\pm=\partial_x^2+ V_{\pm}(x)$, with $V_{\pm}(x)\in K$, 
and assume that the spectrum satisfies $\Lambda\neq\emptyset$. 
Let $\mathcal{L}_\lambda$ (resp. $\widetilde{\mathcal{L}}_\lambda$) denote the Schr\"odinger equation $H_-\Psi^{(-)}=\lambda\Psi^{(-)}$
(resp. $H_+\Psi^{(+)}=\lambda\Psi^{(+)}$).\\
Let $\mathrm{DT}$ be the transformation
such that $\mathcal{L}\mapsto\widetilde{\mathcal{L}}$, $V_-\mapsto
V_{+}$,
 $\Psi^{(-)}\mapsto\Psi^{(+)}$. Then the following statements
 holds:
\begin{description}
\item[i)] $\mathrm{DT}(V_-)=V_+=\Psi^{(-)}_{\lambda_1}\partial_x^2\left({1\over
\Psi^{(-)}_{\lambda_1}}\right)+\lambda_1=V_--2\partial_x^2(\ln\Psi^{(-)}_{\lambda_1}),$\\
$\mathrm{DT}(\Psi^{(-)}_{\lambda_1})={\Psi}^{(+)}_{\lambda_1}={1\over\Psi^{(-)}_{\lambda_1}}$,
 where $\Psi^{(-)}_{\lambda_1}$ is a particular solution of $\mathcal{L}_{\lambda_1}$, $\lambda_1\in\Lambda$.

\item[ii)] $\mathrm{DT}(\Psi^{(-)})=\Psi^{(+)}=\partial_x\Psi^{(-)}-\partial_x(\ln\Psi^{(-)}_{\lambda_1})\Psi^{(-)}
={W(\Psi^{(-)}_{\lambda_1},\Psi^{(-)})\over
W(\Psi^{(-)}_{\lambda_1})}$, $\lambda\neq\lambda_1$, where
$\Psi^{(-)}=\Psi^{(-)}_{\lambda}$ is the general solution of
$\mathcal{L}_\lambda$ for
$\lambda\in\Lambda\setminus\{\lambda_1\}$ and
$\Psi^{(+)}=\Psi^{(+)}_{\lambda}$ is the general solution of
$\widetilde{\mathcal{L}}_\lambda$ also for
$\lambda\in\Lambda\setminus\{\lambda_1\}$.
\end{description}
\end{theorem}

\begin{proof} We treat each item separately.
\begin{description}
\item[i)] By hypothesis, we have  $\mathrm{Card}(\Lambda)\geq 1$,
$\lambda_1\in\Lambda$,\\
$H_-\Psi^{(-)}_{\lambda_1}=\lambda_1\Psi^{(-)}_{\lambda_1},$ and setting
$\widetilde H_\pm= H_\pm-\lambda_1$, $\widetilde V_\pm= V_\pm-\lambda_1$ we obtain
$\widetilde H_-\Psi^{(-)}_{\lambda_1}=0$. Thus, the superpotential $W=-\partial_x\ln\Psi^{(-)}_{\lambda_1}$, depending of $\lambda_1$, exists and satisfies
$$\left(-\partial_x +W\right)\left({\partial_x}
+W\right)\Psi^{(-)}_{\lambda_1}=\widetilde{H}_-\Psi^{(-)}_{\lambda_1}=0.$$
Now, developing both sides of the expression, we have 
$\widetilde{V}_-=-\partial_x W+W^2$. So, interchanging operators we have
$$\left(\partial_x +W \right)\left(-\partial_x +W\right)\Psi^{(+)}_{\lambda_1}=\widetilde{H}_+\Psi^{(+)}_{\lambda_1}=0,$$
so that developing both sides of the expression we have 
$\widetilde{V}_+=\partial_xW+W^2$ which lead us to $W=\partial_x\ln\Psi^{(+)}_{\lambda_1}$. Thus, we obtain
$$\mathrm{DT}(\Psi^{(-)}_{\lambda_1})=\Psi^{(+)}_{\lambda_1}={1\over
\Psi^{(-)}_{\lambda_1}}.$$ On another hand, as
$\widetilde{V}_+-\widetilde{V}_-=V_+-V_-=2\partial_xW$,
$\widetilde{V}_++\widetilde{V}_-=V_++V_--2\lambda_1=2W^2$, we obtain $\mathrm{DT}(V_-)=V_+$ which is given by
$$V_--2\partial_x \left({\partial_x\Psi^{(-)}_{\lambda_1}\over
\Psi^{(-)}_{\lambda_1} }\right)=2\left({\partial_x\Psi^{(-)}_{\lambda_1}\over
\Psi^{(-)}_{\lambda_1}}\right)^2-V_-+2\lambda_1=\Psi^{(-)}_{\lambda_1}\partial_x^2\left({1\over\Psi^{(-)}_{\lambda_1}}\right)+\lambda_1.$$

\item[ii)] Suppose that $\mathrm{Card}(\Lambda)>1$,
$\lambda\neq\lambda_1$ and $\Psi^{(\pm)}=\Psi^{(\pm)}_\lambda$ satisfying

$$H_-\Psi^{(-)}=\lambda\Psi^{(-)},\quad \lambda\neq\lambda_1,$$
applying again the raising and lowering operators we have

$$\left(-\partial_x+W\right)\left(\partial_x
+W\right)\Psi^{(-)}\neq 0,\quad \left(\partial_x+W\right)\left(-\partial_x+W\right)\Psi^{(+)}\neq 0,$$ so that

$$\left(-\partial_x+W\right)\left(\partial_x+W\right)\Psi^{(-)}=\left(-\partial_x+W\right)\Psi^{(+)},$$ we obtain
$$\Psi^{(+)}=\left(\partial_x+W\right)\Psi^{(-)}=\partial_x\Psi^{(-)}+ W\Psi^{(-)}$$ $$=\partial_x\Psi^{(-)}-{\partial_x\Psi^{(-)}_{\lambda_1}\over
\Psi^{(-)}_{\lambda_1}}\Psi^{(-)}={W(\Psi^{(-)}_{\lambda_1},\Psi^{(-)})\over
W(\Psi^{(-)}_{\lambda_1})}.$$ Now, we will show that effectively
$\Psi^{(+)}$ satisfy
$H_+\Psi^{(+)}=\lambda\Psi^{(+)}$
for $\lambda\neq\lambda_1$. Let suppose that
$$\Psi^{(+)}=\partial_x\Psi^{(-)}-\Psi^{(-)}{\partial_x\Psi^{(-)}_{\lambda_1}\over \Psi^{(-)}_{\lambda_1}},$$
we can see that

$$\partial^2_x\Psi^{(-)}-\Psi^{(-)}{\partial^2_x\Psi^{(-)}_{\lambda_1}\over \Psi^{(-)}_{\lambda_1}}=\Psi^{(-)}(\lambda_1-\lambda).$$
therefore, differentiating$\Psi^{(+)}$, we
have 
$$\partial_x\Psi^{(+)}=\partial^2_x\Psi^{(-)}-\partial_x\Psi^{(-)}{\partial_x\Psi^{(-)}_{\lambda_1}\over \Psi^{(-)}_{\lambda_1}}-\Psi^{(-)}{\partial^2_x\Psi^{(-)}_{\lambda_1}\over \Psi^{(-)}_{\lambda_1}}+\Psi^{(-)}\left({\partial_x\Psi^{(-)}_{\lambda_1}\over \Psi^{(-)}_{\lambda_1}}\right)^2$$ $$=\Psi^{(-)}(\lambda_1-\lambda)+\Psi^{(-)}\left({\partial_x\Psi^{(-)}_{\lambda_1}\over \Psi^{(-)}_{\lambda_1}}\right)^2-\partial_x\Psi^{(-)}{\partial_x\Psi^{(-)}_{\lambda_1}\over \Psi^{(-)}_{\lambda_1}},$$

By differentiating $\partial_x\Psi^{(+)}$, we have
$$\partial_x^2\Psi^{(+)}=\partial_x\Psi^{(-)}(\lambda_1-\lambda)+2\partial_x\Psi^{(-)}\left({\partial_x\Psi^{(-)}_{\lambda_1}\over \Psi^{(-)}_{\lambda_1}}\right)^2+2\Psi^{(-)}{\partial_x\Psi^{(-)}_{\lambda_1}\over \Psi^{(-)}_{\lambda_1}}{\partial^2_x\Psi^{(-)}_{\lambda_1}\over \Psi^{(-)}_{\lambda_1}}$$ $$-2\Psi^{(-)}\left({\partial_x\Psi^{(-)}_{\lambda_1}\over \Psi^{(-)}_{\lambda_1}}\right)^3-\partial^2_x\Psi^{(-)}{\partial_x\Psi^{(-)}_{\lambda_1}\over \Psi^{(-)}_{\lambda_1}}-\partial_x\Psi^{(-)}{\partial^2_x\Psi^{(-)}_{\lambda_1}\over \Psi^{(-)}_{\lambda_1}}$$
$$=\partial_x\Psi^{(-)}(\lambda_1-\lambda)+2\partial_x\Psi^{(-)}\left({\partial_x\Psi^{(-)}_{\lambda_1}\over \Psi^{(-)}_{\lambda_1}}\right)^2-2\Psi^{(-)}\left({\partial_x\Psi^{(-)}_{\lambda_1}\over \Psi^{(-)}_{\lambda_1}}\right)^3-\partial_x\Psi^{(-)}{\partial^2_x\Psi^{(-)}_{\lambda_1}\over \Psi^{(-)}_{\lambda_1}}$$ $$-{\partial_x\Psi^{(-)}_{\lambda_1}\over \Psi^{(-)}_{\lambda_1}}\left(\partial^2_x\Psi^{(-)}-\Psi^{(-)}{\partial^2_x\Psi^{(-)}_{\lambda_1}\over \Psi^{(-)}_{\lambda_1}}\right)+\Psi^{(-)}{\partial_x\Psi^{(-)}_{\lambda_1}\over \Psi^{(-)}_{\lambda_1}}{\partial^2_x\Psi^{(-)}_{\lambda_1}\over \Psi^{(-)}_{\lambda_1}}$$
$$=\partial_x\Psi^{(-)}(\lambda_1-\lambda)+2\partial_x\Psi^{(-)}\left({\partial_x\Psi^{(-)}_{\lambda_1}\over \Psi^{(-)}_{\lambda_1}}\right)^2-2\Psi^{(-)}\left({\partial_x\Psi^{(-)}_{\lambda_1}\over \Psi^{(-)}_{\lambda_1}}\right)^3-\partial_x\Psi^{(-)}{\partial^2_x\Psi^{(-)}_{\lambda_1}\over \Psi^{(-)}_{\lambda_1}}$$ $$-\Psi^{(-)}{\partial_x\Psi^{(-)}_{\lambda_1}\over \Psi^{(-)}_{\lambda_1}}\left(\lambda_1-\lambda\right)+\Psi^{(-)}{\partial_x\Psi^{(-)}_{\lambda_1}\over \Psi^{(-)}_{\lambda_1}}{\partial^2_x\Psi^{(-)}_{\lambda_1}\over \Psi^{(-)}_{\lambda_1}}$$
$$=(\lambda_1-\lambda)\left(\partial_x\Psi^{(-)}-\Psi^{(-)}{\partial_x\Psi^{(-)}_{\lambda_1}\over \Psi^{(-)}_{\lambda_1}}\right)-{\partial^2_x\Psi^{(-)}_{\lambda_1}\over \Psi^{(-)}_{\lambda_1}}\left(\partial_x\Psi^{(-)}-\Psi^{(-)}{\partial_x\Psi^{(-)}_{\lambda_1}\over \Psi^{(-)}_{\lambda_1}}\right)$$ $$+2\left({\partial_x\Psi^{(-)}_{\lambda_1}\over \Psi^{(-)}_{\lambda_1}}\right)^2\left(\partial_x\Psi^{(-)}-\Psi^{(-)}{\partial_x\Psi^{(-)}_{\lambda_1}\over \Psi^{(-)}_{\lambda_1}}\right)$$
$$=\left(2\left({\partial_x\Psi^{(-)}_{\lambda_1}\over \Psi^{(-)}_{\lambda_1}}\right)^2-{\partial^2_x\Psi^{(-)}_{\lambda_1}\over \Psi^{(-)}_{\lambda_1}}+\lambda_1-\lambda\right)\left(\partial_x\Psi^{(-)}-\Psi^{(-)}{\partial_x\Psi^{(-)}_{\lambda_1}\over \Psi^{(-)}_{\lambda_1}}\right)$$ $$=\left(2\left({\partial_x\Psi^{(-)}_{\lambda_1}\over \Psi^{(-)}_{\lambda_1}}\right)^2-{\partial^2_x\Psi^{(-)}_{\lambda_1}\over \Psi^{(-)}_{\lambda_1}}+\lambda_1-\lambda\right)\Psi^{(+)}$$
$$=\left(2\left({\partial_x\Psi^{(-)}_{\lambda_1}\over
\Psi^{(-)}_{\lambda_1}}\right)^2-V_-+2\lambda_1-\lambda\right)\Psi^{(+)},$$
which implies that, for $\lambda\neq\lambda_1$, the function
$\Psi^{(+)}$ satisfies the Schr\"odinger equation $H_+\Psi^{(+)}=\lambda\Psi^{(+)}$,
$$V_+=2\left({\partial_x\Psi^{(-)}_{\lambda_1}\over
\Psi^{(-)}_{\lambda_1}}\right)^2-V_-+2\lambda_1=V_--2\partial_x\left({\partial_x\Psi^{(-)}_{\lambda_1}\over
\Psi^{(-)}_{\lambda_1}
}\right)=\Psi^{(-)}_{\lambda_1}\partial_x^2\left({1\over\Psi^{(-)}_{\lambda_1}}\right)+\lambda_1.$$
\end{description}
\end{proof}

\begin{remark}
According to the corollary of the general Darboux transformation, see section \ref{introd}, we have  $m=-\lambda$, $m_1=-\lambda_1$, $f(x)=V_-$, $y=\Psi^{(-)}$ and $u=y=\Psi^{(+)}$. This gives us a faster proof of theorem \ref{darth}.
\end{remark}
In agreement with the previous theorem, we obtain the following
results.

\begin{proposition}\label{darbiso}
$\mathrm{DT}$ is isogaloisian and virtually strong isogaloisian.
Furthermore, if $\partial_x(\ln\Psi^{(-)}_{\lambda_1})\in K$, then
$\mathrm{DT}$ is strong isogaloisian.
\end{proposition}

\begin{proof}
 Let $K$ and  ${L_\lambda}$ be respectively the coefficient field and
the Picard-Vessiot extension of the equation
$\mathcal{L}_\lambda$. Similarly, let $\widetilde{{K}}$,
$\widetilde{{L}}_\lambda$ be the coefficient field and the
Picard-Vessiot extension of the equation
$\widetilde{\mathcal{L}}_\lambda$. As $\mathrm{DT}(V_-)=V_+=2W^2-V_--2\lambda_1$, where
$W=-\partial_x(\ln\Psi^{(-)}_{\lambda_1})$, we have
$\widetilde{{K}}=K\left\langle
\partial_x(\ln\Psi^{(-)}_{\lambda_1})\right\rangle.$ The Riccati equation
$\partial_xW=V_--W^2$ has one algebraic solution (Singer 1981,
\cite{sil}), in this case $W=-\partial_x(\ln
\Psi^{(-)}_{\lambda_1})$. Let $\langle
\Psi^{(-)}_{(1,\lambda)},\Psi^{(-)}_{(2,\lambda)}\rangle$ be a
basis of solutions for equation $\mathcal{L}_\lambda$ and $\langle
\Psi^{(+)}_{(1,\lambda)},\Psi^{(+)}_{(2,\lambda)}\rangle$ a basis
of solutions for equation $\widetilde{\mathcal{L}}_\lambda$. Since
the coefficient field for equation
$\widetilde{\mathcal{L}}_\lambda$ is $\widetilde
K=K\langle\partial_x(\ln \Psi^{(-)}_{\lambda_1})\rangle$, we have
 $L=K\langle
\Psi^{(-)}_{(1,\lambda)},\Psi^{(-)}_{(2,\lambda)}\rangle$ and
$$\widetilde L=\widetilde K\langle \Psi^{(+)}_{(1,\lambda)},\Psi^{(+)}_{(2,\lambda)}\rangle=K\langle
\Psi^{(+)}_{(1,\lambda)},\Psi^{(+)}_{(2,\lambda)},\partial_x(\ln
\Psi^{(-)}_{\lambda_1})\rangle$$ $$=K\langle
\Psi^{(-)}_{(1,\lambda)},\Psi^{(-)}_{(2,\lambda)},\partial_x(\ln
\Psi^{(-)}_{\lambda_1})\rangle=\widetilde K\langle
\Psi^{(-)}_{(1,\lambda)},\Psi^{(-)}_{(2,\lambda)}\rangle,$$ for
$\lambda=\lambda_1$ and for $\lambda\neq\lambda_1$. Since
$\partial_x(\ln \Psi^{(-)}_{\lambda_1})$ is algebraic over $K$,
then
$$(\mathrm{DGal}(L_\lambda/K))^0=(\mathrm{DGal}(\widetilde L_\lambda/K))^0,\quad \mathrm{DGal}(L_\lambda/K)
=\mathrm{DGal}(\widetilde L_\lambda/\widetilde K),$$
which means that $\mathrm{DT}$ is a virtually strong and isogaloisian transformation.\\
In the case $\partial_x(\ln \Psi^{(-)}_{\lambda_1})\in K$, then
$\widetilde K=K$ and $\widetilde L=L$, which means that
$\mathrm{DT}$ is a strong isogaloisian transformation.
\end{proof}

\begin{proposition}\label{dareige} Consider  Schr\"odinger operators $\mathfrak{L}_\lambda:=H_--\lambda$ and $\widetilde{
\mathfrak{L}}_\lambda:=H_+-\lambda$ such that
$\mathrm{DT}(H_--\lambda)=H_+-\lambda$. The eigenrings of
$\mathfrak{L}_\lambda$ and $\widetilde{\mathfrak{L}}_\lambda$ are
isomorphic.

\end{proposition}

\begin{proof}
Let $\mathcal E(\mathfrak L_\lambda)$ and $\mathcal
E(\widetilde{\mathfrak L}_\lambda)$ denote the eigenrings of
$\mathfrak{L}_\lambda$ and $\widetilde{\mathfrak{L}}_\lambda$
respectively. By proposition \ref{darbiso} the connected identity
component of the Galois group is preserved by Darboux
transformation.
% and for instance the eigenrings is preserved by Darboux transformation. %%JA that's what we're proving
Let $T\in\mathcal E
(\mathfrak L_\lambda)$. Let $\mathrm{Sol}(\mathfrak L_\lambda)$ and
$\mathrm{Sol}(\widetilde{\mathfrak{L}}_\lambda)$ be the solution
spaces for $\mathfrak L_\lambda\Psi^{(-)}=0$ and
$\widetilde{\mathfrak L}_\lambda\Psi^{(-)}=0$ respectively. To
transform $\mathcal E (\mathfrak L_\lambda)$ into $\mathcal E
(\widetilde{\mathfrak{L}}_\lambda)$, we follow the commuting diagram:

 $$\xymatrix{\mathrm{Sol}(\mathfrak L_\lambda) \ar[r]^-{T} & \mathrm{Sol}(\mathfrak L_\lambda) \\
   \mathrm{Sol}(\widetilde{\mathfrak L}_\lambda)  \ar[r]_{\widetilde T} \ar[u]^{A^\dagger} & \ar@{<-}[u]_{A} \mathrm{Sol}(\widetilde{\mathfrak L}_\lambda) }
   \qquad \begin{array}{c}\\ \\ \\
  \quad\hbox{\rm so }\quad  \widetilde{T}=ATA^\dagger{\rm mod}(\widetilde{\mathfrak L}_\lambda)\in\mathcal E(\widetilde{\mathfrak L}_\lambda),\end{array}$$
where $A^\dagger$ and $A$ are the raising and lowering operators.
\end{proof}

\begin{example}
Consider the Schr\"odinger equation $\mathcal L_\lambda$ with
potential $V=V_-=0$, which means that $\Lambda=\mathbb{C}$. If we
choose $\lambda_1=0$ and as particular solution
$\Psi_{0}^{(-)}=x$, then for $\lambda\neq 0$ the general solution
is given by
$$\Psi^{(-)}_{\lambda}=c_1e^{\sqrt{-\lambda}x}+c_2e^{-\sqrt{-\lambda}x}.$$
Applying the Darboux transformation $\mathrm{DT}$, we have
$\mathrm{DT}(\mathcal L_\lambda)=\widetilde{\mathcal L}_\lambda$,
where
$$\mathrm{DT}(V_-)=V_+={2\over x^2}$$ and for $\lambda\neq
0$
$$\mathrm{DT}(\Psi^{(-)}_{\lambda})=\Psi^{(+)}_{\lambda}={c_1(\sqrt{-\lambda}x-1)e^{\sqrt{-\lambda}x}\over x}-{c_2(\sqrt{-\lambda}x+1)e^{-\sqrt{-\lambda}x}\over x}.$$
We can see that $\widetilde{K}=K=\mathbb{C}(x)$ for all
$\lambda\in \Lambda$ and the Picard-Vessiot extensions are given
by $L_0=\widetilde{L}_0=\mathbb{C}(x)$,
$L_\lambda=\widetilde{L}_\lambda=\mathbb{C}(x,e^{\sqrt{\lambda}x})$
for $\lambda\in\mathbb{C}^*$. It follows that  
$\mathrm{DGal}(L_0/K)=\mathrm{DGal}(\widetilde{L}_0/K)=e$; for
$\lambda\neq0$, we have
$\mathrm{DGal}(L_\lambda/K)=\mathrm{DGal}(\widetilde{L}_\lambda/K)=\mathbb{G}_m$.
The eigenrings of the operators $\mathfrak L_0$ and $\mathfrak
L_\lambda$ are given by
$$\mathcal E(\mathfrak L_0)=\mathrm{Vect}\left(1, \partial_x,x\partial_x-1,x^2\partial_x-x\right),$$
$$\mathcal E(\widetilde{\mathfrak
L}_0)=\mathrm{Vect}\left(1,x\partial_x-1,x^4\partial_x-2x^3,{\partial_x\over
x^2}+{1\over x^3}\right)$$ and for $ \lambda\neq 0$
$$\mathcal E(\mathfrak L_{\lambda})=\mathrm{Vect}\left(1,\partial_x\right), \quad \mathcal
E(\widetilde{\mathfrak
L}_{\lambda})=\mathrm{Vect}\left(1,-\left(\lambda +{1\over x^2}
\right)\partial_x-{1\over x^3}\right),$$ where $\mathcal
L_{\lambda}:=\mathfrak L_{\lambda}\Psi^{(-)}=0$ and
$\widetilde{\mathcal L}_{\lambda}:=\widetilde{\mathfrak
L}_{\lambda}\Psi^{(+)}=0$.

\end{example}
\medskip

Applying iteratively the Darboux transformation, theorem
\ref{darth}, and by propositions \ref{darbiso}, \ref{dareige}
we have the following results.\\

\begin{proposition}[Galoisian version of DT$_n$]\label{darbit1}
We continue with the notations of theorem \ref{darth}. 
We set $V_0=V_-$, $H^{(0)}=H_-$, $\Psi^{(0)}=\Psi^{(-)}$
We still assume that the algebraic spectrum $\Lambda$ is non-empty.
We consider the operators $\mathcal{L}^{(n)}_\lambda$ (given by $H^{(n)}\Psi^{(n)}=\lambda\Psi^{(n)}$) obtained by successive Darboux transformations.
We let $K_n$ be the coefficient field of  $\mathcal{L}^{(n)}_\lambda$ (i.e $V_n\in{K}_n$) and $K=K_0$.
Let $L_\lambda^{(n)}$ denote the Picard-Vessiot
extension of $\mathcal{L}^{(n)}_\lambda$. 
%Let $\mathcal{L}_\lambda^{(n+1)}$ be given by $H^{(n+1)}\Psi^{(n+1)}=\lambda\Psi^{(n+1)},$ $V_{n+1}\in K_{n+1}$.
%\\
The Darboux transformation $\mathrm{DT}_n$  is such that
$\mathcal{L}_\lambda^{(n)}\mapsto\mathcal{L}_\lambda^{(n+1)}$,
$V_n\mapsto V_{n+1}$,
 $\Psi^{(n)}_{\lambda}\mapsto\Psi^{(n+1)}_\lambda$ and $L_\lambda^{(n+1)}$ the Picard-Vessiot
extension of $\mathcal{L}^{(n+1)}_\lambda$. 
\\
Then the following statements hold:
 \begin{description}
\item[i)] $\mathrm{DT}_n(V_-)=\mathrm{DT}(V_n)=V_{n+1}=V_n-2\partial_x^2\left(\ln
\Psi_{\lambda}^{(n)}\right)=V_--2\sum_{k=0}^{n}\partial_x^2\left(\ln
\Psi_{\lambda_k}^{(k)}\right)$, where $\Psi_{\lambda_k}^{(k)}$ is
a particular solution for $\lambda=\lambda_k$, $k=0,\ldots, n$. In
particular, if $\lambda_n=\lambda_0$ and $\Lambda=\mathbb{C}$,
then there exists $\Psi^{(n)}_{\lambda_n}$ such that $V_{n}\neq
V_{n-2}$, with $n\geq 2$.
\item[ii)] $\mathrm{DT}(\Psi_{\lambda}^{(n)})=\mathrm{DT}_n(\Psi^{(-)}_{\lambda})=\Psi^{(n+1)}_{\lambda}=
\partial_x\Psi^{(n)}_{\lambda}-\Psi^{(n)}_{\lambda}{\partial_x\Psi_{\lambda_n}^{(n)}\over
\Psi_{\lambda_n}^{(n)}}={W(\Psi^{(n)}_{\lambda_n},\Psi^{(n)}_{\lambda})\over
W(\Psi^{(n)}_{\lambda_n})} $ where $\Psi^{(n)}_{\lambda}$ is a
general solution for $\lambda\in\Lambda\setminus\{\lambda_n\}$ of
$\mathcal{L}^{(n)}_\lambda$.

\item[iii)] $K_{n+1}=K_n\left\langle\partial_x(\ln
\Psi^{(n)}_{\lambda_n}) \right\rangle$.

\item[iv)] $\mathrm{DT}_n$ is isogaloisian and virtually strongly isogaloisian. Furthermore, if $\partial_x(\ln\Psi^{(n)}_{\lambda_n})\in K_n$ then $\mathrm{DT}_n$
is strongly isogaloisian.
\item[v)] The eigenrings of $H^{(n)}-\lambda$ and $H^{(n+1)}-\lambda$
are isomorphic.
\end{description}
\end{proposition}
\begin{proof} By induction on theorem \ref{darth} we obtain i) and
ii). By induction on proposition \ref{darbiso} we obtain iii) and
iv). By induction on proposition \ref{dareige} we obtain v).

\end{proof}
\medskip

\begin{example}
Starting with $V=0$, the following potentials can be obtained
using Darboux iteration $\mathrm{DT}_n$ (see \cite{beev,blbose}).
$$
I)\,\,V_n={n(n-1)b^2\over(bx+c)^2},\quad
II)\,\,V_n={m^2n(n-1)(b^2-a^2)\over (a\cosh(mx)+b\sinh(mx))^2},
$$
$$
III)\,\,V_n={-4abm^2n(n-1)\over (ae^{mx}+be^{-mx})^2},\quad
IV)\,\,V_n={m^2n(n-1)(b^2+a^2)\over (a\cos(mx)+b\sin(mx))^2}.
$$ In particular for the rational potential given in \textit{I}), we have $K=K_n=\mathbb{C}(x)$
and for $\lambda_n=\lambda=0$, we have $$\Psi_0^{(n)}={c_1\over
(bx+c)^n}+c_2(bx+c)^{n+1}, \textrm{   so that   }
\mathrm{DGal}(L_0/K)=\mathrm{DGal}(L_0^{(n)}/K)=e,$$
$$\mathcal E(H^{(n)})=\mathrm{Vect}\left(1,x\partial_x-1,x^{2n+2}\partial_x-(n+1)x^{2n+1},{\partial_x\over
x^{2n} }+{n\over x^{2n+1}}\right),$$ whilst for $\lambda\neq 0$
and $\lambda_n=0$, the general solution $\Psi_\lambda^{(n)}$ is
given by
$$\Psi_\lambda^{(n)}(x)=c_1f_n(x,\lambda)h_n(\sin(\sqrt{\lambda}x)+c_2g_n(x,\lambda)j_n(\cos(\sqrt{\lambda}x),$$
where $f_n,g_n,h_n,j_n\in\mathbb{C}(x)$, so that
$$ \mathrm{DGal}(L_\lambda/K)=\mathrm{DGal}(L^{(n)}_\lambda/K)=\mathbb{G}_m,$$ and
$$\dim_{\mathbb{C}}\mathcal E(H-\lambda)=
\dim_{\mathbb{C}}\mathcal E(H^{(n)}-\lambda)=2.$$

\end{example}

\subsection{Crum Iteration}

\begin{proposition}[Galoisian version of CI$_n$]\label{crumit1}
Consider  $\mathcal L_\lambda$ given by $H\Psi=\lambda\Psi$,
$H=-\partial_x^2+V$, $V\in K$, such that
$\mathrm{Card}(\Lambda)>n$ for a fixed $n\in\mathbb{Z}_+$. Let
$\mathcal L_\lambda^{(n)}$ be given by
$H^{(n)}\Psi^{(n)}=\lambda\Psi^{(n)}$, where
$H^{(n)}=\partial_x^2+V_n$, $V_n\in K_n$. Let $\mathrm{CI}_n$ be
the transformation such that
$\mathcal{L}_\lambda\mapsto\mathcal{L}_\lambda^{(n)}$, $V\mapsto
V_n$, $(\Psi_{\lambda_1},\ldots,
\Psi_{\lambda_n},\Psi_\lambda)\mapsto\Psi_\lambda^{(n)}$, where
for $k=1,\ldots, n$ and the equation $\mathcal L_\lambda$,
 the function $\Psi_\lambda$ is   the general solution  for
$\lambda\neq\lambda_k$ and $\Psi_{\lambda_k}$ is a particular
solution for $\lambda=\lambda_k$. Then the following statements
 holds:
 \begin{description}
\item[i)] $\mathrm{CI}_n(\mathcal{L}_\lambda)=\mathcal L_\lambda^{(n)}$
where $\mathrm{CI}_n(V)=V_n=V-2\partial_x^2\left(\ln
W(\Psi_{\lambda_1},\ldots,\Psi_{\lambda_n})\right)$ and
$$\mathrm{CI}_n(\Psi_\lambda)=\Psi_\lambda^{(n)}={W(\Psi_{\lambda_1},\ldots,\Psi_{\lambda_n},\Psi_\lambda)\over
W(\Psi_{\lambda_1},\ldots,\Psi_{\lambda_n})},$$ where
$\Psi_\lambda^{(n)}$ is the general solution of $\mathcal
L_\lambda^{(n)}$.

\item[ii)] $K_{n}=K\langle \partial_x\left(\ln
W(\Psi_{\lambda_1},\ldots,\Psi_{\lambda_n})\right)\rangle$.
\item[iii)] $\mathrm{CI}_n$ is isogaloisian and virtually strongly isogaloisian. Furthermore, if $$\partial_x\left(\ln
W(\Psi_{\lambda_1},\ldots,\Psi_{\lambda_n})\right)\in K_n,$$ then
$\mathrm{CI}_n$ is strongly isogaloisian.
\item[iv)] The eigenrings of $H-\lambda$ and
$H^{(n)}-\lambda$ are isomorphic.
\end{description}
\end{proposition}
\begin{proof} By induction on theorem \ref{darth} we obtain i). By induction on proposition \ref{darbiso} we obtain ii) and
iii). By induction on proposition \ref{dareige} we obtain iv).

\end{proof}
\medskip

\begin{example}
To illustrate the Crum iteration with rational potentials, we
consider $V=\frac2{x^2}$. The general solution of $\mathcal
L_\lambda:= H\Psi=\lambda\Psi$ is $${c_1e^{kx}(kx-1)\over
x}+{c_2e^{-kx}(kx+1)\over x},\quad \lambda=-k^2,$$ the
eigenfunctions for $\lambda_1=-1$, and $\lambda_2=-4$, are
respectively given by $$\Psi_{-1}={e^{-x}(x+1)\over x},\quad
\Psi_{-4}={e^{-2x}(2x+1)\over 2x}.$$ Thus, we obtain
$$\mathrm{CI}_2(V)=V_2=\frac8{(2x+3)^2}$$ and the
general solution of $\mathcal L
_\lambda^{(2)}:=H^{(2)}\Psi^{(2)}=\lambda\Psi^{(2)}$ is
$$\mathrm{CI}_2(\Psi_\lambda)=\Psi_\lambda^{(2)}={\frac{c_1\left( k \left( 2\,x+3
 \right) -2 \right) {e^{kx}}}{2\,x+3}}+{\frac {c_2 \left( 2+k
 \left( 2\,x+3 \right)  \right) {e^{-kx}}}{4\,x+6}},\quad
 \lambda=-k^2.$$ The differential Galois groups and eigenrings are
 given by:
$$\mathrm{DGal}(L_0/K)=\mathrm{DGal}(L^{(2)}_0/K)=e,\quad \dim_{\mathbb{C}}\mathcal E(H)=\dim_{\mathbb{C}}\mathcal E(H^{(2)})=4,$$ and for
$\lambda\neq 0$
$$\mathrm{DGal}(L_\lambda/K)=\mathrm{DGal}(L^{(2)}_\lambda/K)=\mathbb{G}_m,\quad \dim_{\mathbb{C}}\mathcal
E(H-\lambda)=\dim_{\mathbb{C}}\mathcal E(H^{(2)}-\lambda)=2.$$
\end{example}
\begin{proposition}\label{sspp}
The supersymmetric partner potentials $V_\pm$ are rational
functions if and only if the superpotential $W$ is a rational
function.
\end{proposition}

\begin{proof} The supersymmetric partner potentials $V_\pm$ are
written as $V_\pm=W^2\pm\partial_xW$. We start considering the
superpotential $W\in\mathbb{C}(x)$, so trivially we have 
$V_\pm\in\mathbb{C}(x)$. Now assuming that $V_\pm\in\mathbb{C}(x)$
we have  $\partial_xW\in\mathbb{C}(x)$ and
$W^2\in\mathbb{C}(x)$, which implies that
$W\partial_xW\in\mathbb{C}(x)$ and therefore $W\in\mathbb{C}(x)$.

\end{proof}
\medskip

\begin{corollary} The superpotential $W\in\mathbb{C}(x)$ if and only if $\mathrm{DT}$ is strong isogaloisian.
\end{corollary}

\begin{proof} Assume that the superpotential $W\in\mathbb{C}(x)$.
Thus, by proposition \ref{darbiso}, $\mathrm{DT}$ is strong
isogaloisian. Now, assume that $\mathrm{DT}$ is strong
isogaloisian. Thus, $V_\pm\in\mathbb{C}(x)$ and by proposition
\ref{sspp} we have  $W\in\mathbb{C}(x)$.

\end{proof}
\subsection{Shape Invariance}

The following definition is a partial Galoisian adaptation of the
original definition given in \cite{ge} ($K=\mathbb{C}(x)$). The
complete Galoisian adaptation is given when $K$ is any
differential field.\\

\begin{definition}[Rational Shape Invariant Potentials] Assume
$V_\pm(x;\mu)\in\mathbb{C}(x;\mu)$, where $\mu$ is a family of
parameters. The potential $V=V_-\in\mathbb{C}(x)$ is called a
rational shape invariant potential with respect to $\mu$ and
$E=E_n$ %being $n\in \mathbb{Z}_+$, 
if there exists $f$ such that
$$V_+(x;a_0)=V_-(x;a_1)+R(a_1),\quad a_1=f(a_0),\quad E_n=\sum_{k=2}^{n+1}R(a_k),\quad E_0=0.$$
\end{definition}
\medskip

\begin{remark}\label{remshape} We propose the following steps to check whether $V\in\mathbb{C}(x)$ is shape
invariant.

\begin{description}
\item[Step 1.] Introduce parameters in $W(x)$ to obtain $W(x;\mu)$,
write $V_\pm(x;\mu)=W^2(x;\mu)\pm\partial_xW(x;\mu)$, and replace
$\mu$ by $a_0$ and $a_1$.
\item [Step 2.] Obtain polynomials $\mathcal P\in\mathbb{C}[x;a_0,a_1]$ and $\mathcal
Q\in\mathbb{C}[x;a_0,a_1]$ such that
$$\partial_x(V_+(x;a_0)-V_-(x;a_1))={\mathcal P(x;a_0,a_1)\over \mathcal Q(x;a_0,a_1)}.$$
\item [Step 3.] Set $\mathcal P(x;a_0,a_1)\equiv 0$, as polynomial in $x$, to obtain $a_1$ in function of $a_0$, i.e.,
$a_1=f(a_0)$. Also obtain $R(a_1)=V_+(x;a_0)-V_-(x;a_1)$ and
verify that exists $k\in\mathbb{Z}^+$ such that
$R(a_1)+\cdots+R(a_k)\neq 0$.
\end{description}

\end{remark}
\medskip

\begin{example} Consider the superpotential of the three
dimensional harmonic oscillator $W(r;\ell)=r-\frac{\ell+1}{r}$. By
step 1, the supersymmetric partner potentials are
$$V_-(r;\ell)=r^2+\frac{\ell(\ell+1)}{r^2}-2\ell-3,\quad
V_+(r;\ell)=r^2+\frac{(\ell+1)(\ell+2)}{r^2}-2\ell-1.$$ By step 2,
we have $\partial_r(V_+(r;a_0)-V_-(r;a_1))=- 2{a_0^2 + 3a_0 -
a_1^2 - a_1 + 2\over r^3}$. By step 3, $(a_0+1)(a_0+2)
=a_1(a_1+1)$, so that $a_1=f(a_0)=a_0+1$, $a_n=f(a_{n-1})=a_0+n$,
$R(a_1)=2$. Thus, we obtain the energy levels $E_n=2n$ and the
wave functions $\Psi^{(-)}_n(r;\ell)=A^\dagger(r;\ell)\cdots
A^\dagger(r;\ell+n-1)\Psi^{(-)}_0(r;\ell+n)$, compare with \cite{duka}.\\

\end{example}
By theorem \ref{darth} and propositions \ref{darbiso},
\ref{dareige} and \ref{sspp} we have the following result.

\begin{theorem}\label{thsip} Consider $\mathcal
L_n:=H\Psi^{(-)}=E_n\Psi^{(-)}$ with Picard-Vessiot extension
$L_n$, where $n\in\mathbb{Z}_+$. If $V=V_-\in\mathbb{C}(x)$ is a
shape invariant potential with respect to $E=E_n$, then
$$\mathrm{DGal}(L_{n+1}/K)=\mathrm{DGal}(L_{n}/K),\quad \mathcal E(H-E_{n+1})\simeq \mathcal E(H-E_n),\quad n>0.$$
\end{theorem}
\medskip

\begin{remark} 
As $W\in\mathbb{C}(x)$, a differential automorphism $\sigma$ commutes
with the raising and lowering operators $A$ and $A^\dagger$.
Furthermore the wave functions $\Psi^{(-)}_n$
can be written as $\Psi^{(-)}_n=P_nf_n\Psi^{(-)}_0$, where $P_n$
is a polynomial of degree $n$ in $x$ and $f_n$ is a sequence of
functions with $f_0(x)=1$ such as was shown in the case of
Harmonic oscillators and Coulomb potentials.
\end{remark}

%%%%%%%%%%%%%%%%%%%%%%%%%%

\section{Schr\"odinger Equations with Rational Potentials}\label{susyrapo} 
Along this section we take  $K=\mathbb{C}(x)$ to be our
coefficient field. So all our differential operators will have coefficients in $K=\mathbb{C}(x)$.

\subsection{Polynomial Potentials}
 We start by considering the Schr\"odinger equation
\eqref{equ1} with polynomial potentials, i.e.
$V\in\mathbb{C}[x]$, see \cite{bo,va1}. For simplicity and without
loss of generality, we consider monic polynomials  (if the polynomial is $c_nx^n+\ldots+c_1x+ c_0$,
apply the transformation $x\mapsto \sqrt[n+2]{c_n} x$).

When $V$ is a polynomial of odd degree, is well known that the
differential Galois group of such a Schr\"odinger equation
\eqref{equ1} is
$\mathrm{SL}(2,\mathbb{C})$, see \cite{ko}.\\

We present here the complete result for the Schr\"odinger equation
\eqref{equ1} with non-constant polynomial potential (Theorem
\ref{polynint}), see also \cite[\S 2]{acbl}. The following lemma
is useful for our purposes.\\

\begin{lemma}[Completing Squares, \cite{acbl}]\label{cosq}
Every even degree monic polynomial of  can be written in one only
way completing squares, that is,
\begin{equation}\label{square}
Q_{2n}(x)=x^{2n}+\sum_{k=0}^{2n-1}q_kx^k=\left(x^n+\sum_{k=0}^{n-1}a_kx^k\right)^2+\sum_{k=0}^{n-1}b_kx^k,
\end{equation}
where
$$a_{n-1}={q_{2n-1}\over 2},\quad a_{n-2}={q_{2n-2}-a^2_{n-1}\over 2},\quad a_{n-3}={q_{2n-3}-2a_{n-1}a_{n-2}\over 2},\cdots,$$

$$a_0={q_n-2a_1a_{n-1}-2a_2a_{n-2}-\cdots\over 2},\quad
b_0=q_0-a_0^2,\quad b_1=q_1-2a_0a_1,\quad \cdots,$$
$$b_{n-1}=q_{n-1}-2a_0a_{n-1}-2a_1a_{n-2}-\cdots.$$
\end{lemma}
\medskip
\begin{proof} See lemma 2.4 in \cite[p. 275]{acbl}.

\end{proof}
\medskip

We remark that $V(x)$ as in equation \eqref{square} can be written
in terms of the superpotential $W(x)$, i.e.,
$V(x)=W^2(x)-\partial_xW(x)$, when
$$nx^{n-1}+\sum_{k=1}^{n-1}ka_kx^{k-1}=-\sum_{k=0}^{n-1}b_kx^k$$
and $W(x)$ is given by $$x^n+\sum_{k=0}^{n-1}a_kx^k.$$

The following theorem also can be found in \cite[\S 2]{acbl}, see
also \cite{acbl2}. Here we present a quantum mechanics adapted
version.\\

\begin{theorem}[Polynomial potentials and Galois groups, \cite{acbl}]\label{polynint}
Let us consider the Schr\"odinger equation \eqref{equ1}, with
$V(x)\in\mathbb{C}[x]$ a polynomial of degree $k>0$. Then,  its
differential Galois group $\mathrm{DGal}(L_{\lambda}/K)$ falls in
one of the following cases:
\begin{enumerate}
\item $\mathrm{DGal}(L_{\lambda}/K)=\mathrm{SL}(2,\mathbb{C})$,
\item $\mathrm{DGal}(L_{\lambda}/K)=\mathbb{B}$,
\end{enumerate}
and the eigenring  of $H-\lambda$ is trivial, i.e., $\mathcal
E(H-\lambda)=\mathrm{Vect}(1)$. Furthermore,
$\mathrm{DGal}(L_{\lambda}/K)=\mathbb{B}$ if and only if the
following conditions hold:
\begin{enumerate}
\item $V(x)-\lambda$ is a polynomial of degree $k=2n$ writing in the form of equation
\eqref{square}.
\item $b_{n-1}-n$ or $-b_{n-1}-n$ is a positive even number $2m$, $m\in\mathbb{Z}_+$.
\item There exist a monic polynomial $P_m$ of degree $m$,
satisfying {\small \begin{displaymath}
\partial_x^2P_m + 2\left(x^n+\sum_{k=0}^{n-1}a_kx^k\right)\partial_xP_m +
\left(nx^{n-1}+\sum_{k=0}^{n-2}(k+1)a_{k+1}x^{k} -
\sum_{k=0}^{n-1}b_kx^k\right)P_m = 0,
\end{displaymath}}
or {\small \begin{displaymath}
\partial_x^2P_m - 2\left(x^n+\sum_{k=0}^{n-1}a_kx^k\right)\partial_xP_m -
\left(nx^{n-1}+\sum_{k=0}^{n-2}(k+1)a_{k+1}x^{k} +
\sum_{k=0}^{n-1}b_kx^k\right)P_m = 0.
\end{displaymath}}

\end{enumerate}
In such cases, the only possibilities for eigenfunctions with
rational superpotentials are given by
$$
\Psi_\lambda=P_me^{f(x)},\quad \textit{or}\quad
\Psi_\lambda=P_me^{-f(x)},\quad \textit {   where
}f(x)={x^{n+1}\over n+1}+\sum_{k=0}^{n-1}{a_kx^{k+1}\over k+1}.
$$
\end{theorem}
\medskip
\begin{proof} See theorem 2.5 in \cite[p. 276]{acbl}.

\end{proof}
\medskip
An easy consequence of the above theorem is the following.\\

\begin{corollary} Assume that $V(x)$ is an algebraically solvable polynomial
potential. Then $V(x)$ is of degree $2$.
\end{corollary}
\medskip
\begin{proof} Writing $V(x)-\lambda$ in the form of equation
\eqref{square} we see that $b_{n-1}-n=2m$ or $-b_{n-1}-n=2m$,
where $m\in\mathbb{Z}_+$. Thus, the integrability of the
Schr\"odinger equation  with $\mathrm{Card}(\Lambda)>1$ is
obtained when $b_{n-1}$ is constant, so $n=1$.
\end{proof}
\medskip

\begin{remark}\label{rembs} Given a polynomial potential $V(x)$
 such that $\mathrm{Spec}_p(H)\cap\Lambda\neq \emptyset$, we can obtain bound
states and normalized wave functions if and only if the potential
$V(x)$ is a polynomial of degree $4n+2$. Furthermore, one
integrability condition of $H\Psi=\lambda\Psi$ for
$\lambda\in\Lambda$ is that $b_{2n}$ must be an odd integer. In
particular, if the potential
$$V(x)=x^{4n+2n}+\mu x^{2n},\quad n>0$$ is
a quasi-exactly solvable, then $\mu$ is an odd integer. For this
kind of potentials, we obtain bound states only when $\mu$ is a
negative odd integer.\\

On another hand, the non-constant polynomial potentials $V(x)$ of
degree $4n$ are associated to non-hermitian Hamiltonians and
$\mathcal{PT}$ invariance which are not considered here, see
\cite{bebo}. Furthermore, one integrability condition of
$H\Psi=\lambda\Psi$ for $\lambda\in\Lambda$ is that $b_{2n-1}$
must be an even integer. In particular, if the Schr\"odinger
equation
$$H\Psi=\lambda\Psi, \quad V(x)=x^{4n}+\mu x^{2n-1}, \quad \lambda\in\Lambda$$ is
integrable, then $\mu$ is an even integer.
\end{remark}
\medskip

 We present the following examples to illustrate the
previous
theorem and remark.\\

\noindent \textbf{Weber's Equation and Harmonic Oscillator.} The
Schr\"odinger equation with potential $V(x)=x^2+q_1x+q_0$
corresponds to the Rehm form of the Weber's equation. By lemma \ref{cosq} we have
$$V(x)-\lambda=(x+a_0)^2+b_0,\quad a_0=q_1/2, \quad
b_0=q_0-q_1^2/4-\lambda.$$ So that we obtain $\pm b_0-1=2m$, where
$m\in\mathbb{Z}_+$. If $b_0$ is an odd integer, then
$$\mathrm{DGal}(L_{\lambda}/K)=\mathbb{B}, \, \mathcal E(H-\lambda)=\mathrm{Vect}(1),\,
\lambda\in\Lambda=\{\pm(2m+1)+q_0-q_1^2/4:m\in\mathbb{Z}_+\}$$ and
the set of eigenfunctions is either
$$\Psi_\lambda=P_me^{\frac12(x^2+q_1x)},\textrm{ or, }
\Psi_\lambda=P_me^{-\frac12(x^2+q_1x)}.$$ In the second case we
have bound state wave function and
$\mathrm{Spec}_p(H)\cap\Lambda=\mathrm{Spec}_p(H)=\{E_m=2m+1+q_0-q_1^2/4:m\in\mathbb{Z}_+\}$,
which is infinite. The polynomials $P_m$ are related with the
Hermite polynomials $H_m$, \cite{ch,is,niuv}.\\

In particular we have the harmonic oscillator potential
$$V(x)={1\over 4}\omega^2x^2-{\omega\over 2},$$ and the Schr\"odinger equation
$H\Psi=E\Psi$. Through the change of independent variable
$x\mapsto\sqrt{2\over \omega}x$ we obtain $V(x)=x^2-1$ and
$\lambda={2\over\omega}E$, that is, $q_1=0$ and $q_0=-1$. 
It follows that  $\Lambda=\{\pm( 2m+1)-1:m\in\mathbb{Z}_+\}$ and the set of
eigenfunctions is either
$$\Psi_\lambda=P_me^{\frac12x^2},\textit{ or, }
\Psi_\lambda=P_me^{-\frac12x^2},$$ where as below,
$\mathrm{DGal}(L_{\lambda}/K)=\mathbb{B}$ and
$\mathcal{E}(H-\lambda)=\{1\}$ for all $\lambda\in\Lambda$. In the
second case we have bound state wave function,
$\mathrm{Spec}_p(H)\cap\Lambda=\mathrm{Spec}_p(H)=\Lambda_+=\{2m:m\in\mathbb{Z}_+\}$
and $P_m=H_m$. The wave functions of $H\Psi=E\Psi$ for the
harmonic oscillator potential are given by

$$\Psi_m=H_m\left(\sqrt{2\over \omega}x\right)\Psi_0,\quad \Psi_0=e^{-{\omega\over4}x^2},\quad E=E_m=m\omega.$$
\medskip

\noindent \textbf{Quartic and Sextic Anharmonic Oscillator.} The
Schr\"odinger equation with potential
$V(x)=x^4+q_3x^3+q_2x^2+q_1x+q_0$ can be obtained through
transformations of confluent Heun's equation, though we will not adress this here.  
By lemma \ref{cosq} we have
$$V(x)-\lambda=(x^2+a_1x+a_0)^2+b_1x+b_0,$$ where $a_1=q_3/2$, $a_0=q_2/2-a_1^2/2$, $b_1=q_1-2a_0a_1$ and $b_0=q_0-a_0^2-\lambda.$
 So we obtain $\pm b_1-2=2m$, where
$m\in\mathbb{Z}_+$. If $\Lambda\neq\emptyset$, then $b_1$ is an
even integer, $P_m$ satisfy the relation \eqref{recu1} and
$\mathrm{DGal}(L_{\lambda}/K)=\mathbb{B}$ for all
$\lambda\in\Lambda$. The set of eigenfunctions is either
$$\Psi_\lambda=P_me^{\frac{x^3}3+\frac{a_1x^2}2+a_0x},\textit{ or, }
\Psi_\lambda=P_me^{-\left(\frac{x^3}3+\frac{a_1x^2}2+a_0x\right)},$$
where $\lambda$ and $m$ are related, which means that $\Lambda$ is
finite, i.e., the potential is algebraically quasi-solvable. In
particular for $q_3=2\imath l$, $q_2=l^2-2k$, $q_1=2\imath(lk-J)$
and $q_0=0$, we have the quartic anharmonic oscillator potential,
which can be found in \cite{bebo}.\\

Now, considering the potentials $V(x,\mu)=x^4+4x^3+2x^2-\mu x$,
again by lemma \ref{cosq} we have 
$$V(x,\mu)-\lambda=(x^2+2x-1)^2+(4-\mu)x-1-\lambda,$$ so that $\pm(4-\mu)-2=2n$, where $n\in\mathbb{Z}_+$ and in consequence $\mu\in2\mathbb{Z}$. Such $\mu$ can be either
$\mu=2-2n$ or $\mu=2n+6$, where $n\in\mathbb{Z}_+$. By theorem
\ref{polynint}, there exists a monic polynomial $P_n$ satisfying
respectively
$$\partial_x^2P_n+(2x^2+4x-2)\partial_xP_n+((\mu-2)x+3+\lambda)P_n=0,\quad \mu=2-2n,\quad\textrm{or}$$
$$\partial_x^2P_n-(2x^2+4x-2)\partial_xP_n+((\mu-6)x-1+\lambda)P_n=0,\quad\mu=2n+6$$
 for $\Lambda\neq\emptyset$. This algebraic relation between the coefficients of $P_n$, $\mu$ and $\lambda$ give us the set
 $\Lambda$ in the following way:

 \begin{enumerate}
 \item Write
 $P_n=x^n+c_{n-1}x^{n-1}+\ldots+c_0$, where $c_i$ are unknown.
 \item Pick $\mu$ and replace $P_n$ in the algebraic relation \eqref{recu1} to obtain a polynomial of degree $n$ with $n+1$ undetermined coefficients
involving $c_0,\ldots c_{n-1}$ and $\lambda$. Each of such
coefficients must be zero.\item  The term $n+1$ is linear in
$\lambda$ and $c_{n-1}$, thus we write $c_{n-1}$ in terms of
$\lambda$. After of the elimination of the term $n+1$, we replace
$c_{n-1}$ in the term $n$ to obtain a quadratic polynomial in
$\lambda$ and so on until we reach the constant term which is a
polynomial $Q_{n+1}(\lambda)$ of degree $n+1$ in $\lambda$. It follows that   $\Lambda=\{\lambda:Q_{n+1}(\lambda)=0\}$ and
$c_0,\ldots,c_{n-1}$ are determined for each value of $\lambda$.
\end{enumerate}

For $\mu=2n+6$, we have:

$$\begin{array}{llll}
n=0,& V(x,6),&P_0=1,& \Lambda=\{1\}\\
n=1,& V(x,8),&P_1=x+1\mp\sqrt2,& \Lambda=\{3\pm2\sqrt2\}\\
\vdots& & &
 \end{array}$$ and the set of eigenfunctions is
$$\Psi_{\lambda,\mu}=P_ne^{-\frac13x^3-x^2+x}.$$

In the same way, we can obtain $\Lambda$, $P_n$ and
$\Psi_{\lambda,\mu}$ for $\mu=2-2n$. However, we do not have bound
states, $\mathrm{Spec}_p(H)\cap\Lambda=\emptyset$,
$\mathrm{DGal}(L_{\lambda}/K)=\mathbb{B}$ and
$\mathcal{E}(H-\lambda)=\mathrm{Vect}(1)$ for all $\lambda\in\Lambda$.\\

The well known \textit{sextic anharmonic oscillator}
$x^6+q_5x^5+\cdots+q_1x+q_0$ can be treated in a similar way,
obtaining bound states wave functions and the Bender-Dunne
orthogonal polynomials, which corresponds to $Q_{n+1}(\lambda)$,
i.e. we also recover the results of \cite{bedu,give,sahaci}. 
The Schr\"odinger equation with this potential, under suitable
transformations, also falls in the class of confluent Heun equations.\\

%\begin{comment}
\subsection{Rational Potentials and Kovacic's Algorithm}
Here we apply Kovacic's algorithm (see appendix A) to solve the
Schr\"odinger
equation with rational (non-polynomial) potentials. \\

 \textbf{Three dimensional harmonic oscillator
potential:}
$$V(r)={1\over 4}\omega^2r^2+{\ell(\ell+1)\over r^2}-\left(\ell+{3\over 2}\right)\omega,\quad \ell\in\mathbb{Z}.$$ 
In this case, the Schr\"odinger equation
is written as
$$\partial_r^2\Psi=\left(\left(\frac12\omega
r\right)^2+{\ell(\ell+1)\over
r^2}-\left(\ell+\frac32\right)\omega-E\right)\Psi.$$ By the change
$r\mapsto\left(\sqrt{2\over\omega}\right)r$ we obtain the
Schr\"odinger equation
$$\partial_r^2\Psi=\left(r^2+{\ell(\ell+1)\over r^2}-(2\ell+3)-\lambda\right)\Psi,\quad \lambda={2\over
\omega}E$$ and, in order to apply Kovacic's algorithm, we set
$$R=r^2+{\ell(\ell+1)\over r^2}-(2\ell+3)-\lambda.$$

We can see that this equation could fall in case 1, in case 2 or
in case 4 of Kovacic's algorithm. We start discarding the case 2
because by step 1 (of Kovacic's algorithm) we should have
conditions $c_2$ and $\infty_3$; so we should have
$E_c=\{2,4+4\ell,-4\ell\}$ and $E_{\infty}=\{-2\}$, and by step 2,
we would have  $n=-4\notin\mathbb{Z}_+$, so that
$D=\emptyset$. So this Schr\"odinger equation never falls in
case 2. Now, we only work with case 1; by step 1, conditions $c_2$
and $\infty_3$ are satisfied, so that
$$\left[ \sqrt {R}\right] _{c}=0,\quad\alpha_{c}^{\pm}={1\pm(2\ell+1)\over 2}, \quad\left[  \sqrt{R}\right]  _{\infty}=r,\quad \alpha_{\infty}^{\pm }={\mp(\lambda+2\ell+3)-1\over 2}.$$ By
step 2 we have the following possibilities for $n\in\mathbb{Z}_+$
and for $\lambda\in\Lambda$:
$$
\begin{array}{lll}
\Lambda_{++})\quad& n=\alpha^+_\infty - \alpha^+_0=-{1\over
2}\left(4\ell +6+\lambda\right),& \lambda=-2n-4\ell-6,\\ & &\\
\Lambda_{+-}) & n=\alpha^+_\infty - \alpha^-_0=-{1\over 2}\left(4
+\lambda\right),& \lambda=-2n-4, \\& &\\
\Lambda_{-+}) & n=\alpha^-_\infty - \alpha^+_0={\lambda\over 2},&
\lambda=2n,\\& & \\ \Lambda_{--}) & n=\alpha^-_\infty -
\alpha^-_0={1\over 2}\left(4\ell +2+\lambda\right),&
\lambda=2n-4\ell-2,
\end{array}
$$
where
$\Lambda_{++}\cup\Lambda_{+-}\cup\Lambda_{-+}\cup\Lambda_{--}=\Lambda$,
which means that $\lambda=2m,$ $m\in\mathbb{Z}$. Now, for
$\lambda\in\Lambda$, the rational function $\omega$ in Kovacic's
algorithm is given by:
$$
\begin{array}{lll}
\Lambda_{++})\quad& \omega=r+{\ell+1\over r},&
R_n=r^2+{\ell(\ell+1)\over r^2}+(2\ell+3)+2n,\\& & \\
\Lambda_{+-}) & \omega=r-{\ell\over r}, &
R_n=r^2+{\ell(\ell+1)\over r^2}-(2\ell-1)+2n,\\& & \\
\Lambda_{-+}) & \omega=-r+{\ell+1\over r},&
R_n= r^2+{\ell(\ell+1)\over r^2}-(2\ell+3)-2n,\\& &\\
\Lambda_{--}) & \omega=-r-{\ell\over r},&
R_n=r^2+{\ell(\ell+1)\over r^2}+(2\ell-1)-2n,
\end{array}
$$
where $R_n$ is the coefficient of the differential equation
$\widetilde{\mathcal L}_n:=\partial_r^2\Psi=R_n\Psi$, which is
integrable for every $n$. For every $\lambda\in\Lambda$, we can
see that
$\mathrm{DGal}(\widetilde{L}_n/K)=\mathrm{DGal}(L_\lambda/K)$,
where $\mathcal L_\lambda:=H\Psi=\lambda\Psi$.

By step 3, there exists a polynomial of degree $n$ satisfying the
relation (\ref{recu1}):
$$
\begin{array}{lllll}
\Lambda_{++})\quad& \partial_r^2P_n+2\left(r+{\ell+1\over
 r}\right)\partial_rP_n-2nP_n&=&0,&\lambda\in\Lambda_-,\\& & \\
\Lambda_{+-}) & \partial_r^2P_n+2\left(r-{\ell\over r}\right)\partial_rP_n-2nP_n&=&0,&\lambda\in\Lambda_-,\\& &\\
\Lambda_{-+}) & \partial_r^2P_n+2\left(-r+{\ell+1\over r}\right)\partial_rP_n+2nP_n&=&0,&\lambda\in\Lambda_+,\\ & &\\
\Lambda_{--}) & \partial_r^2P_n+2\left(-r-{\ell\over
 r}\right)\partial_rP_n+2nP_n&=&0,&\lambda\in\Lambda.
\end{array}
$$

These polynomials exist for all $\lambda\in\Lambda$ when their
degrees are $n\in 2\mathbb{Z}$, while for $n\in 2\mathbb{Z}+1$,
they exist only for the cases $\Lambda_{-+})$ and $\Lambda_{--})$
with special conditions. Hence we have obtained the
algebraic spectrum $\Lambda= 2\mathbb{Z}$, where
$\Lambda_{++}=4\mathbb{Z}_-$, $\Lambda_{+-}=2\mathbb{Z}_-$,
$\Lambda_{-+}=4\mathbb{Z}_+$, $\Lambda_{--}=2\mathbb{Z}$.

The possibilities for eigenfunctions, considering only $\lambda\in
4\mathbb{Z}$, are given by
$$
\begin{array}{lllll}
\Lambda_{++})\quad& \Psi_n(r)&=&r^{\ell+1}P_{2n}(r)e^{r^2\over 2},
&\lambda\in\Lambda_-,\\& &\\
\Lambda_{+-}) &\Psi_n(r)&=&r^{-\ell}P_{2n}(r)e^{r^2\over
2},&\lambda\in\Lambda_-,\\& & \\
\Lambda_{-+}) & \Psi_n(r)&=&r^{\ell+1}P_{2n}(r)e^{-r^2\over 2},
&\lambda\in\Lambda_+,\\& &\\ \Lambda_{--}) &
\Psi_n(r)&=&r^{-\ell}P_{2n}(r)e^{-r^2\over 2}, &\lambda\in\Lambda.
\end{array}
$$
To obtain the point spectrum, we look for $\Psi_n$ satisfying the
bound state conditions which is in only true for
$\lambda\in\Lambda_{-+}$. With the change
$r\mapsto\sqrt{\omega\over 2}r$, the point spectrum and ground
state of the Schr\"odinger equation with the 3D-harmonic
oscillator potential are respectively
$\mathrm{Spec}_p(H)=\{E_n:n\in\mathbb{Z}_+\},$ for
$E_n=2n\omega$, where $\omega$ is the angular velocity, and
$$\Psi_0=\left(\sqrt{\omega\over 2}r\right)^{\ell +
1}e^{-{\omega\over4}r^2}.$$

The bound state wave functions are obtained as
$\Psi_n=P_{2n}\Psi_0.$ Now, we can see that
$\mathrm{DGal}({L_0}/K)=\mathbb{B}$ and $\mathcal
E(H)=\mathrm{Vect}(1)$. Since $\Psi_n=P_{2n}\Psi_0,$ for all
$\lambda\in\Lambda$ we have 
$\mathrm{DGal}(L_\lambda/K)=\mathbb{B}$ and $\mathcal
E(H-\lambda)=\mathrm{Vect}(1)$. In particular,
$\mathrm{DGal}(L_\lambda/K)=\mathbb{B}$ and $\mathcal
E(H-\lambda)=\mathrm{Vect}(1)$ for all
$\lambda\in\mathrm{Spec}_p(H)$, where $\lambda=\frac2{\omega}E$.
\\

We remark that the Schr\"odinger equation with the 3D-harmonic
oscillator potential, through the changes $r\mapsto \frac12\omega
r^2$ and $\Psi\mapsto \sqrt{r}\Psi$, fall in the class of Whittaker
differential equations in which the parameters are given by
$$\kappa={(2\ell+3)\omega +2E\over 4\omega},\quad
\mu=\frac12\ell+\frac14.$$ By Martinet-Ramis theorem, we have
integrability when $\pm\kappa\pm \mu$ is a half integer. These
conditions coincides with our four sets $\Lambda_{\pm\pm}$.
\\

 \textbf{Coulomb potential:}
$$V(r)=-{\mathrm{e}^2\over r}+{\ell(\ell+1)\over r^2}+{\mathrm{e}^4\over4(\ell+1)^2},\quad \ell\in\mathbb{Z}.$$ 
The Schr\"odinger equation in this case can be written as
$$\partial_r^2\Psi=\left({\ell(\ell+1)\over r^2}-{{\rm e}^2\over r} +\frac{{\rm
e}^4}{4(\ell+1)^2}-E\right)\Psi.$$ By the change
$r\mapsto{2(\ell+1)\over {\rm e}^2}r$ we obtain the Schr\"odinger
equation
$$\partial_r^2\Psi=\left({\ell(\ell+1)\over r^2}-{2(\ell+1)\over r}+1-\lambda\right)\Psi,\quad \lambda={4(\ell+1)^2\over {\rm
e}^4}E$$ and in order to apply Kovacic algorithm, we set
$$R={\ell(\ell+1)\over r^2}-{2(\ell+1)\over r}+1-\lambda.$$

First we analyze the case for $\lambda=1$: we can see that this
equation only could fall in case 2 or in case 4 of Kovacic's
algorithm. We start discarding the case 2 because by step 1 we
should have conditions $c_2$ and $\infty_3$. So we should
have $E_c=\{2,4+4\ell,-4\ell\}$ and $E_{\infty}=\{1\}$ and by step
2, we would have  $n\notin\mathbb{Z}$. Thus,
$n\notin\mathbb{Z}_+$ and $D=\emptyset$, so the differential
Galois group of this Schr\"odinger equation for $\lambda=1$ is
${\rm SL}(2,\mathbb{C})$.
\\

Now, we analyze the case for $\lambda\neq 1$: we can see that this
equation could fall in case 1, in case 2 or in case 4. We start
discarding the case 2 because by step 1 we should have conditions
$c_2$ and $\infty_3$, so that we should have
$E_c=\{2,4+4\ell,-4\ell\}$ and $E_{\infty}=\{0\}$. By step 2, we
should have  $n=2\ell\in\mathbb{Z}_+$, so that $D=\{2\ell\}$
and the rational function $\theta$ is $\theta={-2\ell\over x}$,
but we discard this case because there could only exist one polynomial
of degree $2\ell$ for a fixed $\ell$, and for instance, there could only
exist one eigenstate and one eigenfunction for the Schr\"odinger
equation.
\\

Now, we only work with case 1, by step 1, conditions $c_2$ and
$\infty_3$ are satisfied. Thus,
$$\left[ \sqrt {R}\right] _{c}=0,\quad\alpha_{c}^{\pm}={1\pm(2\ell+1)\over 2}, \quad\left[  \sqrt{R}\right]  _{\infty}=\sqrt{1-\lambda},\quad \alpha_{\infty}^{\pm }=\mp{\ell+1\over \sqrt{1-\lambda}}.$$ By
step 2 we have the following possibilities for $n\in\mathbb{Z}_+$
and for $\lambda\in\Lambda$:
$$
\begin{array}{lll}
\Lambda_{++})\quad& n=\alpha^+_\infty -
\alpha^+_0=-(\ell+1)\left(1+{1\over\sqrt{1-\lambda}}\right), &
\lambda=1-\left({\ell+1\over\ell+1+n}\right)^2,\\& &\\
\Lambda_{+-}) & n=\alpha^+_\infty - \alpha^-_0=-{\ell+1\over
\sqrt{1-\lambda}}+\ell,&
\lambda=1-\left({\ell+1\over\ell-n}\right)^2, \\& &\\
\Lambda_{-+}) & n=\alpha^-_\infty -
\alpha^+_0=(\ell+1)\left({1\over \sqrt{1-\lambda}}-1\right),&
\lambda=1-\left({\ell+1\over\ell+1+n}\right)^2,\\& &\\
\Lambda_{--}) & n=\alpha^-_\infty - \alpha^-_0={\ell+1\over
\sqrt{1-\lambda}}+\ell,&
\lambda=1-\left({\ell+1\over\ell-n}\right)^2.
\end{array}
$$
We can see that $\lambda\in\Lambda_-$ when $\lambda\leq 0$, while
$\lambda\in\Lambda_+$  when $0\leq \lambda< 1$. Furthermore:
$$
\begin{array}{lll}
\Lambda_{++})\quad& \ell\leq -1,
\quad&\lambda\in\Bigg\{{\begin{array}{ll}
 \Lambda_-,\quad& \ell\leq {-n-2\over 2}\\&\\
  \Lambda_+,&
{-n-2\over 2}\leq\ell\leq-1
\end{array}}\\& &\\
\Lambda_{+-}) & \ell>0, &\lambda\in\Bigg\{{\begin{array}{ll}
 \Lambda_-,\quad& \ell\geq {n-1\over 2}\\&\\
  \Lambda_+,&
0\leq\ell\leq{n-1\over 2}
\end{array}}\\& &\\
\Lambda_{-+}) &
\ell\in\mathbb{Z},&\lambda\in\Bigg\{{\begin{array}{ll}
 \Lambda_-,\quad& \ell\leq {-n-2\over 2}\\&\\
  \Lambda_+,&
\ell\geq-1
\end{array}}\\& &\\
\Lambda_{--}) & \ell>0, &\lambda\in\Bigg\{{\begin{array}{ll}
 \Lambda_-,\quad& \ell\geq {n-1\over 2}\\&\\
  \Lambda_+,&
0\leq\ell\leq {n-1\over 2}
\end{array}}
\end{array}
$$
 So the possible algebraic spectrum can be
$\Lambda=\Lambda_{++}\cup\Lambda_{+-}\cup\Lambda_{-+}\cup\Lambda_{--}$,
that is
\begin{equation}\label{EQL}
\Lambda=\left\{1-\left({\ell+1\over\ell+1+n}\right)^2:
n\in\mathbb{Z}_+\right\}\cup
\left\{1-\left({\ell+1\over\ell-n}\right)^2:
n\in\mathbb{Z}_+\right\},
\end{equation}

Now, for $\lambda\in\Lambda$, the rational function $\omega$ is
given by: {\small
$$
\begin{array}{llll}
\Lambda_{++})\quad& \omega={\ell+1\over\ell+1+n}+{\ell+1\over r},&\lambda\in\Lambda_{++},& R_n={\ell(\ell+1)\over r^2}-{2(\ell+1)\over r}+\left({\ell+1\over\ell+1+n}\right)^2,\\& &\\
\Lambda_{+-}) & \omega={\ell+1\over\ell-n}-{\ell\over r},&\lambda\in\Lambda_{+-},& R_n={\ell(\ell+1)\over r^2}-{2(\ell+1)\over r}+\left({\ell+1\over\ell-n}\right)^2,\\& &\\
\Lambda_{-+}) & \omega=-{\ell+1\over\ell+1+n}+{\ell+1\over
r},&\lambda\in\Lambda_{-+},& R_n={\ell(\ell+1)\over
r^2}-{2(\ell+1)\over r}+\left({\ell+1\over\ell+1+n}\right)^2,\\&
&\\ \Lambda_{--}) & \omega=-{\ell+1\over\ell-n}-{\ell\over
r},&\lambda\in\Lambda_{--},& R_n={\ell(\ell+1)\over
r^2}-{2(\ell+1)\over r}+\left({\ell+1\over\ell-n}\right)^2,
\end{array}
$$}
\noindent where $R_n$ is the coefficient of the differential
equation $\partial_r^2\Psi=R_n\Psi$, which is integrable for every
$n$.
\\

\noindent By step 3, there exists a polynomial of degree $n$
satisfying the relation (\ref{recu1}),
$$
\begin{array}{llll}
\Lambda_{++})\quad&
\partial_r^2P_n+2\left({\ell+1\over\ell+1+n}+{\ell+1\over r}\right)\partial_rP_n+{2(\ell+1)\over r}\left(1+{\ell+1\over\ell+1+n}\right)P_n&=&0,\\& & &\\
\Lambda_{+-}) & \partial_r^2P_n+2\left({\ell+1\over\ell-n}-{\ell\over r}\right)\partial_rP_n+{2(\ell+1)\over r}\left(1-{\ell+1\over\ell-n}\right)P_n&=&0,\\& & &\\
\Lambda_{-+}) & \partial_r^2P_n+2\left(-{\ell+1\over\ell+1+n}+{\ell+1\over r}\right)\partial_rP_n+{2(\ell+1)\over r}\left(1-{\ell+1\over\ell+1+n}\right)P_n&=&0,\\& & &\\
\Lambda_{--}) &
\partial_r^2P_n+2\left(-{\ell+1\over\ell-n}-{\ell\over
r}\right)\partial_rP_n+{2(\ell+1)\over
r}\left(1+{\ell+1\over\ell-n}\right)P_n&=&0.
\end{array}
$$

These polynomials exist for every $\lambda\in\Lambda$ when $n\in
\mathbb{Z}$, but $P_0=1$ is satisfied only for
$\lambda\in\Lambda_{-+}$. Hence we have confirmed that the
algebraic spectrum $\Lambda$ is given by equation (\ref{EQL}).\\

The possibilities for the eigenfunctions are given by

$$
{\small
%{\tiny
\begin{array}{llllll}
\Lambda_{++})\quad&
\Psi_n(r)&=&r^{\ell+1}P_n(r)f_n(r)e^{r},&f_n(r)=e^{{-nr\over
\ell+1+n}}, &\lambda\in\Bigg\{{\begin{array}{ll}
 \Lambda_-,& \ell\leq {-n-2\over 2}\\&\\
  \Lambda_+,&
{-n-2\over 2}\leq\ell\leq-1
\end{array}}\\&&&&&\\
\Lambda_{+-}) &
\Psi_n(r)&=&r^{-\ell}P_n(r)f_n(r)e^{r},&f_n(r)=e^{{n+1\over
\ell-n}r}, &\lambda\in\Bigg\{{\begin{array}{ll}
 \Lambda_-,& \ell\geq {n-1\over 2}\\&\\
  \Lambda_+,&
0\leq\ell\leq {n-1\over 2}
\end{array}}\\&&&&&\\
\Lambda_{-+}) &
\Psi_n(r)&=&r^{\ell+1}P_n(r)f_n(r)e^{-r},&f_n(r)=e^{{nr\over
\ell+1+n}}, &\lambda\in\Bigg\{{\begin{array}{ll}
 \Lambda_-,& \ell\leq {-n-2\over 2}\\&\\
  \Lambda_+,&
\ell\geq-1
\end{array}}\\&&&&&\\
\Lambda_{--}) &
\Psi_n(r)&=&r^{-\ell}P_n(r)f_n(r)e^{-r},&f_n(r)=e^{{n+1\over
n-\ell}r}, &\lambda\in\Bigg\{{\begin{array}{ll}
 \Lambda_-,& \ell\geq {n-1\over 2}\\&\\
  \Lambda_+,&
0\leq\ell\leq {n-1\over 2}
\end{array}}
\end{array}}
$$
but $\Psi_n$ should satisfy the bound state conditions which is
only true for $\lambda\in\Lambda_{-+}\cap \Lambda_+$, so that we
choose $\Lambda_{-+}\cap \Lambda_+={\rm spec}_p(H)$, that is
$$
{\rm spec}_p(H)=\left\{1-\left(
{\ell+1\over\ell+n+1}\right)^2:\quad n\in\mathbb{Z}_+,\quad
\ell\geq -1 \right\}.
$$
By the change $r\mapsto{\rm e
 ^2\over 2(\ell+1)}r$, the point spectrum and ground state of the Schr\"odinger equation with Coulomb potential are respectively
$$\mathrm{Spec}_p(H)=\{E_n:n\in\mathbb{Z}_+\},\quad E_n={{\rm e^4}\over 4}\left({1\over (\ell+1)^2}-{1\over
(\ell+1+n)^2}\right)$$ and
$$\Psi_0=\left({\rm e
 ^2\over 2(\ell+1)}r\right)^{\ell + 1}e^{-{\rm e
 ^2\over 2(\ell+1)}r}.$$
The eigenstates are given by $\Psi_n=P_{n}f_n\Psi_0,$ where $$
f_n(r)=e^{{n{\rm e}^2r\over 2(\ell+1+n)(\ell+1)}}.$$ Now, we can
see that $\mathrm{DGal}({L_0}/K)=\mathbb{B}$ and $\mathcal
E(H)=\mathrm{Vect}(1)$. Since $\Psi_n=P_{2n}f_n\Psi_0,$ for all
$\lambda\in\Lambda$ we have 
$\mathrm{DGal}({L_\lambda}/K)=\mathbb{B}$ and $\mathcal
E(H-\lambda)=\mathrm{Vect}(1)$. In particular,
$\mathrm{DGal}(L_\lambda/K)=\mathbb{B}$ and $\mathcal
E(H-\lambda)=\mathrm{Vect}(1)$ for all $E\in\mathrm{Spec}_p(H)$,
where $E=\frac{\mathrm{e}^4}{4(\ell +1)^2}\lambda$.
\\

We remark that, as in the three dimensional harmonic oscillator,
the Schr\"odinger equation with the Coulomb potential, through the
change $$r\mapsto {\frac {\sqrt {-4\, \left(\ell+1
\right)^{2}E+{\rm e}^{4 }}}{\ell+1}}r,$$ falls in the class of Whittaker
differential equation with the parameters given by
$$\kappa={\frac {{\rm e}^{2
} \left( \ell+1 \right) }{\sqrt {-4\, \left( \ell+1 \right)
^{2}E+{\rm e}^ {4}}}} ,\quad \mu=\ell+\frac12.$$ Applying
Martinet-Ramis theorem, we can impose $\pm\kappa\pm \mu$ half
integer, to coincides with our four sets $\Lambda_{\pm\pm}$.
\\

\begin{remark}
By direct application of Kovacic's algorithm we have:
\begin{itemize}
\item The Schr\"odinger equation \eqref{equ1} with potential
$$V(x)=ax^2+{b\over x^2}$$ is
integrable for $\lambda\in\Lambda$ when

\begin{itemize}
\item $a=0,\quad b=\mu(\mu+1), \quad \mu\in \mathbb{C}, \quad \Lambda=\mathbb{C},$
\item $a=1,\quad b=0, \quad \lambda\in\Lambda=2\mathbb{Z}+1,$
\item $a=1,\quad b=\ell(\ell+1), \quad \ell\in \mathbb{Z}^*, \Lambda=2\mathbb{Z}+1.$
\end{itemize}
\item The only rational potentials (up to transformations) in which the elements of the algebraic spectrum are placed at the same
distance, belongs to the family of potentials given by
$$V(x)=\sum_{k=-\infty}^2 a_kx^k,\quad a_2\neq 0.$$ In particular, the set $\Lambda$ for the harmonic oscillator ($a=1$, $b=0$)
 and $3D$ harmonic oscillator ($a=1$, $b=\ell(\ell +1)$) satisfies this.

\end{itemize}
\end{remark}

\begin{proposition}\label{luckphyst} Let $\mathcal L_\lambda$ be the Schr\"odinger equation \eqref{equ1} with $K=\mathbb{C}(x)$ and Picard-Vessiot extension $L_\lambda$.
 If $\mathrm{DGal}(L_0/K)$ is finite primitive, then $\mathrm{DGal}(L_\lambda/K)$ is not finite primitive for
  all $\lambda\in\Lambda\setminus \{0\}$.
\end{proposition}

\begin{proof} Pick $V\in\mathbb{C}(x)$ such that $\mathcal L_0$ falls in case 3 of Kovacic
algorithm, then $\circ u_\infty\geq 2$. Assume $t,
s\in\mathbb{C}[x]$ such that $V=\frac{s}t$, then ${\rm deg}(t)\geq
{\rm deg}
 (s)+2$ and
$V-\lambda=\frac{s-\lambda t}{t}$. Now, for $\lambda\neq 0 $ we
have  $\mathrm{deg}(s-\lambda t)=\mathrm{deg}(t)$ and
therefore $\circ (V-\lambda)_\infty=0$. So, for $\lambda\neq0$, the equation $\mathcal L_\lambda$ does not falls in case 3 of
Kovacic algorithm and therefore $\mathrm{DGal}(L_\lambda/K)$ is
not finite primitive.
\end{proof}

\begin{corollary}\label{luckyphysc}
 Let $\mathcal L_\lambda$ be the Schr\"odinger equation \eqref{equ1} with $K=\mathbb{C}(x)$ and Picard-Vessiot extension $L_\lambda$. If
 $\mathrm{Card}(\Lambda)>1$, then there is either zero or one
 value of $\lambda$ such that $\mathrm{DGal}(L_\lambda/K)$ is a finite primitive group.
\end{corollary}
\begin{proof} Assume that $\mathrm{Card}(\Lambda)>1$. Thus, by proposition \ref{luckphyst}, the
Schr\"odinger equation does not falls in case 3 of Kovacic's
algorithm.
\end{proof}
\medskip

We note that the study of the differential Galois groups of the
Schr\"odinger equation with the Coulomb potential had also been
analyzed by Jean-Pierre Ramis (\cite{comramis}) using his summability theory in the
eighties of the past century.

%\begin{comment}
\section{Transcendental Potentials and the Algebrization Method}\label{algebsusy} 

In supersymmetric quantum mechanics,
there exists potentials which are not rational functions and, for
this reason, it is difficult to apply our Galoisian approach from section \ref{susyrapo}. In this section, we address this problem and present 
some results concerning differential
equations with non-rational coefficients. For these differential
equations it is useful, whenever possible, to replace it by a new
differential equation over the Riemann sphere $\mathbb{P}^1$ (that
is, with rational coefficients). 
We will propose techniques to find such a change of variable and demonstrate, through numerous examples, how our approach allows us to handle many examples from quantum mechanics.
The equation over $\mathbb{P}^1$ is called an
\emph{algebraic form} or \emph{algebrization} of the original
equation.

\medskip

This algebraic form is already studied in the 19th century litterature (Liouville,
Darboux), but the problem of obtaining an algebraic form (if it
exists) of a given differential equation is in general not an easy
task. Here we develop a new methodology using the concept of
\emph{Hamiltonian change of variables}. Such a change of variables
allows us to compute the algebraic form of a large number of
differential equations of different types. In particular, for
second order linear differential equations, we can apply
\textit{Kovacic's algorithm} over the algebraic form to solve the
original equation.

\medskip
The following definition can be found in \cite{ber,howe,howe2}. 

\begin{definition}[Pullbacks of differential equations] Let $\mathfrak L_1\in
K_1[\partial_z]$ and $\mathfrak L_2\in K_2[\partial_x]$ be
differential operators. The expression $\mathfrak L_2\otimes
(\partial_x+v)$ refers to the operator whose solutions are the
solutions of $\mathfrak L_2$ multiplied by the solution $e^{-\int
vdx}$ of $\partial_x+v$.

\begin{itemize}
\item $\mathfrak L_2$ is a \emph{proper pullback} of $\mathfrak L_1$ by means of $f\in K_2$ if the change of variable $z=f(x)$ changes
$\mathfrak L_1$ into $\mathfrak L_2$.
\item $\mathfrak L_2$ is a \emph{pullback} (also known as weak pullback) of $\mathfrak L_1$ by means of $f\in K_2$ if there exists $v\in K_2$
such that $\mathfrak L_2\otimes (\partial_x+v)$ is a proper
pullback of $\mathfrak L_1$ by means of $f$.
\end{itemize}
\end{definition}

In case of compact Riemann surfaces, the geometric mechanism
behind the algebrization is a ramified covering of compact Riemann
surfaces, see \cite{mo2,mo}.

\subsection{Second Order Linear Differential Equations}\label{sec21}
Some results presented in this subsection also can be found in
\cite[\S 2 ]{acbl}.
\begin{proposition}[Change of the independent variable, \cite{acbl}]\label{pr1}
  Let us consider the following equation, with coefficients in
$\mathbb{C}(z)$:
\begin{equation}\label{eq1} \mathcal L_z:=
\partial_z^2 y+a(z)\partial_zy+b(z)y=0,
\end{equation}
and $\mathbb{C}(z)\hookrightarrow  L$ the corresponding
Picard-Vessiot extension. Let $(K,\delta)$ be a differential field
with $\mathbb C$ as field of constants. Let $\theta\in K$ be a
non-constant element. Then, by the change of variable $z =
\theta(x)$, equation \eqref{eq1} is transformed in
\begin{equation}\label{eq2} \mathcal L_x:=
\partial_x^2{r}+\left(a(\theta)\partial_x{\theta}-{\partial_x^2{\theta}\over
\partial_x{\theta}}\right)\partial_x{r}+b(\theta)(\partial_x{\theta})^2r=0,
\quad
 \partial_x=\delta, \quad r=y\circ\theta.
\end{equation}
  Let $K_0\subset K$ be the smallest differential field containing
$\theta$ and $\mathbb C$. Then equation \eqref{eq2} is a
differential equation with coefficients in $K_0$. Let
$K_0\hookrightarrow L_0$ be the corresponding Picard-Vessiot
extension. Assume that
  $$\mathbb C(z) \to K_0,\quad z\mapsto \theta$$
is an algebraic extension, then
$$(\mathrm{DGal}(L_0 / K_0))^0 = (\mathrm{DGal}(L / \mathbb C(z)))^0.$$
\end{proposition}

\begin{proposition}\label{propjaw}
Assume that $\mathcal L_x$ and $\mathcal L_z$ are as in proposition \ref{pr1}. Let $\varphi$ be the transformation given by
$$\varphi:\left\{\begin{array}{l}z\mapsto \theta(x)\\ \\
\partial_z\mapsto\frac1{\partial_x\theta}\delta.\end{array}\right.$$ 
Then
$\mathrm{DGal}(L_0/K_0)\simeq\mathrm{DGal}(L/K_0\cap L)\subset
\mathrm{DGal}(L/\mathbb{C}(z))$. Furthermore, if $K_0\cap L$ is
algebraic over $\mathbb{C}(z)$, then
$(\mathrm{DGal}(L_0/K_0))^0\simeq
(\mathrm{DGal}(L/\mathbb{C}(x)))^0.$
\end{proposition}

\begin{proof} By Proposition \ref{pr1}, the transformation $\varphi$
leads us to
$$\mathbb{C}(z)\simeq\varphi(\mathbb{C}(z))\hookrightarrow
K_0.$$ We identify $\mathbb{C}(z)$ with
$\varphi(\mathbb{C}(z))$,  so that we can view $\mathbb{C}(z)$
as a subfield of $K_0$ and then by the standard Galois theory diagram (see
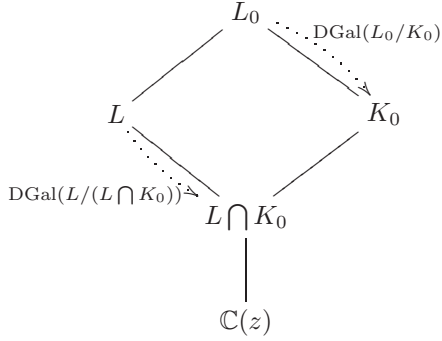
\cite{ka,we2}),

 \xymatrix{
      & L_0 \ar@{-}[dl] \ar@{-}[dr] \ar@/^/@{.>}[dr]^{\mathrm{DGal}(L_0/K_0)} & \\
L \ar@{-}[dr] \ar@/_/@{.>}[dr]_{\mathrm{DGal}(L/(L\bigcap K_0)) }
& & K_{0}
\ar@{-}[dl] \\
         & L\bigcap K_0  \ar@{-}[d]  &  \\
         & \mathbb{C}(z)             &
}

\noindent  we have
$$\mathrm{DGal}(L_0/K_0)\simeq\mathrm{DGal}(L/K_0\cap
L)\subset \mathrm{DGal}(L/\mathbb{C}(z)).$$
If $K_0\cap L$ is
algebraic over $\mathbb{C}(z)$, then
$$(\mathrm{DGal}(L_0/K_0))^0\simeq
(\mathrm{DGal}(L/\mathbb{C}(z)))^0.$$
\end{proof}

For the
remainder of this section we write $z=z(x)$ instead of
$\theta$.\\

\begin{remark}[Hard Algebrization]\label{hardal}  The proper pullback from
equation \eqref{eq2} to equation \eqref{eq1} is an algebrization
process. Therefore, we can try to algebrize any second order
linear differential equations with non-rational coefficients
(proper pullback) if we can put it in the form of equation
\eqref{eq2}. To do this, which will be called \textit{hard
algebrization}, we use the following steps.
\begin{description}
\item[Step 1] Find a term of the form $(\partial_x z)^2$ in the coefficient of
$y$ to obtain a candidate for $\partial_x z$ and then $z$.

\item[Step 2] Divide the coefficient of
$y$  by $(\partial_x z)^2$  to obtain a candidate $b(z)$ and check whether $b\in\mathbb{C}(z)$.

\item[Step 3] Add $(\partial_x^2 z)/\partial_xz$  to the coefficient of $\partial_xy$ and divide the result
by $\partial_xz$ to obtain a candidate
$a(z)$ and check whether $a\in\mathbb{C}(z)$.

\end{description}
\end{remark}
\medskip

\noindent Let us illustrate this method of  hard algebrization on the following example.\\

\begin{example} In \cite[p. 256]{si3}, Singer presents the second order linear differential equation
$$\partial_x^2 r-{1\over x(\ln x+1)}\partial_x r-(\ln x+1)^2r=0.$$ To algebrize
this differential equation we choose $(\partial_xz)^2=(\ln
x+1)^2$, so that $\partial_xz=\ln x+1$ and for instance
$$z=\int (\ln x+1 )dx=x\ln x,\quad b(z)=-1.$$ Now we find $a(z)$
in the expression
$$a(z)(\ln x+1)-{1\over x(\ln x+1)}=-{1\over x(\ln x+1)},$$
obtaining $a(z)=0$. So that the new differential equation is given
by $\partial_z^2y-y=0$, in which $y(z(x))=r(x)$ and one basis of
solutions of this differential equation is given by $\langle
e^z,e^{-z}\rangle$. Thus, the respective basis of solutions of the
first differential equation is given by $\langle e^{x\ln x}
,e^{-x\ln x}\rangle$.
\end{example}
\medskip

In general, this method is strongly heuristic because the quest of $z=z(x)$
in $b(z)(\partial_xz)^2$ can be purely a lottery, or simply there
may not exist  $z$ such that $a(z),b(z)\in\mathbb{C}(z)$. For
example, the equations presented by Singer in \cite[p. 257, 261,
270]{si3} and given by
$$\partial_x^2 r+{\mp 2x\ln^2x\mp 2x\ln x-1\over x\ln
x+x}\partial_x r+{-2x\ln^2 x-3x\ln x-x\mp 1\over x\ln x+ x }r=0,$$
$$\partial_x^2 r+{4x\ln x+ 2x\over 4x^2\ln x}\partial_x
r-{1\over 4x^2\ln x}r=0,\quad (x^2\ln^2x)\partial_x^2
r+(x\ln^2x-3x\ln x)\partial_x r+3r=0,$$ cannot be algebrized
systematically with this method, although they correspond to
pullbacks (not proper pullbacks) of differential equations with
constant
coefficients.\\

In \cite{brfr}, Bronstein and Fredet developed and implemented an
algorithm to solve differential equations over
$\mathbb{C}(x,e^{\int_{}^{}f})$ without algebrizing them, see also
\cite{fr}. As an application of proposition \ref{pr1} we have the
following result\footnote{This result is given in \cite[\S
2]{acbl}, but we include here the proof for completeness.}.\\

\begin{proposition}[Linear differential equation over
$\mathbb{C}(x,e^{\int_{}^{}f})$, \cite{acbl}]\label{brofe}

Let $f\in\mathbb{C}(x)$ be a rational function. Then, the
differential equation

\begin{equation}\label{exp}
\partial_x^2{r}-\left(f+{\partial_x{f}\over f}
-fe^{\int_{}^{}f}a\left(e^{\int_{}^{}f}\right)\right)\partial_x{r}+\left(f\left(e^{\int_{}^{}f}\right)\right)^2b\left(e^{\int_{}^{}f}\right)r=0,
\end{equation}
is algebrizable by the change $z=e^{\int_{}^{}f}$ and its
algebraic form is given by
$$\partial_z^2y+a(z)\partial_zy+b(z)y=0, \quad
r(x)=y(z(x)).$$
\end{proposition}
\begin{proof}
Assume that $r(x)=y(z(x)),$  and $z=z(x)=e^{\int f dx}$. We can
see that
$$\partial_xz=fz,\quad \partial_zy={\partial_x r\over fz},\quad \partial_z^2y={1\over fz}\partial_x\left({\partial_x r\over fz}\right)=
{1\over \left(fe^{\int_{}^{}f}\right)^2}\left(\partial_x^2
r-f+\left({\partial_x{f}\over f}\right)\right)\partial_x r.$$
Replacing in $\partial_z^2y+a(z)\partial_zy+b(z)y=0$ we obtain
equation \eqref{exp}.
\end{proof}

\begin{example}
The differential equation $$ \partial_x^2{r}-\left(x+{1\over x}
-2xe^{x^2}\right)\partial_x{r}+\lambda x^2e^{x^2}r=0,
$$
is algebrizable by the change $z=e^{x^2\over 2}$ and its algebraic
form is given by
$$\partial_z^2y+2z\partial_zy+\lambda y=0.$$

\end{example}
\begin{remark}
According to proposition \ref{brofe}, we have the following cases.

\begin{enumerate}
\item $f=n{\partial_xh\over h}$, for a rational function $h$, $n\in\mathbb{Z}_+$, we have the trivial case, both equations
are over the Riemann sphere and they have the same differential
field, so that does not need to be algebrized.

\item $f={1\over n}{\partial_xh\over h}$, for a rational function $h$,
$n\in\mathbb{Z}^+$, \eqref{exp} is defined over an algebraic
extension of $\mathbb{C}(x)$ and so that this equation is not
necessarily over the Riemann sphere.

\item $f\neq q{\partial_xh\over h}$, for any rational function $h$,  $q\in\mathbb{Q}$, \eqref{exp} is defined over a
transcendental extension of $\mathbb{C}(x)$ and so that this
equation is not over the Riemann sphere.
\end{enumerate}
\end{remark}

To algebrize second order linear differential equations may be easier
when the term in $\partial_x r$ is absent and the change of
variable is {\it{Hamiltonian}}.

\begin{definition}[Hamiltonian change of variable, \cite{acbl}]\label{def2} A change of
variable $z=z(x)$ is called \textit{Hamiltonian} if
$(z(x),\partial_xz(x))$ is a solution curve of the autonomous
classical one degree of freedom Hamiltonian system
$$\left\{\begin{array}{l}\partial_xz=\partial_wH\\ \partial_xw=-\partial_zH\end{array}\right.\quad \textrm{with}\quad H=H(z,w)={w^2\over 2}+V(z),$$
for some $V\in K$.
\end{definition}
\begin{remark}
  Assume that we algebrize equation \eqref{eq2} through a Hamiltonian
change of variables $z = z(x)$, i.e. $V\in\mathbb{C}(z)$. Then
$K_0 = \mathbb C(z,
\partial_xz, \ldots)$, but we have the algebraic relation
  $$(\partial_xz) ^2 = 2h - 2V(z), \quad h = H(z,\partial_xz) \in \mathbb C,$$
so that $K_0 = \mathbb C(z,\partial_xz)$ is an algebraic extension
of $\mathbb C(z)$. By proposition \ref{pr1}, the identity connected
component of the differential Galois group is preserved. On the
other hand, we can identify a Hamiltonian change of variable
$z=z(x)$ when there exists $\alpha\in K$ such that $(\partial_x
z)^2= \alpha(z)$. Thus, we introduce the {\it Hamiltonian
algebrization}, which corresponds to the algebrization process
done through a Hamiltonian change of variable.

\end{remark}

The following result, which also can be found in \cite[\S
2]{acbl}, is an example of Hamiltonian algebrization and
correspond to the case of reduced second order linear differential
equations.

\begin{proposition}[Hamiltonian Algebrization, \cite{acbl}]\label{pr2}
The differential equation
$$\partial_x^2{r}=q(x)r$$
is algebrizable through a Hamiltonian change of variable $z=z(x)$
if and only if there exist $f,\alpha$ such that
$${\partial_z\alpha\over\alpha},\quad {f\over \alpha}\in \mathbb{C}(z),\text{ where } f(z(x))=q(x),\quad \alpha(z)=2(H-V(z))=(\partial_xz)^2.$$
Furthermore, the algebraic form of the equation
$\partial_x^2{r}=q(x)r$ is
\begin{equation}\label{eq4}
\partial_z^2y+{1\over2}{\partial_z\alpha\over \alpha}\partial_zy-{f\over\alpha}y=0,\quad
r(x)=y(z(x)).
\end{equation}
\end{proposition}

\begin{remark}[Using the Algebrization Method]
The goal is to algebrize the differential equation $\partial_x^2
r=q(x)r$. We propose the following steps.

\begin{description}
\item[Step 1] Look for a \emph{Hamiltonian change of variable} $z=z(x)$ and two functions  $f$ and $\alpha$ such that $q(x)=f(z(x))$ and
$(\partial_xz (x))^2=\alpha(z(x))$.
\item[Step 2] Verify whether or not $f(z)/\alpha (z)\in\mathbb{C}(z)$ and $\partial_z\alpha(z)/\alpha
(z)\in\mathbb{C}(z)$ to see if the equation $\partial_x^2 r=q(x)r$
is algebrizable.
\item[Step 3] If the equation $\partial_x^2
r=q(x)r$ is algebrizable, its algebrization is
$$\partial_z^2y+{1\over2}{\partial_z\alpha\over \alpha}\partial_zy-{f\over\alpha}y=0,\quad y(z(x))=r(x).$$
\end{description}
When we have algebrized the differential equation $\partial_x^2
r=q(x)r$, we study its integrability, its eigenring  and its
differential Galois group.
\end{remark}
%\end{remark}
\begin{examples} Consider the following examples.
\begin{itemize}
\item Given the differential equation $\partial_x^2 r=f(\tan x)r$ with $f\in\mathbb{C}(\tan x)$, we can choose
 $z=z(x)=\tan x$ to obtain $\alpha(z)=(1+z^2)^2$, so that $z=z(x)$ is a Hamiltonian change of variable. We can
see that ${\partial_z\alpha\over\alpha}, {f\over
\alpha}\in\mathbb{C}(z)$ and the algebraic form of the
differential equation $\partial_x^2 r=f(\tan x)r$ with this
Hamiltonian change of variable is
$$\partial_z^2y+{2z\over 1+z^2}\partial_zy-{f(z)\over (1+z^2)^2}y=0,\quad y(\tan x)=r(x).$$

\item Given the differential equation $$\partial_x^2 r={\sqrt{1+x^2}+x^2\over 1+x^2}r,$$ we can choose
 $z=z(x)=\sqrt{1+x^2}$ to obtain $$f(z)={z^2+z-1\over z^2},\quad \alpha(z)={z^2-1\over z^2},$$ so that $z=z(x)$ is a Hamiltonian change of variable. We can
see that ${\partial_z\alpha\over\alpha}, {f\over
\alpha}\in\mathbb{C}(z)$ and the algebraic form for this case is
$$\partial_z^2y+{1\over z(z^2-1)}\partial_zy-{z^2+z-1\over z^2-1}y=0,\quad y(\sqrt{1+x^2})=r(x).$$
\end{itemize}
\end{examples}

We remark that in general the method of Hamiltonian algebrization
is not an algorithm, because it may not be obvious  to obtain a suitable
Hamiltonian $H$ satisfying definition \ref{def2}. We present now a
particular case of Hamiltonian algebrization considered as an
algorithm\footnote{Proposition \ref{cor1} is a slight improvement
of a similar result given in \cite[\S 2]{acbl}. Furthermore, we
include the proof here for completeness.}.

\begin{proposition}[Hamiltonian Algebrization Algorithm, \cite{acbl}]\label{cor1}
Consider $q(x)=g(z_1,\cdots, z_n)$, where $z_i=e^{\lambda_i x}$,
$\lambda_i\in\mathbb{C}^*$. The equation $\partial_x^2{r}=q(x)r$
is algebrizable if and only if.
$${\lambda_i\over \lambda_j}\in \mathbb{Q^*},\quad 1\leq i\leq n,\, 1\leq j\leq n,\quad g\in \mathbb{C}(z).$$
Furthermore, $\lambda_i=c_i\lambda$, where
$\lambda\in\mathbb{C}^*$ and $c_i\in \mathbb{Q}^*$ and for the
Hamiltonian change of variable
\begin{displaymath}
z=e^{\lambda x\over q},\text{ where } c_i={p_i\over q_i},\,
p_i,q_i\in\mathbb{Z}^*,\text{ }\gcd(p_i,q_i)=1 \text{ and }
q=\mathrm{lcm}(q_1,\cdots,q_n),
\end{displaymath} the algebrization of the differential equation
$\partial_x^2{r}=q(x)r$ is $$\partial_z^2y+{1\over
z}\partial_zy-q^2{g(z^{m_1},\ldots, z^{m_n})\over
\lambda^2z^2}y=0,\quad m_i={qp_i\over q_i},\quad y(z(x))=r(x).$$
\end{proposition}
\begin{proof} Assuming $\lambda_i/\lambda_j=c_{ij}\in\mathbb{Q}^*$ we can see that there exists $\lambda\in\mathbb{C}^*$
 and $c_i\in\mathbb{Q}^*$ such that $\lambda_i=\lambda c_i$, so that $$e^{\lambda_ix}=e^{c_i\lambda x}=e^{{p_i\over q_i}\lambda x}=
 \left(e^{{\lambda\over q}x}\right)^{qp_i\over q_i},\, p_i,q_i\in\mathbb{Z}^*,\, \gcd(p_i,q_i)=1,\, \mathrm{lcm}(q_1,\ldots,q_n)=q.$$ Now,
 setting $z=z(x)=e^{{\lambda\over q}x}$ we can see that
$$f(z)=g(z^{m_1},\ldots, z^{m_n}),\quad m_i={qp_i\over q_i},\quad \alpha={\lambda^2z^2\over q^2}.$$ 
As $q|q_i$, we have  $m_i\in\mathbb{Z}$, so that
$${\partial_z\alpha\over\alpha},\quad {f\over\alpha}\in\mathbb{C}(z)$$ and
the algebraic form is given by $$\partial_z^2y+{1\over
z}\partial_zy-q^2{g(z^{m_1},\ldots, z^{m_n})\over
\lambda^2z^2}y=0,\quad y(z(x))=r(x).$$
\end{proof}
\begin{remark} Propositions \ref{pr2} and \ref{cor1} allow the algebrization
of a large number of second order differential equations, see for
example \cite{poza}. In particular, under the assumptions of
proposition \ref{cor1}, we can algebrize automatically
differential equations with trigonometrical or hyperbolic
coefficients.
\end{remark}
\begin{examples} Consider the following examples.
\begin{itemize}
\item Given the differential equation $$\partial_x^2 r={e^{\frac12x}+3e^{-\frac23x}-2e^{\frac54x}\over e^x+e^{-\frac32x}}r,\,
 \lambda_1=\frac12,\, \lambda_2=-\frac23,\, \lambda_3=\frac54,\, \lambda_4=1,\,\lambda_5=-\frac32,$$
we see that $\lambda_i/\lambda_j\in\mathbb{Q}$, $\lambda=1$,
$q=\mathrm{lcm}(1,2,3,4)=12$ and the Hamiltonian change of
variable for this case is $z=z(x)=e^{\frac1{12}x}$. We can see
that
$$\alpha(z)={1\over 144}z^2,\quad f(z)={z^6+3z^{-8}-2z^{15}\over
z^{12}+z^{-18}},\quad {\partial_z\alpha\over\alpha}, {f\over
\alpha}\in\mathbb{C}(z)$$ and the algebraic form is given by
$$\partial_z^2y+{1\over z}\partial_zy-144{z^6+3z^{-8}-2z^{15}\over z^{14}+z^{-16}}y=0,\quad y(e^{{1\over 12}x})=r(x).$$

\item Given the differential equation $$\partial_x^2 r=(e^{2\sqrt2x}+e^{-\sqrt2x}-e^{3x})r,\, \lambda_1=2\sqrt2,\, \lambda_2=-\sqrt2,\,\lambda_3=3,$$
we see that $\lambda_1/\lambda_2\in\mathbb{Q}$, but
$\lambda_1/\lambda_3\notin\mathbb{Q}$, so that this differential
equation cannot be algebrized.
\end{itemize}
\end{examples}

We remark that it is also possible to use the algebrization method to
transform differential equations  when either the starting
equation has rational coefficients or the transformed equation does not have rational coefficients.

\begin{examples} As an illustration, we treat the following examples.
\begin{itemize}
\item Consider the following differential equation:
$$\partial_x^2r={x^4+3x^2-5\over x^2+1}r=0.$$ We can choose
$z=z(x)=x^2$ so that $\alpha=4z$ and the new differential equation
is $$\partial_z^2y+{1\over2 z}\partial_zy-{z^2+2z-5\over
4z(z+1)}y=0,\quad y(x^2)=r(x).$$
\item Consider the Mathieu differential
equation $\partial_x^2 r=(a+b\cos(x))r$. We can choose
$z(x)=\ln(\cos(x))$ so that $\alpha=e^{-2z}-1$ and the new
differential equation is
$$\partial_z^2y-{1\over 1-e^{2z}}\partial_zy-{ae^{2z}+be^{3z}\over 1-e^{2z}}y=0,\quad y(\ln(\cos (x)))=r(x).$$
\end{itemize}
\end{examples}

Recently, the Hamiltonian algebrization (propositions \ref{pr2}
and \ref{cor1}) has been applied in \cite{ac,acalde,acbl} to
obtain non-integrability results.

\subsection{The Operator $\widehat{\partial_z}$ and the Hamiltonian Algebrization}
The generalization of proposition \ref{pr1} to higher order linear
differential equations is difficult. But, it is possible to obtain
generalizations of proposition \ref{pr2} by means of Hamiltonian
change of variable. We recall that $z=z(x)$ is a Hamiltonian
change of variable if there exists $\alpha$ such that
$(\partial_xz)^2=\alpha(z)$. More specifically, if $z=z(x)$ is a
Hamiltonian change of variable, we can write
$\partial_xz=\sqrt{\alpha}$, which leads us to the following
notation: $\widehat{\partial}_z=\sqrt{\alpha}\partial_z$.

We can see that $\widehat{\partial}_z$ is a \textit{derivation}
because satisfy
$\widehat{\partial}_z(f+g)=\widehat{\partial}_zf+\widehat{\partial}_zg$
and the Leibnitz rules $$\widehat{\partial}_z(f\cdot
g)=\widehat{\partial}_zf\cdot g +f\cdot\widehat{\partial}_zg,\quad
\widehat{\partial}_z\left(\frac{f}{g}\right)=\frac{\widehat{\partial}_zf\cdot
g -f\cdot\widehat{\partial}_zg}{g^2}.$$ We can notice that the
chain rule is given by $\widehat{\partial}_z(f\circ
g)=\partial_gf\circ g\widehat{\partial}_z(g)\neq
\widehat{\partial}_gf\circ g\widehat{\partial}_z(g)$. The
iteration of $\widehat{\partial}_z$ is given by
$$\widehat{\partial}_z^0=1,\quad
\widehat{\partial}_z=\sqrt{\alpha}\partial_z,\quad
\widehat{\partial}_z^n=\sqrt{\alpha}\partial_z\widehat{\partial}^{n-1}_z=
\underbrace{\sqrt{\alpha}\partial_z\left(\ldots
\left(\sqrt{\alpha}\partial_z\right)\right)}_{n\text{ times }
\sqrt{\alpha}\partial_z}.$$ We say that a Hamiltonian change of
variable is rational when the potential $V\in\mathbb{C}(x)$ and
for instance $\alpha\in\mathbb{C}(x)$. For the
remainder of this
paper, we let
$\widehat{\partial}_z=\sqrt{\alpha}\partial_z$ where $z=z(x)$ is a
Hamiltonian change of variable and $\partial_xz=\sqrt{\alpha}$. In
particular, $\widehat{\partial}_z=\partial_z=\partial_x$ if and
only if $\sqrt{{\alpha}}=1$, i.e., $z=x$.\\

\begin{theorem}\label{propjawpa}  Consider the systems of linear differential equations
$[A]$ and $[\widehat{A}]$ given respectively by
$$\partial_x\mathbf{Y}=-A\mathbf{Y},\quad
\widehat{\partial}_z\widehat{\mathbf{Y}}=-\widehat{A}\widehat{\mathbf{Y}},\quad
A=[a_{ij}],\quad \widehat{A}=[\widehat{a}_{ij}],\quad
\mathbf{Y}=[y_{i1}], \quad
\widehat{\mathbf{Y}}=[\widehat{y}_{i1}],$$ where $a_{ij}\in
K=\mathbb{C}(z(x),\partial_x(z(x)))$,
$\widehat{a}_{ij}\in\mathbb{C}(z)\subseteq
\widehat{K}=\mathbb{C}(z,\sqrt{\alpha})$, $1\leq i\leq n$, $1\leq
j\leq n$, $a_{ij}(x)=\widehat{a}_{ij}(z(x))$ and
$y_{i1}(x)=y_{i1}(z(x))$. Suppose that $L$ and $\widehat{L}$ are
the Picard-Vessiot extensions of $[A]$ and $[\widehat{A}]$
respectively. If the transformation $\varphi$ is given by
$$\varphi: \left\{ \begin{array}{l}x\mapsto z\\ a_{ij}\mapsto
\widehat{a}_{ij}\\y_{i1}(x)\mapsto\widehat{y}_{i1}(z(x))\\
 \partial_x\mapsto\widehat{\partial}_z\end{array}\right.,$$
then the following statements hold.
\begin{itemize}
\item $K\simeq \widehat{K}$,\quad $(K,\partial_x)\simeq (\widehat{K},\widehat{\partial}_z)$.
\item $\mathrm{DGal}(L/K)\simeq\mathrm{DGal}(\widehat{L}/\widehat{K})\subset
\mathrm{DGal}(\widehat{L}/{\mathbb{C}(z)})$.
\item $(\mathrm{DGal}(L/K))^0\simeq
(\mathrm{DGal}(\widehat{L}/{\mathbb{C}(z)})^0.$
\item $\mathcal E(A)\simeq \mathcal E(\widehat{A}).$
\end{itemize}
\end{theorem}

\begin{proof} We proceed as in the proof of proposition \ref{propjaw}. As $z=z(x)$ is a rational Hamiltonian change of variable, the transformation $\varphi$ leads us to
$$\mathbb{C}(z)\simeq\varphi(\mathbb{C}(z))\hookrightarrow
K,\quad K\simeq \widehat{K},\quad \mathbb{C}(z)\hookrightarrow
\widehat{K},\quad (K,\partial_x)\simeq
(\widehat{K},\widehat{\partial}_z)$$ that is, we identify
$\mathbb{C}(z)$ with $\varphi(\mathbb{C}(z))$, and so that we can
view $\mathbb{C}(z)$ as a subfield of $K$ and then, by the
standard Galois theory diagram (see \cite{ka,we2}),

 \xymatrix{
      & L \ar@{-}[dl] \ar@{-}[dr] \ar@/^/@{.>}[dr]^{\mathrm{DGal}(L/K)} & \\
\widehat{L} \ar@{-}[dr]
\ar@/_/@{.>}[dr]_{\mathrm{DGal}(\widehat{L}/{\widehat{K}})} & & K
\ar@{-}[dl] \\
         & \widehat{K}  \ar@{-}[d]  &  \\
         & \mathbb{C}(z)             &
}

\noindent so that we have
$\mathrm{DGal}(L/K)\simeq\mathrm{DGal}(\widehat{L}/{\widehat{K}})\subset
\mathrm{DGal}(\widehat{L}/{\mathbb{C}(z)})$,
$(\mathrm{DGal}(L/K))^0\simeq
(\mathrm{DGal}(\widehat{L}/{\mathbb{C}(z)}))^0,$ and $\mathcal
E(A)\simeq \mathcal E(\widehat{A}).$
\end{proof}
\medskip

We remark that the transformation $\varphi$, given in theorem
\ref{propjawpa}, is virtually strong isogaloisian when
$\sqrt{\alpha}\notin\mathbb{C}(z)$ and for
$\sqrt{\alpha}\in\mathbb{C}(z)$, $\varphi$ is strong isogaloisian.
Furthermore, by cyclic vector method (see \cite{vasi}), we can
write the systems $[A]$ and $[\widehat{A}]$ in terms of the
differential equations ${\mathcal L}$ and $\widehat{\mathcal L}$.
Thus, $\widehat{\mathcal L}$ is the proper pullback of $\mathcal
L$ and $\mathcal E(\mathfrak L)\simeq \mathcal
E(\widehat{\mathfrak L}).$\medskip

\begin{example} Consider the system 
$$[A]:=\left\{ \begin{array}{l}\partial_x\gamma_1 =-{\frac
{2\sqrt{2}}{{e^{x}}+{e^{-x}}}}\gamma_3,\\
\\\partial_x\gamma_2=\frac{e ^{x}-e^{-x}}{e^{x}+e^{-x}}\gamma_3,\\ \\
\partial_x\gamma_3 =\frac {2 \sqrt {2}}{e^{x}+e^{-x}}\gamma_1-\frac {e^{x}- e^{-x}}
{e^{x}+e^{-x}}\gamma_2
\end{array}\right.
$$ 
Through the Hamiltonian change of variable $z=e^x$, and with $\sqrt{\alpha}=z$, it is transformed into the system
$$
\widehat{[A]}:=\left\{ \begin{array}{l}\partial_z\widehat{\gamma}_1 =-\frac{2\sqrt{2}}{z^2+1}\widehat{\gamma}_3,\\
\\\partial_z\widehat{\gamma}_2=\frac{z^2-1}{z(z^2+1)}\widehat{\gamma}_3,\\ \\
\partial_z\widehat{\gamma}_3 =\frac {2 \sqrt {2}}{z^2+1}\widehat{\gamma}_1-\frac{z^2-1}
{x(x^2+1)}\widehat{\gamma}_2
\end{array}\right. .
$$
One solution of the system $\widehat{[A]}$ is given by
$${1\over z^2+1}\begin{pmatrix}{\sqrt{2}\over 2}(1 - z^2)\\ z\\
-z\end{pmatrix},$$ and the corresponding solution for the system
$[A]$ is
$${1\over e^{2x}+1}\begin{pmatrix}{\sqrt{2}\over 2}(1 - e^{2x})\\ e^x\\
-e^x\end{pmatrix}.$$  
\end{example}
%\medskip

\begin{remark}
The algebrization given in proposition \ref{pr2} is an example of
how the introduction of the new derivative $\widehat{\partial}_z$
simplifies the proofs and computations. Such proposition is
naturally extended to $\partial_x^2y+a\partial_xy+by=0$, using
$\varphi$ to obtain
$\widehat{\partial}_z^2\widehat{y}+\widehat{a}\widehat{\partial}_z\widehat{y}+\widehat{b}\widehat{y}=0$,
which is equivalent to
\begin{equation}\label{eqalgen}
\alpha\partial_z^2\widehat{y}+\left({\partial_x\alpha\over
2}+{\sqrt{\alpha}\widehat{a}
}\right)\partial_z\widehat{y}+\widehat{b}\widehat{y}=0,
\end{equation} where $y(x)=\widehat{y}(z(x))$, $\widehat{a}(z(x))=a(x)$ and
$\widehat{b}(z(x))=b(x)$.\\

In general, for $y(x)=\widehat{y}(z(x))$, the equation
$F(\partial_x^ny,\ldots,y,x)=0$ with coefficients given by
$a_{i_k}(x)$ is transformed into the equation
$\widehat{F}(\widehat{\partial}_z^n\widehat{y},\ldots,\widehat{y},z)=0$
with coefficients given by $\widehat{a}_{i_k}(z)$, where
$a_{i_k}(x)=\widehat{a}_{i_k}(z(x))$. In particular, for
$\sqrt{\alpha}, \widehat{a}_{i_k}\in\mathbb{C}(z)$, the equation
$\widehat{F}(\widehat{\partial}_z^n\widehat{y},\ldots,\widehat{y},z)=0$
is the Hamiltonian algebrization of
$F(\partial_x^ny,\ldots,y,x)=0$. Now, if each derivation
$\partial_x$ has even  order, then $\alpha$ and $\widehat{a}_{i_k}$
can be rational functions when we algebrize the equation
$F(\partial_x^ny,\ldots,y,x)=0$, where
$a_{i_k}\in\mathbb{C}(z(x),\partial_xz(x))$. 
For example, that
happens for linear differential equations given by
$$\partial_x^{2n}y+a_{n-1}(x)\partial_x^{2n-2}y+\ldots+a_2(x)\partial_x^4y+a_1(x)\partial_x^2y+a_0(x)y=0.$$
Finally, the algebrization algorithm given in proposition
\ref{cor1} can be naturally extended to a differential equation of the form
$$F(\partial_x^ny,\partial_x^{n-1}y,\ldots,\partial_xy,y,
e^{\mu t})=0,$$ which, through the change of variable $z=e^{\mu
x}$, is transformed into
$$\widehat{F}(\widehat{\partial}_z^n\widehat{y},\widehat{\partial}_z^{n-1}\widehat{y},\ldots,\widehat{\partial}_z\widehat{y},y,
z)=0.$$ 
In particular, we may consider the algebrization of Riccati
equations, higher order linear differential equations and systems.
\end{remark}
\medskip

\begin{examples} The following corresponds
to some examples of algebrizations for differential equations
given in \cite[p. 258, 266]{si3}.

\begin{enumerate}
\item The equation $\mathcal L:=\partial_x^2 y+(-2e^x-1)\partial_x y+e^{2x}y=0$ with the Hamiltonian
change of variable $z=e^x$, $\sqrt{\alpha}=z$, $\widehat{a}=-2z-1$
and $\widehat{b}=z^2$ is transformed into the equation
$\widehat{\mathcal
L}:=\partial_z^2\widehat{y}-2\partial_z\widehat{y}+\widehat{y}=0$
which can be easily solved. Bases of solutions for $\mathcal L$
and $\widehat{\mathcal L}$ are respectively given by $\langle e^z,ze^z\rangle$
and $\langle e^{e^x},e^xe^{e^x}\rangle$. Furthermore
$K=\mathbb{C}(e^x)$, $\widehat K=\mathbb{C}(z)$, $L$ and
$\widehat{L}$ are the Picard-Vessiot extensions of $\mathcal L$
and $\widehat{\mathcal L}$ respectively. Thus,
$\mathrm{DGal}(L/K)=\mathrm{DGal}(\widehat{L}/{\widehat{K}})$.

\item The differential equation $$\mathcal L:= \partial_x^2 y+{-24e^x-25\over 4e^x+5}\partial_x y+{20e^x\over 4e^x+5}y=0$$ with the Hamiltonian
change of variable $z=e^x$, $\sqrt{\alpha}=z$,
$$\widehat{a}={-24z-25\over 4z+5} \textrm{  and   } \widehat{b}={20z\over
4z+5}$$ is transformed into the equation
$$\widehat{\mathcal L}:=\partial_z^2\widehat{y}+{-20(z+1)\over x(4z+5)}\partial_z\widehat{y}+{20\over z(4z+5)}\widehat{y}=0,$$
which can be solved with Kovacic algorithm. A basis of solutions
for $\widehat{\mathcal L}$ is $\langle z+1,z^5\rangle$, so that a
basis for $\mathcal L$ is $\langle e^x+1,e^{5x}\rangle$.
Furthermore $K=\mathbb{C}(e^x)$, $\widehat{K}=\mathbb{C}(z)$, $L$
and $\widehat{L}$ are the Picard-Vessiot extensions of $\mathcal
L$ and $\widehat{\mathcal L}$ respectively. Thus,
$\mathrm{DGal}(L/K)=\mathrm{DGal}(\widehat L/ {\widehat{K}})=e$.
\end{enumerate}

\end{examples}
\medskip

\begin{remark}[Algebrization of the Riccati equation]\label{alric} The Riccati equation
\begin{equation}\label{eqalric0}\partial_xv=a(x)+b(x)v+c(x)v^2.\end{equation} Through
the Hamiltonian change of variable $z=z(x)$, it is transformed into the
Riccati equation
\begin{equation}\label{eqalric}
\partial_z\widehat{v}={1\over
\sqrt{\alpha}}(\widehat{a}(z)+\widehat{b}(z)\widehat{v}+\widehat{c}(z)\widehat{v}^2),
\end{equation}
where $v(x)=\widehat{v}(z(x))$, $\widehat{a}(z(x))=a(x)$,
$\widehat{b}(z(x))=b(x)$, $\widehat{c}(z(x))=c(x)$ and
$\sqrt{\alpha(z(x))}=\partial_xz(x)$. Furthermore, if
$\sqrt{\alpha}$, $\widehat{a}$, $\widehat{b}$, $\widehat{c}\in
\mathbb{C}(x)$, equation \eqref{eqalric} is an algebrization of
equation \eqref{eqalric0}.
\end{remark}
\medskip

\begin{example} Consider the Riccati differential equation
$$\mathcal L:=\partial_x v=\left(\tanh x-{1\over \tanh x}\right)v+\left(3\tanh x-3\tanh^3x\right)v^2,$$
which through the Hamiltonian change of variable $z=\tanh x$, for
instance $\sqrt{\alpha}=1-z^2$, is transformed into the Riccati
differential equation $$\widehat{\mathcal L}:=\partial_zv=-{1\over
z}v+3zv^2.$$ One solution for the equation $\widehat{\mathcal L}$
is
$$-{1\over z(3z-c)}, \text{  where }c \text{ is a constant,}$$ so
that the corresponding solution for equation $\mathcal L$ is
$$-{1\over \tanh x(3\tanh x-c)}.$$
\end{example}

The following result is the algebrized version of the relationship
between the eigenring s of systems and operators.

\begin{proposition}\label{aleigen} Consider the differential fields $K$, $\widehat{K}$ and consider the systems $[A]$ and $[\widehat{A}]$ given by
$$\partial_x{\mathbf{X}}=-A\mathbf{X}, \,
\widehat{\partial}_z{\widehat{\mathbf{X}}}=-\widehat{A}\widehat{\mathbf{X}},
\, \widehat{\partial}_z=\sqrt{\alpha}\partial_z,\, A=[a_{ij}],\,
\widehat{A}=[\widehat{a}_{ij}], \,a_{ij}\in K, \,
\widehat{a}_{ij}\in \widehat{K},$$ where $z=z(x)$,
$\mathbf{X}(x)=\widehat{\mathbf{X}}(z(x))$,
$\widehat{a}_{ij}(z(x))=a_{ij}(x)$, then their eigenrings satisfy $\mathcal{E}(A)\simeq
\mathcal{E}(\widehat{A})$. In particular, if we consider the
linear differential equations $$\mathcal L:=\partial_x^ny+
\sum_{k=0}^{n-1}a_k\partial_x^k y=0\quad \textrm{and} \quad
\widehat{\mathcal L}:=\widehat{\partial}_z^n\widehat{y}+
\sum_{k=0}^{n-1}\widehat{a}_k\widehat{\partial}_z^k
\widehat{y}=0,$$ where $z=z(x)$, $y(x)=\widehat{z}((x))$,
$\widehat{a}_{k}(z(x))=a_{k}(x)$, $a_k\in K$, $\widehat{a}_{k}\in
\widehat{K}$, then $\mathcal{E}(\mathfrak L)\simeq
\mathcal{E}(\widehat{\mathfrak L})$, where $\mathcal L:=\mathfrak
L y=0 $ and $\widehat{\mathcal L}:=\widehat{\mathfrak L}
\widehat{y}=0 $. Furthermore, assuming
$$P=\begin{pmatrix}p_{11}&\dots&p_{1n}\\ \vdots
\\p_{n1}&\ldots&p_{nn}\end{pmatrix},\quad
A=\begin{pmatrix}0&1&\dots&0\\ \vdots
\\-a_0&-a_1&\ldots&-a_{n-1}\end{pmatrix},$$ then
$$\mathcal E(\mathfrak L)=\left\{\sum_{k=1}^np_{1k}\partial_x^{k-1}:\,
\partial_xP=PA-AP,\, p_{ik}\in K\right\},$$
if and only if $$\mathcal E(\widehat{\mathfrak
L})=\left\{\sum_{k=1}^n\widehat{p}_{1k}\widehat{\partial}_z^{k-1}:\,
\widehat{\partial}_z\widehat{P}=\widehat{P}\widehat{A}-\widehat{A}\widehat{P},
\, \widehat{p}_{ik}\in \widehat{K} \right\}.$$
\end{proposition}

\begin{proof} By theorem \ref{propjawpa} we have  $K\simeq \widehat{K}$, $\mathcal{E}(A)\simeq
\mathcal{E}(\widehat{A})$ and $\mathcal{E}(\mathfrak L)\simeq
\mathcal{E}(\widehat{\mathfrak L})$. Using the derivation
$\widehat{\partial}_z$ and by induction on lemma \ref{propjaw} we
complete the proof.

\end{proof}

\begin{examples} We consider two different examples to illustrate the previous proposition.

\begin{itemize}\item Consider the differential equation
$\mathcal L_1:=\partial_x^2y-(1+\cos x-\cos^2x)y=0$. By means of
the Hamiltonian change of variable $z=z(x)=\cos x$, with
$\sqrt{\alpha}=-\sqrt{1-z^2}$, $\mathcal L_1$ is transformed into
the differential equation
$$\widehat{\mathcal L}_1:=\partial_z^2\widehat{y}-{z\over 1-z^2}\partial_z\widehat{y}-{1+z-z^2\over
1-z^2}\widehat{y}=0.$$ Now, computing the eigenring of
$\widehat{\mathfrak L}_1$ we have  $\mathcal
E(\widehat{\mathfrak L}_1)=\mathrm{Vect}(1)$, therefore the
eigenring of $\mathfrak L_1$ is given $\mathcal E(\mathfrak
L_1)=\mathrm{Vect}(1)$.

\item Now we consider the differential equation
$\mathcal L_2:=\partial_x^2y=\left(e^{2x}+\frac94\right)y$. By
means of the Hamiltonian change of variable $z=e^x$, with
$\sqrt{\alpha}=x$, $\mathcal L_2$ is transformed into the
differential equation
$$\widehat{\mathcal L}_2:=\partial_z^2\widehat{y}+{1\over z}\partial_z\widehat{y}-\left(1+{9\over
4x^2}\right)\widehat{y}=0.$$ Now, computing the eigenring of
$\widehat{\mathfrak L}_2$ we have  {\small{$$\mathcal
E(\widehat{\mathfrak L}_2)=\mathrm{Vect}\left(1,-2\left(\frac
{z^2-1}{z^2}\right)\partial_z-\frac{z^2-3}{z^3}
\right)=\mathrm{Vect}\left(1,-2\left(\frac
{z^2-1}{z^3}\right)\widehat{\partial}_z-\frac{z^2-3}{z^3}
\right),$$}} therefore the eigenring of $\mathfrak{L}_2$ is given
by
$$\mathcal E(\mathfrak{L}_2)=\mathrm{Vect}\left(1,-2\left(\frac
{e^{2x}-1}{e^{3x}}\right)\partial_x-\frac{e^{2x}-3}{e^{3x}}
\right).$$ The same result is obtained via matrix formalism, where
$$A=\begin{pmatrix}0&1\\e^{2x}+\frac94&0\end{pmatrix},\,
\widehat{A}=\begin{pmatrix}0&1\\z^{2}+\frac94&0\end{pmatrix},\,
\partial_xP=PA-AP,\,
\widehat{\partial}_z\widehat{P}=\widehat{P}\widehat{A}-\widehat{A}\widehat{P},$$
with $P\in Mat(2,\mathbb{C}(e^x))$ and $\widehat{P}\in
Mat(2,\mathbb{C}(z))$.

\end{itemize}
\end{examples}

\subsection{Applications in Supersymmetric Quantum Mechanics}
In this subsection we apply the derivation $\widehat{\partial}_z$
to the Schr\"odinger equation $H\Psi=\lambda\Psi$, where
$H=-\partial_x^2+V(x)$, $V\in K$. Assume that $z=z(x)$ is a
rational Hamiltonian change of variable for $H\Psi=\lambda\Psi$,
then $K=\mathbb{C}(z(x),\partial_xz(x))$. Thus, the
\textit{algebrized Schr\"odinger equation} is written as
\begin{equation}\label{eqalgs1}\widehat{H}\widehat{\Psi}=\lambda\widehat{\Psi},\quad
\widehat{H}=-\widehat{\partial}_z^2+\widehat{V}(z), \quad
\widehat{\partial}_z^2=\alpha\partial_z^2+\frac12\partial_z\alpha\partial_z,\quad
\widehat{K}=\mathbb{C}(z,\sqrt{\alpha}).\end{equation} The
\textit{reduced algebrized Schr\"odinger equation} is given by

\begin{equation}\label{eqalgs2}\widehat{\mathbf{H}}\Phi=\lambda\Phi,\quad \widehat{\mathbf{H}}=\alpha(z)\left(-{\partial}_z^2+\widehat{\mathbf
V}(z)\right),\quad \begin{array}{l}\widehat{\mathbf
V}(z)=\mathcal{V}+{\widehat{V}(z)\over \alpha},\\ \\ \mathcal{V}=
\partial_z\mathcal{W}+\mathcal{W}^2,\\ \\
\mathcal{W}={1\over 4}{\partial_z\alpha(z)\over
\alpha(z)}.\end{array}\end{equation} The eigenfunctions $\Psi$,
$\widehat{\Psi}$ and $\Phi$ corresponding to the operators $H$,
$\widehat{H}$ and $\widehat{\mathbf{H}}$  are related respectively
by
$$\Phi(z(x))=\sqrt[4]{\alpha}\widehat{\Psi}(z(x))=\sqrt[4]{\alpha}\Psi(x).$$

In order to apply the Kovacic's algorithm, we only consider the
algebrized operator $\widehat{\mathbf{H}}$, whilst the eigenrings
will be computed on $\widehat{H}$. Also it is possible to apply
the version of Kovacic's algorithm given in reference \cite{ulwe}
to the algebraized operator $\widehat{H}$. The following results
are obtained by applying Kovacic's algorithm to the reduced
algebrized Schr\"odinger equation (equation \eqref{eqalgs2})
$\widehat{\mathbf{H}}\Phi=\lambda\Phi$.

\begin{proposition}
Let $\widehat{\mathbf{L}}_\lambda$ be the Picard-Vessiot extension
of the reduced algebrized Schr\"odinger equation
$\widehat{\mathbf{H}}\Phi=\lambda\Phi$ with
$\alpha,\widehat{V}\in\mathbb{C}[z]$. If $\deg\alpha < 2+\deg
\widehat{V}$ then, for every $\lambda\in\Lambda$,
$\mathrm{DGal}(\widehat{\mathbf{L}}/{\widehat{K}})$ is not
finite primitive.
\end{proposition}
\begin{proof} Suppose that $\deg \alpha=n$ and $\deg \widehat{V}=m$. The reduced algebrized Schr\"odinger equation
$\widehat{\mathbf{H}}\Phi=\lambda\Phi$ can be written in the form
$$\partial_z^2\Phi=r\Phi,\quad r={4\alpha\partial_z^2\alpha-3(\partial_z\alpha)^2+16\alpha(\widehat{V}-\lambda)\over
16\alpha^2}.$$ As $m>n-2$ we have  $\circ
(r_\infty)=n-m<2$, which does not satisfy the condition ($\infty$)
of case 3 of Kovacic's algorithm, therefore
$\mathrm{DGal}(\widehat{\mathbf{L}}/{\widehat{K}})$ is not
finite primitive for any  $\lambda\in\Lambda$.

\end{proof}
\begin{proposition}
Let $\widehat{\mathbf{L}}_\lambda$ be the Picard-Vessiot extension
of the reduced algebrized Schr\"odinger equation
$\widehat{\mathbf{H}}\Phi=\lambda\Phi$ with
$\alpha\in\mathbb{C}[z]$, $\widehat{V}\in\mathbb{C}(z)$. 
If
$\circ(\widehat{V})_\infty <2-\deg\alpha$  then, for every $\lambda\in\Lambda$,
$\mathrm{DGal}(\widehat{\mathbf{L}}_\lambda/{\widehat{K}})$ is
not finite primitive.
\end{proposition}
\begin{proof}  Suppose that $\widehat{V}=s/t$, where $s$ and $t$ are coprime polynomials in $\mathbb{C}(z)$. Assume that $\deg \alpha=n$,
$\deg s=m$ and $\deg t=p$. The reduced algebrized Schr\"odinger
equation $\widehat{\mathbf{H}}\Phi=\lambda\Phi$ can be written in
the form
$$\partial_z^2\Phi=r\Phi,\quad r={4t\alpha\partial_z^2\alpha-3t(\partial_z\alpha)^2+16\alpha(s-\lambda t)\over
16t\alpha^2}.$$ 
As $m>n+p-2$, we have  $\circ
(r_\infty)=p+n-m<2$, which does not satisfy condition
($\infty$) of case 3 of Kovacic's algorithm. Therefore
$\mathrm{DGal}(\widehat{\mathbf{L}}/{\widehat{K}})$ is not
finite primitive for any $\lambda\in\Lambda$.

\end{proof}
\medskip

\begin{remark}
In a natural way, we obtain the algebrized versions of Darboux
transformation, i.e., the \textit{algebrized Darboux
transformation}, denoted by $\widehat{\mathrm{DT}}$. The notation 
$\widehat{\mathrm{DT}}_n$ stands for the $n$-th iteration of
$\widehat{\mathrm{DT}}$; the notation $\widehat{\mathrm{CI}}_n$ denotes the $n$-th iteration of the \textit{algebrized Crum iteration}, where the \textit{algebrized wronskian} is given by $$\, \widehat{W}(\widehat{y}_1,\ldots,\widehat{y}_n)=\left|\begin{matrix}\widehat{y}_1&\cdots& \widehat{y}_n\\
 \vdots& & \vdots \\ \widehat{\partial}^{n-1}_z\widehat{y}_1&\cdots& \widehat{\partial}^{n-1}_z\widehat{y}_n\end{matrix}\right|.$$
In the same way, we define \textit{algebrized shape invariant
potentials}, \textit{algebrized superpotential} $\widehat{W}$,
\textit{algebrized supersymmetric Hamiltonians} $\widehat{H}_\pm$,
algebrized supersymmetric partner potentials $\widehat{V}_\pm$,
\textit{algebrized ground state}
$\widehat{\Psi}_0^{(-)}=e^{-\int{\widehat{W}\over
\sqrt{\alpha}}dz}$,\textit{ algebrized wave functions}
$\widehat{\Psi}_\lambda^{(-)}$, \textit{algebrized raising and
lowering operators} $\widehat{A}$ and $\widehat{A}^\dagger$. Thus,
we can rewrite entirely section \ref{susyrapo} using the
derivation $\widehat{\partial}_z$.
\end{remark}
\medskip

The following theorem shows us the relationship between the
algebrization and the Darboux transformation.\\

\begin{theorem}  Given the Schr\"odinger equation $\mathcal L_\lambda:=H_-\Psi^{(-)}=\lambda\Psi^{(-)}$, the relationship
 between the algebrization $\varphi$ and Darboux
transformations $\mathrm{DT}$, $\widehat{\mathrm{DT}}$ with
respect to $\mathcal L_\lambda$ is given by
$\widehat{\mathrm{DT}}\varphi=\varphi\mathrm{DT}$, that is
$\widehat{\mathrm{DT}}\varphi(\mathcal L
)=\varphi\mathrm{DT}(\mathcal L)$. In other words, the Darboux
transformations $\mathrm{DT}$ and $\widehat{\mathrm{DT}}$ are
\emph{intertwined} by the algebrization $\varphi$.
\end{theorem}
\begin{proof} Consider the equations $\mathcal L_{\lambda}:=H_-\Psi^{(-)}=\lambda\Psi^{(-)}$ , $\widehat{\mathcal L}_\lambda:=
\widehat{H}_-\widehat{\Psi}^{(-)}=\lambda\widehat{\Psi}^{(-)}$,
$\widetilde{\mathcal L}_\lambda:=H_+\Psi^{(+)}=\lambda\Psi^{(+)}$
and $\widetilde{\widehat{\mathcal L}}:=
\widehat{H}_+\widehat{\Psi}^{(+)}=\lambda\widehat{\Psi}^{(+)}$,
where the Darboux transformations $\mathrm{DT}$ and
$\widehat{\mathrm{DT}}$ are given by $\mathrm{DT}(\mathcal
L)=\widetilde{\mathcal L}$,
$\widehat{\mathrm{DT}}(\widehat{\mathcal
L})=\widetilde{\widehat{\mathcal L}}$,
$$ \mathrm{DT}:
 \begin{array}{l} V_-\mapsto V_+\\ \\
  \Psi_\lambda^{(-)}\mapsto \Psi_\lambda^{(+)}\end{array},\quad \widehat{\mathrm{DT}}:\begin{array}{l}
   \widehat{V}_-\mapsto \widehat{V}_+\\ \\ \widehat{\Psi}_\lambda^{(-)}\mapsto\widehat{\Psi}_\lambda^{(+)},\end{array}$$
   and $\varphi(\mathcal
L_\lambda)=\widehat{\mathcal L_\lambda}$, where the algebrization
$\varphi$ is given as in theorem \ref{propjawpa}. Then the
following diagram commutes
 $$\xymatrix{\mathcal L_\lambda \ar[r]^-{\mathrm{DT}} &\widetilde{\mathcal L}_\lambda \\
   \widehat{\mathcal L}_\lambda  \ar[r]_{\widehat{\mathrm{DT}}} \ar@{<-}[u]^{\varphi} & \ar@{<-}[u]_{\varphi} \widetilde{\widehat{\mathcal L}}_\lambda }
   \qquad \begin{array}{c}\\ \\ \\
  \; \hbox{\rm so } \; \widehat{\mathrm{DT}}\varphi(\mathcal L
)=\varphi\mathrm{DT}(\mathcal
L)  \; \hbox{\rm i.e. } \;  \widetilde{\widehat{\mathcal
L}}=\widehat{\widetilde{\mathcal L}}.\end{array}$$
\end{proof}
To illustrate $\widehat{\mathrm{DT}}$ we present the following
examples.\\

\begin{examples} Consider the algebrized Schr\"odinger equation $\widehat{H}\widehat{\Psi}^{(-)}=\lambda\widehat{\Psi}^{(-)}$ with:
\begin{itemize}
\item $\sqrt{\alpha(z)}=\sqrt{z^2-1}$ and $\widehat{V}_-(z)=\frac{z}{z-1}$. Taking
$\lambda_1=1$ and $\widehat{\Psi}_1^{(-)}=\sqrt{z+1\over z-1}$, we
have 
$\widehat{\mathrm{DT}}(\widehat{V}_-)=\widehat{V}_+(z)=\frac{z}{z+1}$
and
$$\widehat{\mathrm{DT}}(\widehat{\Psi}^{(-)}_{\lambda})=\widehat{\Psi}^{(+)}_{\lambda}=\sqrt{z^2-1}\partial_z\widehat{\Psi}^{(-)}_\lambda+{1\over
\sqrt{z^2-1}}\widehat{\Psi}_\lambda^{(-)},$$ where
$\widehat{\Psi}_{\lambda}^{(-)}$ is the general solution of
$\widehat{H}_-\widehat{\Psi}^{(-)}=\lambda\widehat{\Psi}^{(-)}$
for $\lambda\neq 1$.

The original potential corresponding to this example is given by
$V_-(x)=\frac{\cosh x}{\cosh x -1}$ and for $\lambda_1=1$ the
particular solution $\Psi_1^{(-)}$ is given by ${\sinh x\over
\cosh x-1}$. Applying $\mathrm{DT}$ we have 
$\mathrm{DT}(V_-)=V_+(x)=\frac{\cosh x}{\cosh x +1}$ and
$\mathrm{DT}(\Psi_\lambda^{(-)})=\Psi_\lambda^{(+)}=\partial_x\Psi^{(-)}_\lambda+\frac1{\sinh
x}\Psi^{(-)}_\lambda$.

\item $\sqrt{\alpha}=-z$, $\widehat{V}_-(z)=z^2-z$. Taking $\lambda_1=0$ and $\widehat{\Psi}^{(-)}_0=e^{-z}$
 we have  $\widehat{\mathrm{DT}}(\widehat{V}_-)=V_+=z^2+z$ and
$\widehat{\mathrm{DT}}(\widehat{\Psi}^{(-)}_{\lambda})=\widehat{\Psi}^{(+)}_{\lambda}=-z\partial_z\widehat{\Psi}^{(-)}_\lambda-z\widehat{\Psi}^{(-)}_\lambda$,
where $\Psi_{\lambda}^{(-)}$ is the general solution of
$\widehat{H}_-\widehat{\Psi}^{(-)}=\lambda\widehat{\Psi}^{(-)}$
for $\lambda\neq 0$. This example corresponds to the \emph{Morse
potential} $V_-(x)=e^{-2x}-e^{-x}$.
\end{itemize}
\end{examples}
\medskip

To illustrate the algebrized Crum iteration $\widehat{\mathrm{CI}}_n$, we present the following
example, which is related with the Chebyshev polynomials.\\

\begin{example}

Consider  $\sqrt{\alpha}=-\sqrt{1-z^2}$, $V=0$ with
eigenvalues and eigenfunctions $\lambda_1=1$, $\lambda_2=4$,
$\widehat{\Psi}_{1}=z$, $\widehat{\Psi}_{4}=2z^2-1$,
$\widehat{\Psi}_{n^2}=T_n(z)$, where $T_n(z)$ is the Chebyshev
polynomial of first kind of degree $n$. The algebrized Wronskian
for $n=2$ is
$$\widehat{W}(z,2z^2-1)=-\sqrt{1-z^2}(2z^2+1),$$
 and by algebrized Crum iteration we obtain the potential $$\widehat{\mathrm{CI}}_2(\widehat{V})=\widehat{V}_2=((2z^2-1)\partial_z^2+z\partial_z)\ln \widehat{W}(z,2z^2-1)$$
 and the algebrized wave functions
$$
\widehat{\mathrm{CI}}_2(\widehat{\Psi}_\lambda)=\widehat{\Psi}_\lambda^{(2)}={\widehat{W}(z,2z^2-1,T_n)\over
\widehat{W}(z,2z^2-1)}.$$

\end{example}
\medskip

In a natural way,  we introduce the notion of algebrized shape
invariant potentials
$\widehat{V}_{n+1}(z,a_n)=\widehat{V}_n(z,a_{n+1})+R(a_n)$, where
the energy levels for $n>0$ are given by $E_n=R(a_1)+\cdots
R(a_n)$ and the algebrized eigenfunctions are given by
$\widehat{\Psi}_n(a_1)=\widehat{A}^\dagger(z,a_1)\cdots
\widehat{A}^\dagger(z,a_n)\Psi_0(z,a_n)$. To illustrate the
algebrized shape invariant potentials and the operators $\widehat
A$ and $\widehat A^\dagger$, we present the following example.\\

\begin{example} Let $\sqrt{\alpha}=1-z^2$ and consider the algebrized super
potential $\widehat{W}(z)=z$. Following the method proposed in
remark \ref{remshape}, step 1, we introduce $\mu\in\mathbb{C}$ to
obtain $\widehat{W}(z;\mu)=\mu z$, and
$$\widehat{V}_\pm(z;\mu)=\widehat{W}^2(z;\mu)\pm\widehat{\partial}_z\widehat{W}(z;\mu)=\mu(\mu\mp1)z^2\pm\mu.$$ 
Thus 
$\widehat{V}_+(z;a_0)=a_0(a_0-1)z^2+a_0$ and
$\widehat{V}_-(z;a_1)=a_1(a_1+1)z^2-a_1.$
 By step 2,  $$\widehat{\partial}_z(V_+(z;a_0)-V_-(z;a_1))=2z(1-z^2)(a_0(a_0-1)-a_1(a_1+1)).$$ By step 3, we obtain
$$a_1(a_1+1)=a_0(a_0-1),\quad a_1^2-a_0^2=-(a_1+a_0)$$ and, assuming $a_1\neq\pm a_0$, we have
$a_1=f(a_0)=a_0-1$ and
$R(a_1)=2a_0+1=(a_0+1)^2-a_0^2=a_1^2-a_0^2$. This means that the
potentials $\widehat{V}_\pm$ are algebrized shape invariant
potentials where the $n$-th energy level $E=E_n$ is easily obtained:
$$E_n=\sum_{k=1}^nR(a_k)=\sum_{k=1}^n\left(a_{k}^2-a_{k-1}^2\right)=a_n^2-a_0^2,\quad a_n=f^{n}(a_0)=a_0+n.$$ Now, the algebrized ground state wave
function of $\widehat{V}_-(z,a_0)$ is
$$\widehat{\Psi}_0=e^{\int{a_0z\over
1-z^2}dz}=\frac1{\left(\sqrt{1-z^2}\right)^{a_0}}.$$ Finally, we
can obtain the rest of eigenfunctions using the algebrized raising
operator:
$$\widehat{\Psi}_n(z,a_0)=\widehat{A}^\dagger(z,a_0)\widehat{A}^\dagger(z,a_1)\cdots\widehat{A}^\dagger(z,a_{n-1})\widehat{\Psi}_0(z,a_n).$$
This example corresponds to the \emph{P\"oschl-Teller potential}.
\end{example}

Now to illustrate the power of Kovacic's algorithm with the
derivation $\widehat{\partial}_z$, we study some Schr\"odinger
equations for non-rational shape invariant potentials.
\\

 \textbf{Morse potential:} $V(x)=e^{-2x}-e^{-x}.$\\

The Schr\"odinger equation $H\Psi=\lambda\Psi$ is
$$\partial_x^2\Psi=\left(e^{-2x}-e^{-x}-\lambda\right)\Psi.$$

By the Hamiltonian change of variable $z=z(x)=e^{-x}$, we obtain
$$\alpha(z)=z^2,\quad \widehat V(z)=z^2-z,\quad \widehat{\mathbf{V}}(z)={z^2-z-\frac14\over
z^2}.$$ Thus, $\widehat{K}=\mathbb{C}(z)$ and $K=\mathbb{C}(e^x)$.
The algebrized Schr\"odinger equation
$\widehat{H}\widehat{\Psi}=\lambda\widehat{\Psi}$ is
$$z^2\partial_z^2\widehat{\Psi}+z\partial_z\widehat{\Psi}-(z^2-z-\lambda)\widehat{\Psi}=0$$ and the reduced
algebrized Schr\"odinger equation
$\widehat{\mathbf{H}}\widehat{\Phi}=\lambda\widehat{\Phi}$ is
\begin{displaymath}\partial_z^2\Phi=r\Phi,\quad r={ z^2- z-\frac14-\lambda\over
z^2}.\end{displaymath}
 This equation could only  fall under
case 1, in case 2 or in case 4 (of Kovacic's algorithm). We start by
analyzing the case 1: by conditions $c_2$ and $\infty_3$ we have
that

$$\left[ \sqrt {r}\right] _{0}=0,\quad\alpha_{0}^{\pm}={1\pm 2\sqrt{-\lambda}\over 2}, \quad\left[  \sqrt{r}\right]  _{\infty}=1,\quad
{\rm and}\quad \alpha_{\infty}^{\pm }=\mp{1\over 2}.$$ By step 2
we have the following possibilities for $n\in\mathbb{Z}_+$ and for
$\lambda\in\Lambda$:
$$
\begin{array}{lll}
\Lambda_{++})\quad& n=\alpha^+_\infty -
\alpha^+_0=-1-\sqrt{-\lambda}, &
\lambda=-\left(n+1\right)^2,\\&&\\
\Lambda_{+-}) & n=\alpha^+_\infty -
\alpha^-_0=-1+\sqrt{-\lambda},& \lambda=-\left(n+1\right)^2,
\\&&\\
\Lambda_{-+}) & n=\alpha^-_\infty - \alpha^+_0=-\sqrt{-\lambda},&
\lambda=-n^2,\\&&\\ \Lambda_{--}) & n=\alpha^-_\infty -
\alpha^-_0=\sqrt{-\lambda},& \lambda=-n^2.
\end{array}
$$
We can see that $\lambda\in\Lambda_-=\{-n^2:n\in\mathbb{Z}_+\}$.
Now, for $\lambda\in\Lambda$, the rational function $\omega$ is
given by:
$$
\begin{array}{llll}
\Lambda_{++})\quad& \omega=1+{3+2n\over 2
z},&\lambda\in\Lambda_{++},& r_n={4n^2+8n+3\over 4 z^2} +
{2n+3\over z}+1,\\&&&\\
\Lambda_{+-}) & \omega=1-{1+2n\over 2 z},&\lambda\in\Lambda_{+-},&
r_n={4n^2+8n+3\over 4 z^2} -
{2n+1\over z}+1,\\&&&\\
\Lambda_{-+}) & \omega=-1+{1+2n\over 2
z},&\lambda\in\Lambda_{-+},& r_n={4n^2-1\over 4 z^2} -{2n+1\over
z}+1,\\&&&\\ \Lambda_{--}) & \omega=-1+{1-2n\over 2
z},&\lambda\in\Lambda_{--},& r_n={4n^2-1\over 4 z^2} + {2n-1\over
z}+1,
\end{array}
$$
where $r_n$ is the coefficient of the differential equation
$\partial_z^2\Phi=r_n\Phi$.
\\

\noindent By step 3, there exists a polynomial of degree $n$
satisfying the relation (\ref{recu1}),
$$
\begin{array}{llll}
\Lambda_{++})\quad& \partial_z^2 \widehat{P}_n+2\left(1+{3+2n\over
2 z}\right)\partial_z \widehat{P}_n+{2(n+2)\over  z}\widehat{P}_n&=&0,\\&&&\\
\Lambda_{+-}) & \partial_z^2 \widehat{P}_n+2\left(1-{1+2n\over
2 z}\right)\partial_z \widehat{P}_n+{2(-n)\over  z}\widehat{P}_n&=&0,\\&&&\\
\Lambda_{-+}) & \partial_z^2 \widehat{P}_n+2\left(-1+{1+2n\over 2
z}\right)\partial_z \widehat{P}_n+{2(-n)\over  z}\widehat{P}_n&=&0,\\&&&\\
\Lambda_{--}) &
\partial_z^2 \widehat{P}_n+2\left(-1+{1-2n\over 2 z}\right)\partial_z \widehat{P}_n+{2n\over  z}\widehat{P}_n&=&0.
\end{array}
$$

These polynomials only exist for $n=\lambda=0$, with
$\lambda\in\Lambda_{-+}\cup\Lambda_{--}$. So  the solutions of
$H\Psi=0$, $\widehat{H}\widehat{\Psi}=0$ and
$\widehat{\mathbf{H}}\Phi=0$ are given by
$$
\Phi_0=\sqrt{ z}e^{- z},\quad \widehat{\Psi}_0=e^{- z},\quad
\Psi=e^{- e^{-x}}.
$$
The wave function $\Psi_0$ satisfy the bound state conditions,
which means that is ground state (see \cite{duka}) and
$0\in\mathrm{Spec}_p(H)$. Furthermore, we have

$$\mathrm{DGal}(L_0/{K})=\mathrm{DGal}(\widehat{L}_0/\widehat{K})=\mathrm{DGal}(\widehat{\mathbf{L}}_0/{\mathbb{C}(z)})=\mathbb{B},$$
$$\mathcal{E}(H)=\mathcal{E}(\widehat{H})=\mathcal{E}(\widehat{\mathbf{H}})=\mathrm{Vect}(1).$$

We proceed with case $2$. The conditions $c_2$ and $\infty_3$
are satisfied. So we have
$$E_c=\left\{2,2+4\sqrt{-\lambda},2-4\sqrt{-\lambda}\right\}\quad
\textrm{and}\quad E_{\infty}=\{0\},$$ and by step two, we have
that $2\pm\sqrt{-\lambda}=m\in\mathbb{Z}_+$, so that
$\lambda=-\left(\frac{m+1}{2}\right)^2$ and the rational function
$\theta$ has the following possibilities
$$\theta_+={2+m\over  z},\quad \theta_-=-{m\over  z}.$$

\noindent By step three, there exist a monic polynomial of degree
$m$ satisfying the recurrence relation (\ref{recu2}):
$$
\begin{array}{llll}
\theta_{+})& \partial_z^3 \widehat{P}_m+{3m+6\over  z}\partial_z^2
\widehat{P}_m-{4· z^2 - 4 z - 2·m^2 - 7·m -6\over  z^2}\partial_z
\widehat{P}_m-{ 4·m· z +8· z - 4·m - 6\over
 z^2}\widehat{P}_m&=&0,\\&&&\\ \theta_{-}) & \partial_z^3 \widehat{P}_m-{3m\over  z}\partial_z^2 \widehat{P}_m-{4· z^2
- 4· z - 2m^2-m\over  z^2}\partial_z \widehat{P}_m+{ 4·m· z - 4·m
- 2\over
 z^2}\widehat{P}_m&=&0.
\end{array}
$$
We can see that, for $m=1$, the polynomial exists only in the case
$\theta_-$, with $\widehat{P}_1=z-1/2$. In general, these
polynomials could exist only in the case $\theta_-$ with
$m=2n-1$, $n\geq 1$, that is $\lambda\in\{-n^2:n\geq 1\}$.
\\

By case one and case two, we obtain
$\Lambda=\{-n^2:n\geq 0\}=\mathrm{Spec}_p(H)$. Now, the rational
function $\phi$ and the quadratic expression for $\omega$ are
$$
\phi=-{2n-1\over z}+{\partial_z \widehat{P}_{2n-1}\over
\widehat{P}_{2n-1}},\quad
\omega^2+M\omega+N=0,\quad\omega=\frac{-M\pm\sqrt{M^2-4N}}{2},
$$
where the coefficients $M$ and $N$ are given by
$$
M={2n-1\over z}-{\partial_z \widehat{P}_{2n-1}\over
\widehat{P}_{2n-1}},\quad N={n^2-n+\frac14\over
z^2}-{(2n-1){\partial_z \widehat{P}_{2n-1}\over
\widehat{P}_{2n-1}}-2\over z}+{\partial_z^2
\widehat{P}_{2n-1}\over \widehat{P}_{2n-1}}-2.
$$
Now, $\triangle=M^2-4N\neq 0$, which means that
$\mathbf{\widehat{H}}\Phi=-n^2\Phi$ with $n\in\mathbb{Z}^+$ has
two solutions given by Kovacic's algorithm:
$$\Phi_{1,n}={\sqrt{ z}\widehat{P}_ne^{- z}\over  z^n},\quad \Phi_{2,n}={\sqrt{ z}\widehat{P}_{n-1}e^{ z}\over  z^n}.$$
The solutions of $\widehat{H}\widehat{\Psi}=-n^2\widehat{\Psi}$
are given by
$$\widehat{\Psi}_{1,n}={\widehat{P}_ne^{- z}\over  z^n},\quad \widehat{\Psi}_{2,n}={\widehat{P}_{n-1}e^{ z}\over  z^n},$$
and therefore, the solutions of the Schr\"odinger equation
$H\Psi=-n^2\Psi$ are

$$\Psi_{1,n}=P_ne^{-e^{-x}}e^{nx},\quad
 \Psi_{2,n}=P_{n-1}e^{e^{-x}}e^{nx},\quad P_n=\widehat{P}_n\circ z.$$

The wave functions $\Psi_{1,n}=\Psi_n$ satisfy the conditions of
bound state and, for $n=0$, this solution coincides with the
ground state presented above. Therefore we have
$$\Phi_n=\Phi_0\widehat{f}_n\widehat{P}_n,\quad \widehat{\Psi}_n=\widehat{\Psi}_0\widehat{f}_n\widehat{P}_n,\quad \widehat{f}_n( z)=\frac{1}{
z^{n}}.$$ Thus, the bound states wave functions are obtained as
$$\Psi_n=\Psi_0f_nP_n,\quad
f_n(x)=\widehat{f}_n(e^{-x})=e^{nx}.$$ The eigenring s and
differential Galois groups for $n>0$ satisfies
$$\mathrm{DGal}(L_n/K)=\mathrm{DGal}(\widehat{L}_n/{\widehat{K}})=\mathrm{DGal}(\widehat{\mathbf{L}}_n/{\mathbb{C}(z)})=\mathbb{G}_m,$$
$$\dim_{\mathbb{C}}\mathcal E(\widehat{\mathbf{H}}+n^2)= \dim_{\mathbb{C}}\mathcal E(\widehat{{H}}+n^2)=\dim_{\mathbb{C}}\mathcal E(H+n^2)=2.$$

We remark that the Schr\"odinger equation with Morse potential,
under suitable changes of variables \cite{lali}, may be transformed to a
Bessel differential equation.
\\

It is known that equations deriving from  Eckart, Rosen-Morse, Scarf and P\"oschl-Teller
potentials  may be mapped, under suitable transformations, to
Hypergeometric equations. These potentials are inter-related by
point canonical coordinate transformations (see \cite[p.
314]{cokasu} ), so that $\Lambda=\mathbb{C}$ because
P\"oschl-Teller potential is obtained by means of Darboux
transformations of $V=0$ (\cite{masa,ro}).
 We consider some particular cases of Eckart, Scarf and Poschl-Teller potentials for which we apply only case 1 of Kovacic's algorithm. Case 1 allows us
 to obtain the enumerable set $\Lambda_n\subset \Lambda$, which includes the classical results obtained by means of supersymmetric quantum mechanics.
 Cases 2 and 3 of Kovacic algorithm also can be applied, but are not considered here.\\

 \textbf{Eckart potential:} $V(x)=4\coth(x)+5$, $x>0$.\\

The Schr\"odinger equation $H\Psi=\lambda\Psi$ is
$$\partial_x^2\Psi=\left(4\coth(x)+5-\lambda\right)\Psi.$$

By the Hamiltonian change of variable $z=z(x)=\coth(x)$, we obtain
$$\alpha(z)=(1-z^2)^2,\quad \widehat V(z)=4z+5,\quad \widehat{\mathbf{V}}(z)=\frac4{(z + 1)·(z - 1)^2}.$$
 Thus, $\widehat{K}=\mathbb{C}(z)$ and $K=\mathbb{C}(\coth(x))$.
The algebrized Schr\"odinger equation
$\widehat{H}\widehat{\Psi}=\lambda\widehat{\Psi}$ is
$$(1-z^2)^2\partial_z^2\widehat{\Psi}-2z(1-z^2)\partial_z\widehat{\Psi}-(4z+5-\lambda)\widehat{\Psi}=0$$ and the reduced
algebrized Schr\"odinger equation
$\widehat{\mathbf{H}}\widehat{\Phi}=\lambda\widehat{\Phi}$ is
\begin{displaymath}\partial_z^2\Phi=r\Phi,\quad r={\frac {4z+4-\lambda}{ \left(z-1 \right) ^{2} \left(z+1 \right) ^{2}}}=
 {\frac {2-\frac{\lambda}4}{\left(  z-1 \right) ^{2}}}+{\frac {\frac{\lambda}4-1}{(
 z-1)}}+
{\frac {-\frac{\lambda}4}{\left(  z+1 \right) ^{2}}}+{\frac
{1-\frac{\lambda}4}{( z+1)}}\end{displaymath}

We can see that this equation could fall in any case of Kovacic's
algorithm. Considering $\lambda=0$, the conditions
$\{c_1,c_2,\infty_1\}$  of case 1 are satisfied, obtaining $$
[\sqrt{r}]_{-1}=[\sqrt{r}]_{1}=[\sqrt{r}]_{\infty}=\alpha^+_{\infty}=0,
\quad\alpha^{\pm}_{-1}=\alpha^-_{\infty}=1,\quad
\alpha^+_{1}=2,\quad\alpha^-_{1}=-1. $$ By step 2, the elements of
$D$ are $0$ and $1$. The rational function $\omega$ for $n=0$ and
for $n=1$ must be
$$\omega={1\over  z+1}+{-1\over  z-1}.$$
By step 3, we search for a monic polynomial of degree $n$ satisfying
the relation (\ref{recu1}). Starting with $n=0$ the only one
possibility is $\widehat{P}_0( z)=1$, which effectively satisfy
the relation (\ref{recu1}), while $\widehat{P}_1( z)= z+a_0$ does
not exist.  
We have obtained one solution using
Kovacic algorithm:
$$\Phi_0={ z+1\over  z-1},\quad \widehat{\Psi}_0=\sqrt{ z+1\over ( z-1)^3},$$
this means that $0\in\Lambda_n$. We can obtain the second solution
using the first solution:
$$\Phi_{0,2}={\frac {{ z}^{2}+ z-4-4\,\ln  \left(  z+1 \right)  z-4\,\ln  \left(  z+1
 \right) }{ z-1}},\quad \widehat{\Psi}_{0,2}={\Phi_{0,2}\over \sqrt{
 z^2-1}}.$$ Furthermore the differential Galois groups and eigenring s for $\lambda=0$ are
 $$\mathrm{DGal}(\widehat{\mathbf{L}}_0/{\mathbb{C}(z)})=\mathbb{G}_a,\quad
 \mathrm{DGal}(L_0/{K})=\mathrm{DGal}(\widehat{L}_0/{\widehat{K}})=\mathbb{G}^{\{2\}},$$
$$\mathcal E(\widehat{\mathbf{H}})=\mathrm{Vect} \left(1,{(z+1)^2\over (1-z)^2}\partial_z+{2(z+1)\over (1-z)^3}\right),$$
 $$ \mathcal E(\widehat{{H}})=\mathrm{Vect}\left(1,{(z+1)^2\over (1-z)^2}\partial_z-{z^2+3z+2\over (1-z)^3}\right),$$
 $$\mathcal E(H)=\mathrm{Vect}\left(1,{(\coth(x)+1)^2\over (1-\coth(x))^2(1-\coth^2(x))}\partial_x-{\coth^2(x)+3\coth(x)+2\over (1-\coth(x))^3}\right).$$

 Now, for $\lambda\neq 0$, the conditions $\{c_2,\infty_1\}$ of case 1 are
 satisfied:

$$\begin{array}{l}
[\sqrt{r}]_{-1}=[\sqrt{r}]_{1}=[\sqrt{r}]_{\infty}=\alpha^+_{\infty}=0,
\quad\alpha^-_{\infty}=1,\\ \\
\alpha^{\pm}_{-1}={1\pm\sqrt{1-\lambda}\over 2
},\quad\alpha^{\pm}_{1}={1\pm\sqrt{9-\lambda}\over 2 }.
\end{array}$$ By step 2 we have the following possibilities for
$n\in\mathbb{Z}_+$ and for $\lambda\in\Lambda$:
$$
\begin{array}{lll}
\Lambda_{++-}) & n=\alpha^+_\infty -
\alpha^+_{-1}-\alpha^-_{1}=-1-{\sqrt{1-\lambda}-\sqrt{9-\lambda}\over
2},& \lambda=4-\frac4{(n + 1)^2} -n^2 - 2·n, \\&&\\
\Lambda_{+--}) & n=\alpha^+_\infty -
\alpha^-_{-1}-\alpha^-_{1}=-1+{\sqrt{1-\lambda}+\sqrt{9-\lambda}\over
2},& \lambda=4-\frac4{(n + 1)^2} - n^2 - 2·n, \\&& \\
\Lambda_{-+-}) & n=\alpha^-_\infty -
\alpha^+_{-1}-\alpha^-_{1}={\sqrt{1-\lambda}+\sqrt{9-\lambda}\over
2},& \lambda=5-\frac4{n^2} -n^2, \\&&\\ \Lambda_{---}) &
n=\alpha^-_\infty -
\alpha^-_{-1}-\alpha^-_{1}={\sqrt{1-\lambda}+\sqrt{9-\lambda}\over
2},& \lambda=5-\frac4{n^2} - n^2.
\end{array}
$$
Therefore, we have 
$$\Lambda_n\subseteq\left\{4-\frac4{(n + 1)^2} - n^2- 2·n :n\in\mathbb{Z}_+\right\}\cup
\left\{5-\frac4{n^2} - n^2:n\in\mathbb{Z}_+\right\}.$$

Now, for $\lambda\in\Lambda$, the rational function $\omega$ is
given by: {\small$$
\begin{array}{lll}
\Lambda_{++-})\quad& \omega={ z·(n - 1) - n^2 - 2·n - 1\over (n +
1)·( z + 1)·( z - 1)},& r_n={-2· z^2·(n - 1) +4· z·(n + 1)^2 +(n +
1)·(n^3 + 3·n^2 + 2·n + 2)\over(n + 1)^2·( z +
1)^2·( z - 1)^2},\\&&\\
\Lambda_{+--}) & \omega={n· z·(n + 1) + 2\over (n + 1)·( z + 1)·(1
-  z)},& r_n={n· z^2·(n + 1)^3 + 4 z·(n + 1)^2 + n^3 +
2·n^2 + n + 4\over(n + 1)^2·( z + 1)^2·( z - 1)^2},\\&&\\
\Lambda_{-+-}) & \omega={ z·(n - 2) - n^2\over n·( z + 1)·( z -
1)},& r_n={-2· z^2·(n - 2) + 4·n^2· z +n·(n^3 - n + 2)\over n^2·(
z + 1)^2·( z - 1)^2},\\&&\\ \Lambda_{---}) & \omega={n· z·(n - 1)
+ 2\over n·( z + 1)·(1 -  z)},& r_n={n^3· z^2·(n - 1) + 4·n^2· z +
n^3 - n^2 + 4\over n^2·( z + 1)^2·( z - 1)^2},
\end{array}
$$}
where $r_n$ is the coefficient of the differential equation
$\partial_z^2\Phi=r_n\Phi$.
\\

\noindent By step 3, there exists a monic polynomial of degree $n$
satisfying the relation (\ref{recu1}),
$$
\begin{array}{llll}
\Lambda_{++-})\quad& \partial_z^2\widehat P_n+2\left({ z·(n - 1) -
n^2 - 2·n - 1\over (n + 1)·( z + 1)·( z -
1)}\right)\partial_z\widehat P_n+{2·(1 -
n)\over((n + 1)^2·( z + 1)·( z - 1)}\widehat{P}_n&=&0,\\&&&\\
\Lambda_{+--}) & \partial_z^2\widehat P_n+2\left({n· z·(n + 1) +
2\over (n + 1)·( z + 1)·(1 -  z)}\right)\partial_z\widehat
P_n+{n·(n + 1)\over ( z +
1)·( z - 1)}\widehat{P}_n&=&0,\\&&&\\
\Lambda_{-+-}) & \partial_z^2\widehat P_n+2\left({ z·(n - 2) -
n^2\over n·( z + 1)·( z - 1)}\right)\partial_z\widehat P_n+{2·(2 -
n)\over n^2·( z + 1)·( z - 1)}\widehat{P}_n&=&0,\\&&&\\
\Lambda_{---}) &
\partial_z^2\widehat P_n+2\left({n· z·(n - 1) + 2\over n·( z +
1)·(1 - z)}\right)\partial_z\widehat P_n+{n·(n - 1)\over ( z +
1)·( z - 1)}\widehat{P}_n&=&0.
\end{array}
$$
The only one case in which there exists such a polynomial
$\widehat{P}_n$ of degree $n$ is for $\Lambda_{+--})$. The
solutions of the equation $\widehat{\mathbf{H}}\Phi=\lambda\Phi$,
with $\lambda\neq 0$, are:
$$
\begin{array}{llll}
\Lambda_{++-})\quad&
\Phi_n=\widehat{P}_n\widehat{f}_n\Phi_0,&\Phi_0={1\over z
-1}&\widehat{f}_n=( z - 1)^{n·(1 - n)´\over 2·(n + 1)}·( z +
1)^{n·(n +
3)\over 2·(n + 1)},\\&&&\\
\Lambda_{+--}) & \Phi_n=\widehat{P}_n\widehat{f}_n\Phi_0,&\Phi_0={
z+1\over z -1}&f_n=( z - 1)^{n·(1 - n)\over 2·(n + 1)}·( z + 1)^{-
n·(n +
3)\over 2·(n + 1)},\\&&&\\
\Lambda_{-+-}) &
\Phi_n=\widehat{P}_n\widehat{f}_n\Phi_1,&\Phi_1={1\over  z
-1}&\widehat{f}_n=( z + 1)^{n^2 + n - 2\over 2·n}·( z - 1)^{-n^2 +
3·n - 2 \over 2·n},\\&&&\\ \Lambda_{---}) &
\Phi_n=\widehat{P}_n\widehat{f}_n\Phi_1,&\Phi_1={ z+1\over  z
-1}&\widehat{f}_n=( z + 1)^{-n^2 - n + 2 \over 2·n}·( z - 1)^{-n^2
+ 3·n- 2 \over 2·n}.
\end{array}
$$
In any case $\widehat{\Psi}_n={\Phi_n\over 1- z^2}$. The case
$\Lambda_{+--})$ includes the classical results obtained by means
of supersymmetric quantum mechanics. Thus, replacing $z$ by $\coth
(x)$ we obtain the eigenstates $\Psi_n$. The eigenring s and
differential Galois groups for $n>0$ and $\lambda\in\Lambda_n$
satisfy
$$\mathrm{DGal}(L_\lambda/{K})\subseteq \mathbb{G}^{\{2m\}},\quad \mathrm{DGal}(\widehat{L}_\lambda/{\widehat{K}})\subseteq
\mathbb{G}^{\{2m\}},$$
$$\mathrm{DGal}(\widehat{\mathbf{L}}_\lambda/{\mathbb{C}(z)})=\mathbb{G}_m,$$
$$\dim_{\mathbb{C}(z)}\mathcal E(\widehat{\mathbf{H}}+\lambda)= 2,\quad \mathcal E(\widehat{{H}}+\lambda)=\mathcal E(H+\lambda)=\mathrm{Vect}(1).$$
\medskip

 \textbf{Scarf potential:} $V(x)=\frac{\sinh^2x-3\sinh x}{\cosh^2 x}.$\\

The Schr\"odinger equation $H\Psi=E\Psi$ is
$$\partial_x^2\Psi=\left({\sinh^2x-3\sinh x\over \cosh^2 x}-E\right)\Psi.$$

By the Hamiltonian change of variable $z=z(x)=\sinh(x)$, we obtain
$$\alpha(z)=1+z^2,\quad \widehat V(z)={z^2-3z\over 1+z^2}.$$
 Thus, $\widehat{K}=\mathbb{C}(z,\sqrt{1+z^2})$ and $K=\mathbb{C}(\sinh(x),\cosh(x))$.
The reduced algebrized Schr\"odinger equation
$\widehat{\mathbf{H}}{\Phi}=\lambda{\Phi}$ is
$$\partial_z^2=\left({\lambda z^2 -12 z
 +\lambda-1\over 4·( z^2 + 1)^2}\right)\Phi, \quad \lambda=3-4E.$$

Applying Kovacic's algorithm for this equation with $\lambda=0$,
we see that it does not satisfy case 1. So we consider only
$\lambda\neq 0 $.  By conditions $\{c_2,\infty_2\}$ of case 1 we
have 

$$\begin{array}{l}
[\sqrt{r}]_{-\mathrm{i}}=[\sqrt{r}]_{\mathrm{i}}=[\sqrt{r}]_{\infty}=0,
\quad\alpha^\pm_{\infty}={1\pm\sqrt{1+\lambda}\over 2},\\ \\
\alpha^{+}_{-\mathrm{i}}=\frac54-\frac{\mathrm{i}}2,\quad\alpha^-_{-\mathrm{i}}=-\frac{1}4+\frac{\mathrm{i}}2\quad
\alpha^{+}_{\mathrm{i}}=\frac54+\frac{\mathrm{i}}2,\quad\alpha^-_{\mathrm{i}}=-\frac{1}4-\frac{\mathrm{i}}2
.
\end{array}$$ By step 2 we have the following possibilities for
$n\in\mathbb{Z}_+$ and for $\lambda\in\Lambda$:
$$
\begin{array}{lll}
\Lambda_{+++}) & n=\alpha^+_\infty -
\alpha^+_{-\mathrm{i}}-\alpha^+_{\mathrm{i}}={\sqrt{\lambda+1}-4\over
2},& \lambda=4n^2 + 16n + 15, \\&&\\ \Lambda_{+--}) &
n=\alpha^+_\infty -
\alpha^-_{-\mathrm{i}}-\alpha^-_{\mathrm{i}}={\sqrt{\lambda+1}+2\over
2},& \lambda=4n^2 - 8n + 3,
\end{array}
$$
thus obtaining  
$$\Lambda_n\subseteq\left\{4n^2 + 16n + 15 :n\in\mathbb{Z}_+\right\}\cup
\left\{4n^2 -8n+3:n\in\mathbb{Z}_+\right\}.$$

Now, the rational function $\omega$ is given by:
$$
\Lambda_{+++})\quad\quad \omega={5· z - 2\over 2·( z^2 +
1)},\qquad \Lambda_{+--})\quad {2 -  z \over 2·( z^2 + 1)}.
$$
By step 3, there exists $\widehat{P}_0=1$ and a polynomial of
degree $n\geq 1$ should satisfy either of the relation
(\ref{recu1}),
$$
\begin{array}{llll}
\Lambda_{+++})\quad& \partial_z^2\widehat P_n+{5· z - 2\over  z^2
+ 1}\partial_z\widehat P_n-{n· z^2·(n
+ 4) + 3· z + n^2 + 4·n - 3\over ( z^2 + 1)^2}\widehat{P}_n&=&0,\\&&&\\
\Lambda_{+--}) & \partial_z^2\widehat P_n+{2 -  z \over  z^2 +
1}\partial_z\widehat P_n-{n· z^2·(n - 2) + 3· z + n^2 - 2·n -
3\over ( z^2 + 1)^2}\widehat{P}_n&=&0.
\end{array}
$$
In both cases there exists such a polynomial $\widehat{P}_n$ of
degree $n\geq 1$. Basis of solutions $\{\Phi_{1,n},\Phi_{2,n}\}$
of the reduced algebrized Schr\"odinger equation are:
$$
\begin{array}{llll}
\Lambda_{+++})\quad&
\Phi_{1,n}=\widehat{P}_n\widehat{f}_n\Phi_{1,0},&\Phi_{1,0}=(1+
z^2)^\frac54e^{-\arctan z},&\widehat{f}_n=1,\\&&&\\
\quad&
\Phi_{2,n}=\widehat{Q}_n\widehat{g}_n\Phi_{2,0},&\Phi_{2,0}={22+21x+12x^2+6x^3\over
\sqrt[4]{1+ z^2}}e^{-\arctan z},&\widehat{g}_n=1.\\&&&\\
\Lambda_{+--}) &
\Phi_{1,n}=\widehat{P}_n\widehat{f}_n\Phi_{1,0},&\Phi_{1,0}={1\over
\sqrt[4]{1+ z^2}}e^{\arctan z},&\widehat{f}_n=1,\\&&&\\
&
\Phi_{2,n}=\widehat{Q}_n\widehat{g}_n\Phi_{2,0},&\Phi_{2,0}={1\over
\sqrt[4]{1+ z^2}}e^{\arctan z}\int{1\over \sqrt{1+
z^2}}e^{-2\arctan z}d z,&\widehat{g}_n=1.
\end{array}
$$
\medskip

In both cases $\widehat{\Psi}={\Phi\over \sqrt[4]{1+ z^2}}$, but
the classical case (see references \cite{cokasu,duka}) is
$\Lambda_{+--})$, so that replacing $ z$ by $\sinh x$ and
$\lambda$ by $3-4E$ we obtain the eigenstates $\Psi_n$.
\\

The eigenring s and differential Galois groups are
$$\mathcal E(H-\lambda)=\mathcal E(\widehat{H}-\lambda)=\mathcal E(\widehat{\mathbf{H}}-\lambda)=\mathrm{Vect}(1),$$
$$\mathrm{DGal}(L_\lambda/K)=\mathrm{DGal}(\widehat{L}_\lambda/{\widehat{K}})=\mathrm{DGal}(\widehat{\mathbf{L}}_\lambda/{\mathbb{C}(x)})=\mathbb{B}.$$
\\

\textbf{P\"oschl-Teller potential (revisited):}
$V(\mathrm{r})={\cosh^4(x)-\cosh^2(x) +2\over \sinh^2(x)\cosh^2
(x)},$ $x>0$. The reduced algebrized Schr\"odinger equation
$\widehat{\mathbf{H}}\Phi=E\Phi$ is
$$\partial_z^2\Phi=\left({\lambda z^4· -
(\lambda+3) z^2 + 8\over 4· z^2( z^2 - 1)^2}\right)\Phi, \quad
\lambda=3-4E.$$

Considering $\lambda=0$ and starting with the conditions
$\{c_2,\infty_1\}$ of case 1, we obtain
$$\begin{array}{l}
[\sqrt{r}]_{0}=[\sqrt{r}]_{-1}=[\sqrt{r}]_{1}=[\sqrt{r}]_{\infty}=\alpha^+_{\infty}=0,\quad
\alpha^-_{\infty}=1,
\\
\\
 \alpha^+_{-1}=\alpha^+_{1}=\frac54,\quad
 \alpha^-_{-1}=\alpha^-_{1}=-\frac14,\quad \alpha^+_{0}=2,\quad\alpha^-_{0}=-1.\end{array}$$ By step 2, the
elements of $D$ are $0$ and $1$. For the rational function $\omega$,
we have the following possibilities for $n=0$ and for $n=1$:
$$\begin{array}{lll}
\Lambda_{++--}) & n=0,&\omega={5/4\over  z+1}+{-1/4\over
 z-1}+{-1\over  z},\\& &\\
\Lambda_{+-+-}) & n=0,&\omega={-1/4\over  z+1}+{5/4\over
 z-1}+{-1\over  z},\\ & & \\
\Lambda_{-+--}) & n=1,&\omega={5/4\over  z+1}+{-1/4\over
 z-1}+{-1\over  z},\\& &\\ \Lambda_{--+-}) &
n=1,&\omega={-1/4\over  z+1}+{5/4\over  z-1}+{-1\over  z}.
\end{array}
$$ 
Following step $3$, we search for a monic polynomial of degree $n$
satisfying   relation (\ref{recu1}). Starting with $n=0$ the
only   possibility for $\Lambda_{++--})$ and $\Lambda_{+-+-})$
is $\widehat{P}_0( z)=1$, which does not satisfy    relation
(\ref{recu1}) in both cases, while $\widehat{P}_1( z)= z+a_0$
effectively does exist, taking $a_0=-{2\over 3}$ for
$\Lambda_{-+--})$ and $a_0={2\over 3}$ for $\Lambda_{--+-})$. 
This way, we have obtained two solutions ($\Phi_{1,0}$,
$\Phi_{2,0}$) using Kovacic's algorithm:
$$\begin{array}{ll}\Phi_{1,0}=\left(1-\frac2{3 z}\right)\sqrt[4]{( z+1)^5\over  z-1},&\widehat{\Psi}_{1,0}=\left(1-\frac2{3 z}\right){ z+1\over \sqrt{ z-1}},\\
& \\
{\Phi}_{2,0}=\left(1+\frac2{3 z}\right)\sqrt[4]{( z-1)^5\over
 z+1},&\widehat{\Psi}_{2,0}=\left(1+\frac2{3 z}\right){ z-1\over
\sqrt{ z+1}} .\end{array}$$ 
This means that $0\in\Lambda_n$.
Furthermore,
$$\mathrm{DGal}(\widehat{\mathbf{L}}_0/{\mathbb{C}(x)})=\mathbb{G}^{[4]},\quad \mathrm{DGal}(\widehat{{L}}_0/{\widehat{K}})=
\mathrm{DGal}(L_0/{K})=e,$$
$$\dim_{\mathbb{C}}\mathcal E(\widehat{\mathbf{H}})=2,\quad \dim_{\mathbb{C}}\mathcal E(\widehat{H})=\dim_{\mathbb{C}} \mathcal E(H)=4.$$

 Now, for $\lambda\neq 0$, we see that conditions $\{c_2,\infty_1\}$ of case 1 leads us to

$$\begin{array}{l}
[\sqrt{r}]_{0}=[\sqrt{r}]_{-1}=[\sqrt{r}]_{1}=[\sqrt{r}]_{\infty}=0,
\quad\quad\alpha^\pm_{\infty}={1\pm\sqrt{1+\lambda}\over 2},\\ \\
 \alpha^+_{-1}=\alpha^+_{1}=\frac54,\quad
 \alpha^-_{-1}=\alpha^-_{1}=-\frac14,\quad \alpha^+_{0}=2,\quad\alpha^-_{0}=-1.
\end{array}$$ By step 2 we have the following possibilities for
$n\in\mathbb{Z}_+$ and for $\lambda\in\Lambda$:
$$
\begin{array}{lll}
\Lambda_{++++}) & n=\alpha^+_\infty -
\alpha^+_{-1}-\alpha^+_{1}-\alpha^+_0={\sqrt{\lambda+1}-8\over
2},& \lambda=4·n^2 + 32·n + 63, \\&&\\
\Lambda_{+++-}) &n=\alpha^+_\infty -
\alpha^+_{-1}-\alpha^+_{1}-\alpha^-_0={\sqrt{\lambda+1}-2\over
2},& \lambda=4·n^2 + 8·n + 3, \\&&\\
\Lambda_{++-+}) & n=\alpha^+_\infty -
\alpha^+_{-1}-\alpha^-_{1}-\alpha^+_0={\sqrt{\lambda+1}-5\over
2},& \lambda=4·n^2 + 20·n + 24, \\&&\\
\Lambda_{++--}) & n=\alpha^+_\infty -
\alpha^+_{-1}-\alpha^-_{1}-\alpha^-_0={\sqrt{\lambda+1}+1\over
2},& \lambda=4·n^2 -4n, \\&&\\
\Lambda_{+-++}) & n=\alpha^+_\infty -
\alpha^-_{-1}-\alpha^+_{1}-\alpha^+_0={\sqrt{\lambda+1}-5\over
2},& \lambda=4·n^2 + 20·n + 24, \\&&\\
\Lambda_{+-+-}) &n=\alpha^+_\infty -
\alpha^-_{-1}-\alpha^+_{1}-\alpha^-_0={\sqrt{\lambda+1}+1\over
2},& \lambda=4·n^2 -4n, \\&&\\
\Lambda_{+--+}) & n=\alpha^+_\infty -
\alpha^-_{-1}-\alpha^-_{1}-\alpha^+_0={\sqrt{\lambda+1}-2\over
2},& \lambda=4·n^2 + 8·n + 3, \\&&\\
\Lambda_{+---}) & n=\alpha^+_\infty -
\alpha^-_{-1}-\alpha^-_{1}-\alpha^-_0={\sqrt{\lambda+1}+4\over
2},& \lambda=4·n^2 -16n + 15,
\end{array}
$$
obtaining
$\Lambda_n\subseteq\Lambda_a\cup\Lambda_b\cup\Lambda_c\cup\Lambda_d\cup\Lambda_e,$
where
$$\begin{array}{ll}
\Lambda_a=\left\{4n^2 + 32·n +63:n\in\mathbb{Z}_+\right\},&
\Lambda_b=\left\{4n^2+8n+3:n\in\mathbb{Z}_+\right\}\\&
\\
\Lambda_c=\left\{4n^2 + 20·n +24:n\in\mathbb{Z}_+\right\},&
\Lambda_d=\left\{4n^2-4n:n\in\mathbb{Z}_+\right\},\\&
\\
\Lambda_e=\left\{4n^2 - 16·n +15:n\in\mathbb{Z}_+\right\}.&
\end{array}$$

Now, the rational function $\omega$ is given by, respectively:
$$\begin{array}{llll}
\Lambda_{++++}) &\omega={5/4\over  z+1}+{5/4\over  z-1}+{2\over
 z},& \Lambda_{+++-}) &\omega={5/4\over  z+1}+{5/4\over
 z-1}+{-1\over  z},\\ && & \\
\Lambda_{++-+}) &\omega={5/4\over  z+1}+{-1/4\over  z-1}+{2\over
 z}, & \Lambda_{++--}) & \omega={5/4\over  z+1}+{-1/4\over
 z-1}+{-1\over  z},\\& &\\ \Lambda_{+-++}) &\omega={-1/4\over
 z+1}+{5/4\over  z-1}+{2\over  z},&\Lambda_{+-+-})
&\omega={-1/4\over  z+1}+{5/4\over  z-1}+{-1\over  z},\\&&&\\
\Lambda_{+--+})&\omega={-1/4\over  z+1}+{-1/4\over  z-1}+{2\over
 z},& \Lambda_{+---}) &\omega={-1/4\over  z+1}+{-1/4\over
 z-1}+{-1\over  z}.
\end{array}
$$

By step 3, there exists a monic polynomial of degree $n$
satisfying  relation (\ref{recu1}):
$$
\begin{array}{ll}
\Lambda_{++++}) & \partial_z^2\widehat P_n+{(3· z + 2)·(3· z -
2)\over  z·( z + 1)·( z -
1)}\partial_z\widehat P_n+{n(n + 8)\over ( z + 1)·(1 -  z)}P_n=0, \\& \\
\Lambda_{+++-}) &\partial_z^2\widehat P_n+{3· z^2 + 2\over  z·( z + 1)·( z - 1)}\partial_z\widehat P_n+{n·(n + 2)\over( z + 1)·(1 -  z)}\widehat{P}_n=0, \\& \\
\Lambda_{++-+}) & \partial_z^2\widehat P_n+{6· z^2 - 3·x - 4\over  z·( z + 1)·( z - 1)}\partial_z\widehat P_n+{n· z·(n + 5) + 6\over  z·( z + 1)·(1 -  z)}\widehat{P}_n=0, \\& \\
\Lambda_{++--}) & \partial_z^2\widehat P_n+{3 z - 2\over  z·( z + 1)·(1 -  z)}\partial_z\widehat P_n+{n· z·(n - 1) - 12\over 4· z·( z + 1)·(1 -  z)}\widehat{P}_n=0, \\& \\
\Lambda_{+-++}) & \partial_z^2\widehat P_n+{6· z^2 + 3· z - 4\over  z·( z + 1)·( z - 1)}\partial_z\widehat P_n+{ z·(n^2 + 5) - 6\over  z·( z + 1)·(1 -  z)}\widehat{P}_n=0, \\& \\
\Lambda_{+-+-}) &\partial_z^2\widehat P_n+{3· z + 2\over  z·( z + 1)·( z - 1)}\partial_z\widehat P_n+{n· z·(n - 1) + 3 \over  z·( z + 1)·(1 -  z)}\widehat{P}_n=0, \\& \\
\Lambda_{+--+}) & \partial_z^2\widehat P_n+{3· z^2 - 4 \over  z·( z + 1)·( z - 1)}\partial_z\widehat P_n+{n·(n + 2)\over ( z + 1)·(1 -  z)}\widehat{P}_n=0, \\& \\
\Lambda_{+---}) & \partial_z^2\widehat P_n+{3· z^2 - 2\over  z·( z
+ 1)·(1 -
 z)}\partial_z\widehat P_n+{n·(4 - n)\over ( z + 1)·( z - 1)}P_n=0.
\end{array}
$$

The polynomial $P_n$ of degree $n$ exists for
$\lambda_n\in\Lambda_b$ with $n$ even, that is,
$\Lambda_n=\{n\in\mathbb{Z}:16n^2+16n+3\}$, for $\Lambda_{++-+})$
and $\Lambda_{+--+})$. Therefore
$E=E_n=\{n\in\mathbb{Z}:-4n^2-4n\}$.

\noindent The corresponding solutions for $\Lambda_n$ are
$$
\begin{array}{lllll}
\Lambda_{+++-})\quad&
\Phi_{1,n}=\widehat{P}_{2n}\widehat{f}_n\Phi_{1,0},&\Phi_{1,0}={\sqrt[4]{(
z^2-1)^5}\over  z}&\widehat{f}_n=1,&\widehat{\Psi}_{1,0}= z-\frac1{ z},\\&&&\\
\Lambda_{+--+}) &
\widehat{\Phi}_{2,n}=\widehat{Q}_{2n}\widehat{f}_n\widehat{\Phi}_{2,0},&\widehat{\Phi}_{2,0}={
z^2\over \sqrt[4]{ z^2 -1}}&\widehat{f}_n=1&\widehat{\Psi}_{2,0}={
z^2\over \sqrt{ z^2 -1}}.
\end{array}
$$
These two solutions are equivalent to the same solution of the
original Schr\"odinger equation and correspond to the well known
supersymmetric quantum mechanics approach to this P\"oschl-Teller
potential, \cite{cokasu,cokasu2}. Furthermore, for all
$\lambda\in\Lambda_n$,
$$\mathrm{DGal}(\widehat{\mathbf{L}}_\lambda/{\mathbb{C}(x)})=\mathbb{G}^{[4]},\quad \mathrm{DGal}(\widehat{{L}}_\lambda/{\widehat{K}})=
\mathrm{DGal}(L_\lambda/{K})=e,$$
$$\dim_{\mathbb{C}}\mathcal E(\widehat{\mathbf{H}}-\lambda)=2,\quad \dim_{\mathbb{C}}\mathcal E(\widehat{H}-\lambda)=\dim_{\mathbb{C}} \mathcal E(H-\lambda)=4.$$
%\medskip
%\vskip 2cm

\subsection{ Searching for Potentials From Parameterized Differential Equations}

In order to search for new potentials using
$\widehat{\partial}_z$, our main object will be the family of differential equations
presented by Darboux in \cite{da1}, which can be written in the
form
\begin{equation}\label{paramet1}
\partial_z^2\widehat{y}+ \widehat{P}\partial_z\widehat{y}+(\widehat{Q} -\lambda \widehat{R})\widehat{y} = 0,\quad \widehat{P},\widehat{Q},\widehat{R}\in \widehat{K}.\end{equation}
 We recall that some Riemann  differential equations, presented
 previously, correspond to this kind.\\

 When we have a
differential equation in the form \eqref{paramet1}, we reduce it
to put it in the form of the reduced algebrized Schr\"odinger
equation $\widehat{\mathbf{H}}\Phi=\lambda\Phi$, checking that
$\mathrm{Card}(\Lambda)>1$. Thus, starting with the potential
$\widehat{\mathbf{V}}$ and arriving to the potential $V$, we obtain
the Schr\"odinger equation $H\Psi=\lambda\Psi$. This heuristic methodology
 is detailed below.\\

\begin{enumerate}
\item Reduce a differential equation of the form \eqref{paramet1} and put it in the form $\widehat{\mathbf{H}}\Phi=\lambda\Phi$, checking that
$\mathrm{Card}(\Lambda)>1$;  to avoid triviality, $\alpha$ should
be a non-constant function.

\item Write $\mathcal{W}={1\over 4}\partial_z(\ln\alpha)$ and
obtain
$\widehat{V}(z)=\alpha(\widehat{\mathbf{V}}-\partial_z\mathcal{W}-\mathcal{W}^2)$.

\item Solve the differential equation $(\partial_xz)^2=\alpha$ 	and write
$z=z(x)$, $V(x)=\widehat{V}(z(x))$.
\end{enumerate}

\noindent To illustrate this method, we present the following examples.\\

\textbf{Bessel Potentials}\\
\begin{itemize}\item (From Darboux transformations over $V=0$) In
the differential equation $$\partial_z^2\Phi=\left({n(n+1)\over
z^2}+\mu\right)\Phi,\quad \mu\in\mathbb{C},$$ we see that
$\lambda=-n(n+1)$ and $\alpha=z^2$. Applying the method, we obtain
$\widehat{\mathbf{V}}=\mu$ and $\widehat{V}(z)=\mu
z^2+\frac14$ with $z=z(x)=e^{\pm x}$. Thus, we have obtained the
potentials $V(x)=\widehat{V}(z(x))=\mu e^{\pm2x}+\frac14$ (compare
with \cite[\S 6.9]{gapa}).

\item (From Bessel differential equation) The equation
$$\partial_z^2y+{1\over z}\partial_zy+{z^2-n^2\over
z^2}y=0,\quad n\in \frac12+\mathbb{Z},$$ is transformed into the
reduced equation
$$\partial_z^2\Phi=\left(\frac{n^2}{z^2} - {4z^2 + 1 \over 4z^2}\right)\Psi.$$ We can see that $\lambda=-n^2$, $\alpha=z^2$, obtaining
$\widehat{\mathbf{V}}=-z^2-\frac14$,
$\widehat{V}=-z^4-\frac14z^2+\frac14$ and $z=z(x)=e^{\pm x}$.
Thus, we have obtained the potential
$V(x)=\widehat{V}(z(x))=-e^{\pm4x}-\frac14e^{\pm2x}+\frac14$
(compare with \cite[\S 6.9]{gapa}).
\end{itemize}

We remark that the previous examples give us potentials related
with the Morse potential, due to their solutions are given in term
of Bessel functions.\\

We can apply this method to equations such as Whittaker,
Hypergeometric and in particular, differential equations involving
orthogonal polynomials (compare with \cite[\S 5]{cokasu}).\\

\begin{comment}
\textbf{Whittaker Potentials}\\

Now we consider the differential equation
$$\partial_z^2\Phi=\left(\frac14-{\kappa\over z}+{4\mu^2-1\over
4z^2}\right)\Phi,\quad \pm \kappa\pm\mu\in\frac12+\mathbb{Z}.$$ We
can see that we can play with $\kappa$ and $\mu$.
\begin{itemize}\item Being $\kappa=\lambda$, we have $\alpha=z^2$,
 $r=\mu$ and $\lambda=n(n+1)$ we obtain
$\widehat{\mathbf{V}}(z)=\frac{z}{4} + {4\mu^2 - 1\over 4z}$,
$\widehat{V}(z)={z^3\over 4} + {4\mu^2 - 1\over 4}z + \frac14$ and
$z=z(x)={x^2\over 4}$. We choose $ w=e^{-x}$, so that we have
obtained the potential $u(x)=\mathfrak u( w (x))=\mu
e^{-2x}+\frac14$ (compare with \cite[\S 6.9]{gapa}).

\item ((Taking $\mu$ as energy)) Taking $\alpha= w^2$,
$r= w^2+\mu$ and $\lambda=n(n+1)$ we obtain $\mathfrak u( w)=
w^4+\mu w^2+\frac14$ and $ w= w(x)=e^{\pm x}$. We choose $
w=e^{-x}$, so that we have obtained the potential $u(x)=\mathfrak
u( w (x))=e^{-4x}+\mu e^{-2x}+\frac14$ (compare with \cite[\S
6.9]{gapa}).
\end{itemize}
\medskip
\textbf{One Case of Riemann Equation}\\
\end{comment}

\section*{Final Remarks}

In this paper we gave, in contemporary terms, a formalization of
original ideas and intuitions given by G. Darboux, E. Witten and
L. \'E. Gendenshte\"{\i}n in the context of the Galois theory of
linear differential equations. We found the following facts.
\begin{itemize}

\item The superpotential is an algebraic
solution of the Riccati equation associated with a potential,
defined over a differential field.

\item The Darboux transformation
was interpreted as an isogaloisian transformation, allowing to
obtain isomorphisms between their eigenrings.

\item  We introduced  the Hamiltonian algebrization
method,which in particular allows to apply algorithmic tools such
as Kovacic's algorithm to obtain the solutions, differential
Galois groups and eigenring s of second order linear differential
equations. We applied successfully this algebrization procedure to
solve problems in Supersymmetric quantum mechanics.
\item We can construct algebraically solvable and non-trivial algebraically quasi-solvable potentials  in the following ways.
\begin{enumerate}

\item Giving a potential where, for $\lambda=\lambda_0$, the
Schr\"odinger equation is integrable. After we put $\lambda\neq
\lambda_0$, we check that the Schr\"odinger equation is integrable
for more than one value of the parameter $\lambda$.

\item Giving a
superpotential to obtain the potential,  after which we check whether the
Schr\"odinger equation is integrable for more than one value of
the parameter $\lambda$.

\item Using parameterized second order linear differential
equations and applying an inverse process in the Hamiltonian
algebrization method. In particular, we can use algebraically
solvable and algebraically quasi-solvable potentials, special
functions with parameters (for example with polynomial
solutions).
\end{enumerate}
\end{itemize}
\medskip

This paper is a starting point to analyze quantum theories
through Galoisian theories. Therefore open questions and future work arise in a natural way: supersymmetric quantum mechanics with dimension greater than 2, relationship between algebraic and analytic spectrums, etc.\\

 As a conclusion, as happen in other
areas of the field of differential equations, in view of the many
families of examples studied along this paper, we can conclude
that the \emph{differential Galois theory} is a natural framework
in which some aspects of \emph{supersymmetric quantum mechanics}
may appear more clearly.

\section*{Acknowledgments}
The research of the first and second author has been partially supported by the MCyT-FEDER Grant MTM2006-00478. Part of this research was completed thanks to the first author's invitation in the XLIM institute
in Limoges. The first author is also supported by the Sergio Arboleda university.

%\section*{Appendix}
\appendix
\section{Kovacic's Algorithm}

Kovacic in 1986 (see \cite{ko}) introduced an algorithm to solve
the differential equation $\partial_x^2\zeta=r\zeta$, where $r\in
\mathbb{C}(x)$, see also \cite{acbl,dulo,fa,ulwe}. We recall this algorithm to set the notations.
\\

Each case in Kovacic's algorithm is related with each one of the
algebraic subgroups of ${\rm SL}(2,\mathbb{C})$ and the associated
Riccatti equation
$$\partial_xv=r-v^{2}=\left( \sqrt{r}-v\right)
\left(  \sqrt{r}+v\right),\quad v={\partial_x\zeta\over \zeta}.$$

There are four cases in Kovacic's algorithm. Case 1 concerns the case when the Riccati equation has a rational solution (or more); cases 2 and 3 concern cases when the Riccati equation admits an algebraic solution; case 4 is the non-integrable case.

%Only for cases 1, 2
%and 3 we can solve the differential equation, but for the case 4
%the differential equation is not integrable. It is possible that
%Kovacic's algorithm can provide us only one solution ($\zeta_1$),
%so that we can obtain the second solution ($\zeta_2$) through
%\begin{equation}\label{second}
%\zeta_2=\zeta_1\int\frac{dx}{\zeta_1^2}.
%\end{equation}

Our differential equation is given by
$\partial_x^2\zeta=r\zeta$ where  $r={s\over t}$, $s,t\in \mathbb{C}[x]$.
We use the following notations:
\begin{enumerate}
\item $\Gamma'$ is
the
set of (finite) poles of $r$:  $\Gamma^{\prime}=\left\{  c\in\mathbb{C}%
:t(c)=0\right\}$.

\item 
$\Gamma=\Gamma^{\prime}\cup\{\infty\}$.
\item The order $\circ(r_c)$ of $r$ at
$c\in \Gamma'$ is the multiplicity of $c$ as a
pole of $r$.

\item The order $\circ\left(r_{\infty}\right)$ of $r$ at $\infty$ is the order of $\infty$ as a zero of
$r$, \\ i.e  $\circ\left( r_{\infty }\right) =\mathrm{deg}(t)-\mathrm{deg}(s)$.

\end{enumerate}

\noindent\underline{\bf Case 1.} 
Let $\left[ \sqrt{r}\right] _{c}$ (resp. $\left[ \sqrt{r}\right] _{\infty}$) represent the Laurent series (if defined) of $\sqrt{r}$ at $c$ (resp. infinity).
For $p\in \Gamma$, we will define $\varepsilon\left( p\right) \in\{+,-\}$ below.
Finally, the complex numbers $\alpha_{c}^{+},\alpha_{c}^{-},\alpha_{\infty}%
^{+},\alpha_{\infty}^{-}$ will be defined in the first step. 
If the differential equation has no poles it only can fall in this case.
\medskip

{\it Step 1.} For each $c\in\Gamma$:
\begin{description}

\item[$(c_{0})$] If $\circ\left(  r_{c}\right)  =0$, then    %%happens only for c=infinity
$$\left[ \sqrt {r}\right] _{c}=0,\quad\alpha_{c}^{\pm}=0.$$

\item[$(c_{1})$] If $\circ\left(  r_{c}\right)  =1$, then
$$\left[ \sqrt {r}\right] _{c}=0,\quad\alpha_{c}^{\pm}=1.$$

\item[$(c_{2})$] If $\circ\left(  r_{c}\right)  =2,$ and 
$$r=  b(x-c)^{-2}+\cdots,\quad \textrm{then}$$
$$\left[ \sqrt {r}\right]_{c}=0,\quad \alpha_{c}^{\pm}=\frac{1\pm\sqrt{1+4b}}{2}.$$

\item[$(c_{3})$] If $\circ\left(  r_{c}\right)  =2v\geq4$, and $$r=
(a\left( x-c\right)  ^{-v}+...+d\left( x-c\right)
^{-2})^{2}+b(x-c)^{-(v+1)}+\cdots,\quad \textrm{then}$$ $$\left[
\sqrt {r}\right] _{c}=a\left( x-c\right) ^{-v}+...+d\left(
x-c\right) ^{-2},\quad\alpha_{c}^{\pm}=\frac{1}{2}\left(
\pm\frac{b}{a}+v\right).$$

\item[$(\infty_{1})$] If $\circ\left(  r_{\infty}\right)  >2$, then
$$\left[\sqrt{r}\right]  _{\infty}=0,\quad\alpha_{\infty}^{+}=0,\quad\alpha_{\infty}^{-}=1.$$

\item[$(\infty_{2})$] If $\circ\left(  r_{\infty}\right)  =2,$ and
$r= \cdots + bx^{2}+\cdots$, then $$\left[
\sqrt{r}\right]  _{\infty}=0,\quad\alpha_{\infty}^{\pm}=\frac{1\pm\sqrt{1+4b}%
}{2}.$$

\item[$(\infty_{3})$] If $\circ\left(  r_{\infty}\right) =-2v\leq0$,
and
$$r=\left( ax^{v}+...+d\right)  ^{2}+ bx^{v-1}+\cdots,\quad \textrm{then}$$
$$\left[  \sqrt{r}\right]  _{\infty}=ax^{v}+...+d,\quad
and\quad \alpha_{\infty}^{\pm }=\frac{1}{2}\left(
\pm\frac{b}{a}-v\right).$$
\end{description}
\medskip

{\it Step 2.} Find $D\neq\emptyset$ defined by
$$D=\left\{
n\in\mathbb{Z}_{+}:n=\alpha_{\infty}^{\varepsilon
(\infty)}-%
{\displaystyle\sum\limits_{c\in\Gamma^{\prime}}}
\alpha_{c}^{\varepsilon(c)},\forall\left(  \varepsilon\left(
p\right) \right)  _{p\in\Gamma}\right\}  .$$ If $D=\emptyset$,
then we should start with the case 2. Now, if
$\mathrm{Card}(D)>0$, then for each $n\in D$ we search $\omega$
$\in\mathbb{C}(x)$ such that
$$\omega=\varepsilon\left(
\infty\right)  \left[  \sqrt{r}\right]  _{\infty}+%
{\displaystyle\sum\limits_{c\in\Gamma^{\prime}}}
\left(  \varepsilon\left(  c\right)  \left[  \sqrt{r}\right]  _{c}%
+{\alpha_{c}^{\varepsilon(c)}}{(x-c)^{-1}}\right).$$
\medskip

{\it Step 3}. For each $n\in D$, search for a monic polynomial
$P_n$ of degree $n$ with
\begin{equation}\label{recu1}
\partial_x^2P_n + 2\omega \partial_xP_n + (\partial_x\omega + \omega^2 - r) P_n = 0.
\end{equation}
If success is achieved then $\zeta_1=P_n e^{\int\omega}$ is a
solution of the differential equation.  Else, case 1 cannot hold.
\bigskip

\noindent\underline{\bf Case 2.}  

{\it Step 1.} For
each $c\in\Gamma^{\prime}$ and for $\infty$, we define the sets
$E_{c}\subset\mathbb{Z}$ and $E_{\infty}\subset\mathbb{Z}$ as
follows:
\medskip

\begin{description}
\item[($c_1$)] If $\circ\left(  r_{c}\right)=1$, then $E_{c}=\{4\}$

\item[($c_2$)] If $\circ\left(  r_{c}\right)  =2,$ and $r= \cdots +
b(x-c)^{-2}+\cdots ,\ $ then $$E_{c}=\left\{
2+k\sqrt{1+4b}:k=0,\pm2\right\}.$$

\item[($c_3$)] If $\circ\left(  r_{c}\right)  =v>2$, then $E_{c}=\{v\}$

\item[$(\infty_{1})$] If $\circ\left(  r_{\infty}\right)  >2$, then
$E_{\infty }=\{0,2,4\}$

\item[$(\infty_{2})$] If $\circ\left(  r_{\infty}\right)  =2,$ and
$r= \cdots + bx^{2}+\cdots$, then $$E_{\infty }=\left\{
2+k\sqrt{1+4b}:k=0,\pm2\right\}.$$

\item[$(\infty_{3})$] If $\circ\left(  r_{\infty}\right)  =v<2$,
then $E_{\infty }=\{v\}$
\medskip
\end{description}
If some $E_c=\emptyset$ then switch to next case.

{\it Step 2.} Compute $D$ defined by
$$D=\left\{
n\in\mathbb{Z}_{+}:\quad n=\frac{1}{2}\left(  e_{\infty}-
{\displaystyle\sum\limits_{c\in\Gamma^{\prime}}} e_{c}\right)
,\forall e_{p}\in E_{p},\quad p\in\Gamma\right\}.$$ 
If $D=\emptyset,$ then switch to case 3. 
If $\mathrm{Card}(D)>0,$ then for each $n\in D$, compute the corresponding rational
function $\theta$ defined by
$$\theta=\frac{1}{2}
{\displaystyle\sum\limits_{c\in\Gamma^{\prime}}}
\frac{e_{c}}{x-c}.$$
\medskip

{\it Step 3.} For each $n\in D$, search (via linear algebra) for a monic polynomial $P_n$
of degree $n$, such that {\small{
\begin{equation}\label{recu2}
\partial_x^3P_n+3\theta
\partial_x^2P_n+(3\partial_x\theta+3\theta
^{2}-4r)\partial_xP_n+\left(
\partial_x2\theta+3\theta\partial_x\theta
+\theta^{3}-4r\theta-2\partial_xr\right)P_n=0.
\end{equation}}}
 If $P_n$ does not
exist, then case 2 does not hold. If such a polynomial $P_n$ is found, set
$\phi = \theta + \partial_xP_n/P_n$ and let $\omega$ be a solution
of
$$\omega^2 + \phi \omega + {1\over2}\left(\partial_x\phi + \phi^2 -2r\right)=
0.$$

Then $\zeta_1 = e^{\int\omega}$ is a solution of the differential
equation.
\bigskip

\noindent\underline{\bf Case 3.} 

{\it Step 1.}  For
each $c\in\Gamma^{\prime}$ and for $\infty$ we define the sets
$E_{c}\subset\mathbb{Z}$ and $E_{\infty}\subset\mathbb{Z}$ as
follows:
\medskip

\begin{description}

\item[$(c_{1})$] If $\circ\left(  r_{c}\right)  =1$, then
$E_{c}=\{12\}$

\item[$(c_{2})$] If $\circ\left(  r_{c}\right)  =2,$ and $r= b(x-c)^{-2}+\cdots$, then
\begin{displaymath}
E_{c}=\left\{ 6+k\sqrt{1+4b}:\quad
k=0,\pm1,\pm2,\pm3,\pm4,\pm5,\pm6\right\}.
\end{displaymath}

\item[$(\infty)$] If $\circ\left(  r_{\infty}\right)  =v\geq2,$ and $r=
\cdots + bx^{2}+\cdots$, then {\small $$E_{\infty }=\left\{
6+{12k\over m}\sqrt{1+4b}:\textrm{ }
k=0,\pm1,\pm2,\pm3,\pm4,\pm5,\pm6\right\},\textrm{ }
m\in\{4,6,12\}.$$}
\medskip
\end{description}

{\it Step 2.} For $m\in\{4,6,12\}$ successively, compute the set $D$ defined by
$$D=\left\{
n\in\mathbb{Z}_{+}:\quad n=\frac{m}{12}\left(
e_{\infty}-{\displaystyle\sum\limits_{c\in\Gamma^{\prime}}}
e_{c}\right)  ,\forall e_{p}\in E_{p},\quad p\in\Gamma\right\}.$$
If $\mathrm{Card}(D)>0,$ then for each $n\in D$ with
its respective $m$,  compute the  rational function
$$\theta={m\over 12}{\displaystyle\sum\limits_{c\in\Gamma^{\prime}}}
\frac{e_{c}}{x-c}$$ and the (denominator) polynomial $S$ defined as $S=
{\displaystyle\prod\limits_{c\in\Gamma^{\prime}}} (x-c)$.

{\it Step 3}. For each $n\in D$, with its respective $m$, compute a
monic polynomial $P_n=P$ of degree $n,$ such that its coefficients
can be determined recursively by
$$\bigskip P_{-1}=0,\quad P_{m}=-P,$$
$$P_{i-1}=-S\partial_xP_{i}-\left( \left( m-i\right)
\partial_xS-S\theta\right)  P_{i}-\left( m-i\right)  \left(
i+1\right)  S^{2}rP_{i+1},$$ where $i\in\{0,1\ldots,m-1,m\}.$ If
$P$ does not exist, then the differential equation is not
integrable because it falls in Case 4. 
\\
If $P$ exists,  a solution
of the differential equation is given by $\zeta=e^{\int\omega}$, where 
$\omega$ is a zero of the polynomial of degree $m$ 
$$ {\displaystyle\sum\limits_{i=0}^{m}}
\frac{S^{i}P_i}{\left( m-i\right)  !}\omega^{i}=0$$

\bigskip
\noindent\underline{\bf Case 4.} In none of the above holds, then the equation is non-integrable and its Galois group is $SL(2)$. 

\begin{remark}[\cite{acbl}]\label{rkov2}
If the differential equation falls only in the case 1 of Kovacic's
algorithm, then its differential Galois group is given by one of
the following groups:

\begin{description}

\item[{\it I1}] $e$ when the algorithm provides two
rational solutions.

\item[{\it I2}] $\mathbb{G}^{[n]}$ when the algorithm provides two algebraic solutions $\zeta_1,\zeta_2$ such that $\zeta_1^n,\zeta_2^n\in\mathbb{C}(x)$
and $\zeta_1^{n-1},\zeta_2^{n-1}\notin\mathbb{C}(x)$.

\item[{\it I3}] $\mathbb{G}^{\{n\}}$ when the algorithm provides only one
algebraic solution $\zeta$ such that $\zeta^n\in\mathbb{C}(x)$
with $n$ minimal.

\item[{\it I4}] $\mathbb{G}_m$ when the algorithm provides two
non-algebraic solutions.

\item[{\it I5}] $\mathbb{G}_a$ when the algorithm provides one
rational solution and the second solution is not algebraic.

\item[{\it I6}] $\mathbb{B}$ when the algorithm only provides one
solution $\zeta$ such that $\zeta$ and its square are not rational
functions.

%\item[I7] $\mathrm{SL}(2,\mathbb{C})$, the special linear group, if the algorithm does not provide any
%solution.
\end{description}
\end{remark}

\bibliographystyle{amsalpha}

\end{document}